\documentclass[reprint, nofootinbib, amsmath, amssymb, twocolumn]{revtex4-1}
\usepackage{amsmath}
\usepackage{amssymb}
\usepackage{physics}
\usepackage{dsfont}
\usepackage{hhline}
\usepackage{enumerate}
\usepackage{subcaption}
\usepackage{graphicx}
\usepackage{makecell}
\usepackage{bold-extra}
\usepackage{rotating} 
\usepackage[utf8]{inputenc}
\usepackage{amsthm}
\newtheorem{theorem}{Theorem}[section]
\newtheorem{corollary}[theorem]{Corollary}
\newtheorem{lemma}[theorem]{Lemma}
\newtheorem{definition}[theorem]{Definition}
\usepackage{natbib}
\usepackage{xcolor}
\usepackage{multirow}
\usepackage{hyperref}
\usepackage{float}
\usepackage{chngcntr}
\counterwithin{paragraph}{section}

\newcommand{\dsone}{\mathds{1}}
\newcommand{\PREP}{\textrm{PREP}}
\newcommand{\SORT}{\textrm{SORT}}
\newcommand{\BE}{\textrm{BE}}
\newcommand{\SEL}{\textrm{SEL}}

\usepackage[bottom]{footmisc}

\begin{document}
\title{Optimized quantum algorithms for simulating the Schwinger effect}
\author{Angus Kan}
\author{Jessica Lemieux}\thanks{jlemieux@psiquantum.com}
\author{Olga Okrut}
\author{Burak \c{S}ahino\u{g}lu}
\affiliation{PsiQuantum, 700 Hansen Way, Palo Alto, California 94304, USA}

\date{\today}

\begin{abstract}
The Schwinger model, which describes lattice quantum electrodynamics in $1+1$ space-time dimensions, provides a valuable framework to investigate fundamental aspects of quantum field theory, and a stepping stone towards non-Abelian gauge theories. Specifically, it enables the study of physically relevant dynamical processes, such as the nonperturbative particle-antiparticle pair production, known as the Schwinger effect.
In this work, we analyze the quantum computational resource requirements associated with simulating the Schwinger effect under two distinct scenarios: (1) a quench process, where the initial state is a simple product state of a non-interacting theory and then interactions are turned on at time $t=0$, and (2) a splitting (or scattering) process where two Gaussian states, peaked at given initial momenta, are shot away from (or towards) each other.
We explore different physical regimes in which the Schwinger effect is expected to be observable. These regimes are characterized by initial momenta and coupling strengths, as well as simulation parameters such as lattice size and electric-field cutoffs. 
Leveraging known rigorous bounds for electric-field cutoffs, we find that a reliable simulation of the Schwinger effect is provably possible at high cutoff scales.
Furthermore, we provide optimized circuit implementations of both the second-order Trotter formula and an interaction-picture algorithm based on the Dyson series to implement the time evolution. Our detailed resource estimates show the regimes in which the interaction-picture approach outperforms the Trotter approach, and vice versa.
The improved theoretical error bounds, optimized quantum circuit designs, and explicitly compiled subroutines developed in this study are broadly applicable to simulations of other lattice models in high-energy physics and beyond.

\end{abstract}

\maketitle

\section{Introduction}

In Schwinger's seminal work on (1+1)-dimensional quantum electrodynamics (QED)~\cite{schwinger1962gauge}, known as the Schwinger model, it was shown that the vector boson could be massive, which is dubbed dynamical mass generation.
Over the years, this model has been studied extensively with classical methods, both analytically~\cite{coleman1975charge, coleman1976more, banks1976strong, hamer1982massive, iso1990hamiltonian} and numerically~\cite{crewther1980eigenvalues, grady1987numerical, hamer1997series, kroger1998massive, hebenstreit2013simulating, banuls2013mass, banuls2016chiral, buyens2016confinement, buyens2017real, papaefstathiou2021density, dempsey2022discrete}. These studies have shown that it shares various properties with quantum chromodynamics (QCD), e.g. confinement and instantons.

Yet, non-perturbative behaviors, particularly dynamical behaviors, of the Schwinger model are hard to study using classical simulations, due to, e.g., the sign problem~\cite{PhysRevLett.94.170201,doi:10.1142/S0217751X16430077}.
It is widely expected that quantum simulation will unlock dynamical studies of the Schwinger model, as evidenced by the numerous proposals for quantum simulation of this model~\cite{hauke2013quantum, martinez2016real, muschik2017u, zache2018quantum, kharzeev2020real, shaw2020quantum, kan2022simulating, sakamoto2023endtoendcomplexitysimulatingschwinger, farrell2024quantum,kan2022lattice,rhodes2024exponential}; some of them involve small-scale, proof-of-concept simulations on error-prone, near-term quantum hardware, while a recent handful, such as Refs.~\cite{shaw2020quantum, kan2022simulating, sakamoto2023endtoendcomplexitysimulatingschwinger, kan2022lattice,rhodes2024exponential}, 
focus on algorithms with provable performance that could be implemented on a universal fault-tolerant quantum computer.
Similar to these recent works, we devise algorithms to study the real-time dynamics of the Schwinger model and develop optimized fault-tolerant circuit implementations.

Quantum algorithms for real-time quantum dynamics have been studied since the late 90's starting with the seminal work by Lloyd~\cite{lloyd1996universal}, which substantiated the proposal by Feynman~\cite{feynman1982simulating}.
This first method was based on a specific product formula called the Suzuki-Trotter formula~\cite{suzuki1990fractal}; product-formula algorithms have since been improved~\cite{berry2007efficient, childs2021theory, morales2022greatly}. Alternative to this approach are methods based on various polynomial series approximations~\cite{childs2012hamiltonian, berry2015hamiltonian, berry2015simulating, low2017optimal, low2019hamiltonian, low2018hamiltonian, kieferova2019simulating}. Each method has its own advantages and potential drawbacks.

While the best (time-independent) Hamiltonian simulation method achieves an additive gate complexity of $\mathcal{O}(\log(1/\epsilon))$~\cite{low2017optimal}, quantum simulations based on a $k$th order product formula have a multiplicative gate complexity scaling as $\mathcal{O}\left((1/\epsilon)^{2k}\right)$~\cite{berry2007efficient}. 
With respect to simulation time $t$, both methods can, in principle, achieve linear scaling in complexity. For product formulas, however, achieving this requires a large $k$, which increases the complexity of the corresponding quantum circuit. 
This, in turn, motivates the use of a higher-order (rather than lower-order) product formula for large  $t$ or $1/\epsilon$.
Other motivations include the fact that product formulas lend themselves more readily to additional physics-informed improvements, such as low-energy subspace~\cite{csahinouglu2021hamiltonian} or symmetry-based~\cite{tran2021faster} methods, and that error bounds based on the commutators of the Hamiltonian terms can be significantly lower than the worst-case analytical bounds~\cite{childs2019nearly, chen2024average}.

From the family of series-approximation-based algorithms, another approach that exploits the physics of the problem is based on the Dyson series~\cite{low2018hamiltonian}.
In particular, when a Hamiltonian is divided into a sum of two terms, such as $H= H_0 + V$ where $\|V\| \ll \|H_0\|$, $H_0$ is fast-forwardable~\cite{atia2017fast, gu2021fast}, and the quantum simulation is performed in the interaction picture of $H_0$, it has gate complexity $\tilde{\mathcal{O}}(\|V t\|)$ rather than $\tilde{\mathcal{O}}(\|H t\|)$ (where $\tilde{\mathcal{O}}$ is the complexity up to hidden logarithmic factors). In this work, we apply this interaction-picture based quantum algorithm and the second-order Suzuki-Trotter formula to analyze the complexity of simulating the Schwinger effect.

Choosing the more gate-efficient algorithm for a given, finite problem instance, however, is hard because in addition to analyzing the asymptotic scaling, it also requires a detailed analysis of constant factors, which is typically accomplished by explicit compilations of the candidate algorithms.
Ref.~\cite{shaw2020quantum} argued employing product formula rather than qubitization is likely to be advantageous due to the linear instead of quadratic scaling in terms of the electric field cutoff. Based on our analysis, we expect the electric field term to be much stronger than other terms when the Schwinger effect is to be observed. Therefore, we need a high electric field cutoff to accurately simulate the phenomenon~\cite{tong2021Provably}\footnote{Note that there are previous~\cite{jordan2012quantum} and more recent~\cite{ciavarella2025truncation} works that argue for lower field cutoffs.
These works rely on energy based truncation and Hilbert space fragmentation arguments, respectively.
While these heuristics may turn out to be approximately correct for certain observables, we base our assumption for the field cutoff on the rigorous bounds proved in Ref.~\cite{tong2021Provably}.}. 
Employing the interaction-picture based algorithm, we can asymptotically achieve an exponential improvement, i.e., a complexity scaling $O((\log \Lambda)^2)$, in terms of the electric field cutoff $\Lambda$.~\footnote{We note that in Ref.~\cite{rhodes2024exponential}, implementations of the interaction-picture based algorithms from~\cite{tong2021Provably} applied to U(1), SU(2) and SU(3) lattice gauge theories, beyond one spatial dimension, were reported. Some of the techniques we develop in this work can be used to improve the implementations from Ref.~\cite{rhodes2024exponential}, which we discuss in Section~\ref{sec:Conclusion}.}

Our main goal in this work is to detail this study; we provide optimized compilations of product-formula and interaction-picture algorithms for various instances over a range of parameters, and compare their performances in terms of T-gate and qubit counts. 
The considered range of parameters are determined such that the Schwinger effect, i.e., pair-production, can be reliably simulated. More specifically, if one simulates various instances of the Schwinger model over the range of parameters, an accurate simulation, in which the Schwinger effect are expected to be observed, is guaranteed. Furthermore, we optimize the cost of the interaction-picture algorithm mainly via improved and novel bounds on the discretization error of the Dyson series and improved compilation techniques, e.g., phase-gradient addition. 

The manuscript is organized as follows.
In Section~\ref{sec:Setting}, we describe the Hamiltonian and the experiments to be simulated. We also define the conditions for observing the Schwinger effect arising from lattice approximations and accuracy considerations.
Then, in Sections~\ref{sec:QSimProductFormula} and~\ref{sec:QSimIP} we summarize the quantum circuits and lay out their resource costs for product formula and interaction-picture based algorithms, respectively.
In Section~\ref{sec:Results} we give the detailed resource costs of each method.
In Section~\ref{subsec:Comparison}, in particular, we compare the asymptotic and numerical resource costs of two methods, and determine the method that has a lower resource cost for a wide range of parameter regimes.
Finally, Section~\ref{sec:Conclusion} concludes with a future outlook for improving the resource costs even further and extending the application of these methods to more complicated models in particle physics.
Appendices are reserved for deriving the final form of the lattice Hamiltonian that is used in Appendix~\ref{sec:MappingHamiltonianToQubits}, and further technical and resource cost analysis of the two methods are in Appendices~\ref{subsec:OverviewIP} and~\ref{app:trotter}.

\section{Schwinger model: The Hamiltonian, parameters, and observables of interest}\label{sec:Setting}
The Hamiltonian for the Schwinger model is given by~\cite{banks1976strong}
\begin{align}
\nonumber H_{\textrm{nat}}= \frac{g}{2\sqrt{x}} \sum^{N}_{r=1} \Big[& (E_r + \alpha)^2 + (-1)^r \frac{2\sqrt{x}m}{g} \psi^\dagger_r \psi_{r} \\
&  + x (\psi^\dagger_r U_r \psi_{r+1} + \textrm{h.c.} ) \Big],
\end{align}
in natural units ($\hbar= c= 1$).
There are a few user specified parameters in this Hamiltonian:
$g$, $m$, and $\alpha$ are respectively the bare coupling constant, bare mass, and the background electric field strength, independent of the specifications of the simulated lattice; $a$ and $N$ are the lattice spacing and the number of lattice sites, specified by the simulation lattice. 
We define the unitless quantities
\begin{align}\label{eq:ModelParametersFromPhysicalParamaters}
x= \frac{1}{g^2 a^2} \quad \textrm{and} \quad \mu= \frac{2\sqrt{x} m}{g},
\end{align}
which determine the strength of the mass and the kinetic term with respect to the electric field term. 
$H_{\textrm{nat}}$ acts on fermionic (electron/positron) and spin (electric field) degrees of freedom placed on the sites and links of the graph that defines the lattice.
To be more precise, $E_r$ and $U_r$ are operators that act on the electric field degrees of freedom placed on the link labeled by index $r$ that is placed right to the vertex labeled by $r$, and $\psi_r$ and $\psi^\dagger_r$ are annihilation and creation operators on the fermion degree of freedom on the lattice site labeled by $r$.
Below, we give definitions of these operators, after a slight massaging of the Hamiltonian.
As is standard in numerical simulations, we render the Hamiltonian $H$ and simulation time $t$ both unitless, i.e., 
\begin{align}\label{eq:SimulationHamiltonian}
\nonumber H= \sum^{N}_{r=1} \Big[&(E_r + \alpha)^2 + (-1)^r \mu \psi^\dagger_r \psi_{r}\\
&+ x (\psi^\dagger_r U_r \psi_{r+1} + \textrm{h.c.} )\Big],
\end{align}
where going back to quantities with units is straightforward, i.e., $E^{\textrm{nat}}= \sqrt{\frac{g}{2\sqrt{x}}} E$ and $t^{\textrm{nat}}= \frac{2\sqrt{x}}{g} t$.
For the sake of completeness, and for future sanity checks, we provide the units of scalars and operators:
\begin{align}
[g]= [m]= [a]^{-1}= [E^\textrm{nat}]^{2}= [\psi^\textrm{nat}]^2.
\end{align}
Note that, ultimately, the continuum limit needs to be taken, i.e., the limit $a \rightarrow 0$ or equivalently $N \rightarrow \infty$, which is done via extrapolation.
As $a$ is swept towards smaller values, $x$ takes higher values hence the model is driven towards the strong-coupling limit, and the physical result is obtained by extrapolating to $a \rightarrow 0$.
More generally, the parameters, $N$, $x$, $a$, $m/g$, and the cutoff $\Lambda$ of the electric field together with the constant external electric field $\alpha$, are parameters that need to be chosen such that the simulation provides reliable results for the desired phenomena. 
During this process, the model is simulated under different parameter settings, and the more resource-efficient simulation algorithm may change depending on the parameter regime.

Note that the lattice spacing $a$ does not appear directly in the Hamiltonian, yet appears in the definition of $x$. In fact, we can redefine the bare coupling constant $g$ and the bare mass $m$ such that they are rescaled by $a$, i.e., $\tilde{g}= ga$, $\tilde{m}= ma$. In this way, all quantities, including $\tilde{g}$ and $\tilde{m}$ are dimensionless, in addition to $x$. 
The Hamiltonian stays the same (for the sake of future convenience $\tilde{g}$ and $\tilde{m}$ are shown as $g$ and $m$ below) but with $x= 1/g^2$. 
Hence, for a set of end-to-end simulations, the parameters we need to specify are: $N, x, m/g, \alpha, \Lambda$, the total evolution time $t$, and the parameters used to prepare the initial state (such as particle momenta, position, etc.; see Secs.~\ref{subsec:cond_schwinger_effet},~\ref{subsec:cond_truncation} for details). 

Following the naming convention made in Ref.~\cite{shaw2020quantum}, we partition the Hamiltonian into three terms as
\begin{align}
H= H_{E} +  H_{\textrm{M}} + H_{I}
\end{align}
where
\begin{align}
H_E&= \sum^N_{r=1} (E_r + \alpha)^2 \\
H_{M} & = \mu \sum_r (-1)^r \psi^\dagger_r \psi_{r}\\
H_I & = x \sum_{r} (U_r \psi^\dagger_r \psi_{r+1} - U^\dagger_r\psi_r\psi_{r+1}^\dagger)
\end{align}
are what we call the electric field, mass, and the interaction (or kinetic energy) term, respectively.
A canonical choice of basis in the literature is as follows:
\begin{align}
E_r &= \sum^{\Lambda-1}_{\varepsilon_r=-\Lambda} \varepsilon_r \ketbra{\varepsilon_r}_r,\\
U_r& = \sum^{\Lambda-1}_{\varepsilon_r= -\Lambda}|\varepsilon+1  \mod [-\Lambda, \Lambda-1] \rangle \langle \varepsilon|_r,\\
\psi^\dagger_r & = \frac{X_r-iY_r}{2} \prod_{j=1}^{r-1} Z_j,
\end{align}
where a Jordan-Wigner transformation~\cite{jw} is assumed to transform the fermionic operators. 
See Appendix~\ref{sec:MappingHamiltonianToQubits} for details.

In what follows, we first determine bounds on $\alpha, \Lambda, N$ such that the Schwinger effect is expected to be observable. 
Then, we set up two experiments: a quench experiment and a scattering experiment, both of which aim to probe the dynamics of pair production. 
Next, we determine how to choose the lattice parameters so that spurious lattice effects are manageable.
Finally, we obtain the constraints over the remaining simulation parameters such that the simulation is expected to describe the continuum accurately.

\subsection{Conditions for observing the Schwinger effect}\label{subsec:cond_schwinger_effet}

The Schwinger effect refers to the phenomenon of electron-positron pair creation. For this to be observable, there needs to be enough energy exchange between the electric field term and the mass term, i.e., the energy that can be stored in the electric field should be high enough to create many particles when that energy is transferred to mass. This corresponds to the following condition:
\begin{align}
(\Lambda + \alpha)^2 - \alpha^2 \gg \mu ,
\end{align}
where the left hand side is the maximum electric field energy that can be stored in one link on the lattice, whereas the right hand side is the minimum energy of an electron-positron pair that are only one link apart from each other on the lattice (hence they are at the same position up to a lattice spacing $a$). 
This is equivalent to
\begin{align}
\Lambda(\Lambda + 2\alpha) \gg \mu.
\end{align}
For our experiments, we can carry the effect of nonzero $\alpha$ to the electric field values in the initial state.
Then, this is possible if $\Lambda$ is chosen such that
\begin{align}
\Lambda \gg \sqrt{\mu}.
\end{align}
Hence, we sweep through the parameter $\Lambda$ in the interval given as
\begin{align}\label{eq:SweepingLambda}
[100 \mu & \geq \Lambda^2 \geq 0.01\mu].
\end{align}
This determines the range of $\Lambda$ that we are interested in for the simulation of the Schwinger model.
We remark that smaller values ($< \sqrt{\mu}$) of $\Lambda$ may still give important information in order to determine the value of $\Lambda$ at which the Schwinger effect starts to emerge.

We propose two experiments, where we probe the number of electron-positron pairs as a function of time.
In both of these experiments the final observable and the time-evolution are the same, and they differ only in terms of the initial states.
In the first experiment, we start from an initial state that contains no particles but nonzero electric field values.
In the second experiment, we imagine a dressed electron-positron pair ($e^- - e^+$ pair) is created with particles distance $l$ apart, and initiated as moving away from each other with momentum $p_1$ and $p_2$, respectively (as depicted in Fig.~\ref{fig:experiment-2}).
In both cases, we let the time evolve with the full Hamiltonian and monitor the dynamics of the particle density.

\begin{figure}
    \centering
     \includegraphics[scale=0.45]{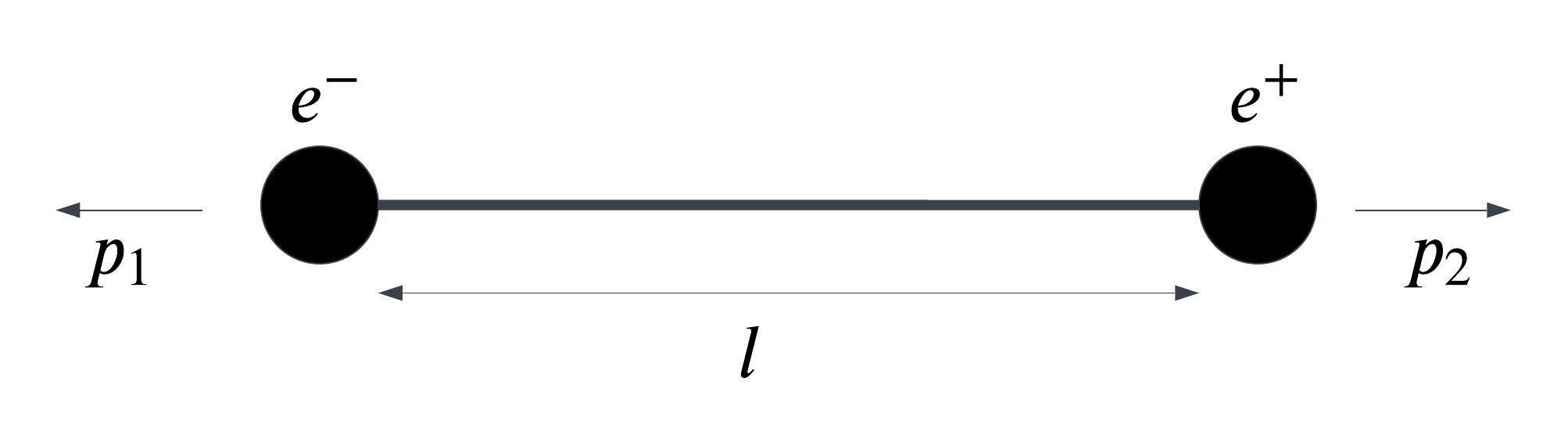}
    \caption{A schematic representation of the initial state for the second experiment: an electron-positron pair ($e^- - e^+$) is created that are $l$ apart, and moving away from each other with momentum $p_1$ and $p_2$, respectively.}
    \label{fig:experiment-2}
\end{figure}

\noindent \textbf{The first experiment:} In the first experiment, we start with the following initial state:
\begin{align}
\nonumber \ket{\psi(t=0)}= &\ket{0}_0  \ket{\gamma}_{0,1} \ket{1}_1 \ket{\gamma}_{1,2} \ket{0}_{2} \ket{\gamma}_{2,3} \ket{1}_{3}\\
&\ket{\gamma}_{3,4} \ket{0}_{4} \ldots \ket{1}_{N-2} \ket{\gamma}_{N-2,N-1} \ket{0}_{N-1}
\end{align}
where $\gamma$ is a fixed electric field value that is uniform in the whole system, and a subscript $(\cdot)_r$ denotes the site $r$ while a subscript $(\cdot)_{r,r+1}$ denotes the link that connects the sites $r$ and $r+1$. By definition, $|\gamma| < \Lambda$. Note that from now on, we will also assume that the background field $\alpha = 0$.
The value of $\gamma$ can be connected to the maximum particle density that is kinematically possible, which can be at most $1$. 

We wish to be able to have an initial state energy that can capture a pair production of about average (over whole system) particle density, say $\rho < 1$.
This would imply a choice of $\gamma$ such that $\gamma^2 \geq \rho \mu$.
Furthermore, by definition, $\gamma$ can be at most as big as the cutoff $\Lambda$, hence we obtain
\begin{align}
\rho \dfrac{\sqrt{x}m}{g}\leq \gamma^2 < \Lambda^2.
\end{align}
Note that this is only a rough estimate; we still sweep through the values of $x$ and $\frac{m}{g}$ as independent parameters.

Then, we evolve the system with the full Hamiltonian for time $t$ by applying the operator $e^{-iHt}$ to the initial state $\ket{\psi(t=0)}$. 
The minimum time we want to carry out this time evolution is $t_{\min}= \rho / x$, which is the average minimum time (i.e., fastest possible) that $N \rho$ many ($e^{-} - e^{+}$) pairs can be generated.
This is obtained by considering the time evolution with the term $H_I$.
In the small time limit and when this term is perturbative it can be argued, similar to Fermi's golden rule, that the rate at which particles are generated is proportional to the transition amplitude of $H_I$ between initial and intermediate states which generate a particle-antiparticle pair while decreasing the electric field. 
This overlap is given by the strength of the Hamiltonian term, $x$.
Ideally, we observe the system for times $t \gg t_{\min}$, and settle with a maximum evolution time $t_{\max}$ when the observable profile seemingly starts to repeat itself in time. 
Note that this quantity, by itself, is an interesting observable and worth measuring as a result of the quantum simulation.
To be more precise, we run the time-evolution for times
\begin{align}
\frac{\rho}{x}= t_{\min}< t < t_{\max}
\end{align}
where $t_{\max}$ is determined empirically as a result of the experiment.
Furthermore, we need to choose the system large enough, i.e., $N$ high enough, so that the boundary effects are negligible.
This can be studied rigorously using Lieb-Robinson bounds, such as given in Ref.~\cite{tong2018gauge} and Ref.~\cite{haah2021quantum} (see Lemma 5).
These results imply that for any observable (such as a defect, or any kind of operator that was not supposed to be there) on the boundary to affect considerably any other observable that is of distance $l$ from the boundary, a certain time needs to pass.
In other words, for a given time evolution $t$, we can bound a sufficient minimum system size, above which the boundary effects are only visible up to an error $\epsilon$.

Following the notation from Lemma 5 of Ref.~\cite{haah2021quantum}, where $\zeta_0$ sums over all $H_I(Z)$ in $H_I$ which acts nontrivially on lattice site $r$;
\begin{align}
    \zeta_0 = \max_r \sum_{Z: r \in Z} \|H_I(Z)\|= 4|x|,
\end{align}
 and given the extent of the initial state $N_0$, we want to solve for $l$ such that
\begin{align}
N_0 (8x|t|)^l/ l! \leq \epsilon. 
\end{align}

\noindent This results in
\begin{align}
l \geq \max\{\ln(N_0/\epsilon), e 8|x||t|\}.
\end{align}
Then for all
\begin{align}
N \geq N_0 + 2l,
\end{align}
the simulation on $N$ sites is accurate up to $\epsilon$ error.
$N_0$ is usually given as the extent of a few well-separated quasi-particles that the system is expected to support.
This can either be estimated via classical simulations or can also be calculated with various quantum simulations where $N_0$ is swept through a range of values.
The first observable we measure is the particle density given as
\begin{align}
O_{\textrm{dens}}:= \frac{1}{N} \sum^{N-1}_{r= 0} \frac{1}{2}(\dsone - (-1)^r \sigma_{Z}(r)).
\end{align}
This observable, measured at a given time, reveals the number of particles in the system, whether the time average of this observable reaches an equilibrium, and if so, the time in which it does. These observables are important to reveal the onset of \emph{thermalization} and other properties of the system.

The second observable measures the \emph{global} electric field polarization of the half-system, which is given as
\begin{align}
O_{\textrm{pol}}:= E_{\lfloor N/2\rfloor} - E_{0},
\end{align}
where the subindices $0$ and $\lfloor N/2 \rfloor$ refer to the first and middle links on the lattice.
This observable and its higher moments are interesting because they measure how many electrons versus positrons are on one half of the whole chain. 
This can be easily seen from Gauss' law, i.e., $(E_r - E_{r-1} - \rho_r) \ket{\psi}$ for any physical, i.e., \emph{gauge-invariant} state $\ket{\psi}$ where 
\begin{align}
    \rho_r = \frac{1 - (-1)^r}{2} -\psi^\dagger_r \psi_r
\end{align}
is the electric charge on site $r$~\footnote{For even sites, and occupied electron state gives rise to $\rho_r \ket{1}_r= +1 \ket{1}_r$ and for odd sites an occupied positron state gives rise to $\rho_r \ket{0}_r= -1 \ket{1}_r$, and for unoccupied state at even/odd sites $\rho_r$ results in $0$ eigenvalue.}. It is straightforward to see that the initial state fulfills Gauss' law and thus, is gauge-invariant.
Note that by construction of the initial state and the Hamiltonian, the number of electrons is always equal to the number of positrons across the whole system.
However, the positions of the electrons and positrons can fluctuate among the chain, which could lead to states with asymmetric electron/positron occupations with respect to the middle of the chain. 
This could be probed by $\langle O_{\textrm{pol}} \rangle$ or higher moments $\langle O^n_{\textrm{pol}} \rangle$ for $n\geq 2$, and gives information about the disparity of electron-positron distribution across the middle of the chain.

Note that we can also check the \emph{k-local} field polarization at any region centered at $k_0$ on the lattice for size $k$ ($k_0 > k/2$), which is given by 
\begin{align}
O_{\textrm{pol-k}}:=E_{k_0 + \frac{k}{2}} -  E_{k_0 - \frac{k}{2}}.
\end{align}
This would give information about the electron-positron disparity within $k$-local regions of the spatial lattice around the lattice point $k_0$.

\noindent \textbf{The second experiment:} In the second experiment, we start from the following initial state 
\begin{align}
\ket{\psi(t=0)}= O_{p_0}\ket{\psi_{\textrm{gr}}},
\end{align}
where $O_{p_0}$ is a unitary operator that initializes two gauge-invariant, quasi-particle wavepackets (a \emph{real} positron-electron pair) moving away from each other with momenta $p_0$, where $\ket{\psi_{\textrm{gr}}}$ is the ground state at the specified parameters of the given Hamiltonian.
Note that we do not know, a priori, what this operator looks like, but we rely on classical methods such as matrix product states~\cite{white1993density, byrnes2002density,hastings2007entropy, haegeman2013elementary, banuls2013mass, buyens2017real, abrahamsen2023entanglement, Davoudi2024scatteringwave} and we assume an efficient quantum circuit that can prepare this state~\cite{fomichev2024initial, Davoudi2024scatteringwave,lemieux2024quantum}.
The choice of $p_0$ can be made depending on $\rho$, the particle-antiparticle density that we are hoping to produce within the experiment.
Considering relativistic particles obeying Dirac equation, as momentum is proportional to the energy, $p_0 \sim \rho \mu N$ which corresponds to the mass of $\rho N$ many particle-antiparticle pairs.
Note that we should be careful that $p_0$ does not exceed a certain small fraction of $N$, because the lattice is only a good approximation of the continuum for $p \ll \pi/a \sim \pi N$ when the simulation cell is chosen to be $V \sim \mathcal{O}(1)$.
The minimum time $t_{\min}$ for which we want to carry out the time evolution is determined by the minimum time it takes for all energy in the initial momenta to be transferred into the electric field and additional particle-antiparticle pairs.
Using similar arguments to the first experiment, we can argue that the third term gives the rate of creating particle-antiparticle pairs, i.e., $x$.
Hence we pick $t_{\min}= \rho/x$.
We then run the time-evolution for times
\begin{align}
\frac{\rho}{x}= t_{\min}< t.
\end{align}
The observables are the same as in the first experiment, i.e., we measure the particle density and the global electric field polarization of the half-system.

We remark that more involved measurements following metrology could be done, such as those studied in Refs.~\cite{wu2020schwinger, wu2020optimal, fan2024quantum}
, if efficiently employed within the overall quantum algorithm.

\subsection{Conditions for accurate field truncation}\label{subsec:cond_truncation}

The previous sections give a bound on the initial cutoff $\Lambda_0$.
However, as time evolution takes place, the support of the quantum state might leak out of the subspace defined as the image of the projector $\Pi_{[-\Lambda_0, \Lambda_0]}$, which projects onto the electric field values between $[-\Lambda_0, \Lambda_0]$ at each link. 
Ref.~\cite{tong2021Provably} demonstrates how one can bound the leakage outside of the subspace, and proves rigorous expressions for choosing sufficiently high electric field cutoffs $\Lambda = \Lambda_0 +\Delta$.
$\Delta$ is chosen minimally such that the field cutoff error is bounded by $\epsilon_{\textrm{cutoff}}$.

Following Ref.~\cite{tong2021Provably}, we apply the long-time leakage bound (Theorem 3) to the Schwinger Hamiltonian. 
Our electric field term belongs to the  $r=0$ case, and we compute the quantity $\chi$ which appears in the norm of $H^{\nu}_I$ acting on $\nu$th site, in the subspace $\Pi_{[-\Lambda_0, \Lambda_0]}$, i.e.,
\begin{align}
    \|H^{(\nu)}_I \Pi_{[-\Lambda,\Lambda]}\| \leq \chi (\Lambda+1)^r \substack{(r=0) \\ =\mathrel{\mkern-3mu}=} \chi.
\end{align}
In more detail, 
\begin{align}
\|H^{(\nu)}_I \Pi_{[-\Lambda,\Lambda]}\| & \leq \sum_{\nu} \| x \left( U_{\nu} \sigma^+_{\nu}\sigma^-_{\nu+1} + U^\dagger_{\nu}\sigma^-_{\nu}\sigma^+_{\nu+1} \right)  \| \\
\nonumber     & = \frac{x}{4}\sum_{\nu} \|  \left( U_{\nu}+U^\dagger_{\nu} \right)\left(X_{\nu}X_{\nu+1}+Y_{\nu}Y_{\nu+1}\right)\\
&+i\left( U_{\nu}-U^\dagger_{\nu} \right)\left(X_{\nu}Y_{\nu+1}-Y_{\nu}X_{\nu+1}\right) \| \\
& \leq 2x,
\end{align}
which implies that $\chi= 2x$.
Using this together with Theorem 3 of Ref.~\cite{tong2021Provably}, we can now bound the leakage of the electric field to higher values as a result of the full time evolution as follows:
\begin{align}
\| \bar \Pi_{[-\Lambda(t),\Lambda(t)]} e^{-itH}\Pi_{[-\Lambda_0,\Lambda_0]}  \| & \leq \frac{\lceil 4xt \rceil }{2^{\Delta-1}\Delta!}\\ 
& < \frac{2\lceil 4xt \rceil}{\sqrt{2\pi \Delta}} \left(\frac{e}{2\Delta}\right)^\Delta.
\end{align}
Assuming that $\Delta > e$, we have :
\begin{align}
\| \bar \Pi_{[-\Lambda(t),\Lambda(t)]} e^{-itH}\Pi_{[-\Lambda_0,\Lambda_0]}  \|  \leq \frac{2\lceil 4xt\rceil}{\sqrt{2\pi e}} \left(\frac{1}{2}\right)^\Delta. 
\end{align}
To ensure that the above expression is smaller than $\epsilon_\textrm{cutoff}$, we require 
\begin{align}
\Delta \geq \max(3, \lceil \log_2(2\lceil 4xt \rceil /\epsilon_{\textrm{cutoff}}\sqrt{2\pi e}) \rceil). 
\end{align}
This implies that for all
\begin{align}
\Lambda(t) & = \Lambda_0 + \lceil 2\chi t \rceil (\Delta - 1) =   \Lambda_0 + \lceil 4xt \rceil (\Delta - 1) \\
&\geq \Lambda_0 +\lceil 4xt \rceil \max\left (2, \left \lceil \log_2\left(\frac{2\lceil 4xt \rceil}{\epsilon_{\textrm{cutoff}}\sqrt{2\pi e}}\right) \right \rceil-1\right )
\end{align}
the simulation of the truncated Hamiltonian is, within an $\epsilon_\textrm{cutoff}$ approximation, faithful to the original, infinite-dimensional Hamiltonian.

\begin{table*}
    \centering
    \begin{tabular}{|c|c|c|}
    \hline
    Symbol & Meaning & Expression \\
    \hline
    $\Lambda_0$ & Electric field cutoff at time $t=0$ & $100 \mu \geq \Lambda^2_0 \geq 0.01 \mu$ \\
    \hline
    $\Lambda(t)$ & Electric field cutoff at time $t$ & $\Lambda(t) \geq \Lambda_0 +\lceil 4xt \rceil \max \left(2,\left \lceil \log_2\left(\frac{\lceil 4xt \rceil}{\epsilon_{\textrm{cutoff}}\sqrt{2\pi e}}\right) \right\rceil -1\right)$\\
    \hline
    $\gamma$ & Initial electric field value in the 1st experiment & $\Lambda^2_0 > \gamma^2 \geq \rho\mu$ \\
    \hline
    $|p_0|$ & Initial momenta of wavepackets in the 2nd experiment & $(\pi N \gg) |p_0| \geq  0.1 \rho \mu N$ \\
    \hline
    $t_{\min}$ & Minimum time for observing Schwinger effect & $\rho/x$ \\
    \hline
    $N$ & Number of sites on the lattice & $N_0 + 2 \max\{\ln(N_0/\epsilon), e 8|x||t|\}$  \\
    \hline
    \end{tabular}
    \caption{The parameter choices for the quantum simulation.
    The free parameters that determine the interaction strengths in the model, namely $x$ and $\mu = 2\sqrt{x}\,m/g$, are swept from $0.1$ up to $10$.
    The parameter $\rho$, the density of electron-positron pairs expected to be produced at any time during the simulation, is kept free and implicitly depends on the evolution time $t$ and the properties of the initial state. 
    The parameter $\epsilon$ denotes the target precision.
    The parameter $N_0$ is the spatial extent of the initial state and depends on the type of experiment; it can be estimated as the extent of the support of a few quasi-particles.
    All other parameters, which are swept depending on $x$ and $\mu$, can be inferred from the above table. Notice that the provided expression for $N$ acts as a lower bound for reliable simulation. In practice, we often take the number of links ($N-1$ for non-periodic boundary conditions and $N$ for periodic boundary conditions) to be the next power of two for ease of implementation.}
    \label{tab:Parameters}
\end{table*}

In Table~\ref{tab:Parameters}, we provide the set of parameters that should be simulated to induce the Schwinger effect, and potentially access different phases of the Schwinger model.
Note that the model has two essential parameters: the ratio $\mu= 2\sqrt{x}m/g$ and the interaction term $x (= 1/g^2)$, appearing as in Eq.~\eqref{eq:SimulationHamiltonian}.
The simulation can be performed for the set of parameters that are swept through a range of values, such as $0.1 \leq \mu \leq 10$ and $0.1 \leq x \leq 10$, with increments $0.1$.
One also makes a choice for a density $\rho \leq 1$ that corresponds to electron-positron pairs that can exist in the quantum state at any point during the simulation.
Then, given a particular instance of $\mu$, $x$, $\rho$, and target accuracy $\epsilon$, other parameters, i.e., the cutoff $\Lambda$, the minimum time evolution $t_{\min}$, and the minimum system size required $N$, are deduced in terms of the interaction strengths $\mu$ and $x$, and the time evolution $t$.
We assume that these parameters are already chosen high enough such that the error of making them finite is negligible.
In fact, this is reflected in the choice of $\Lambda$ explicitly by using the results of Ref.~\cite{tong2021Provably}.

\section{Methods}

The two quantum simulation methods considered in this work are the second-order Suzuki–Trotter product formula ($\textrm{PF}_2$) and interaction-picture (IP)-based algorithms.
On the one hand, the $\textrm{PF}_2$ method allows one to exploit the Hamiltonian’s structure through the commutators of its individual terms, which has been shown to be particularly advantageous for local Hamiltonians~\cite{jordan2012quantum, childs2021theory}.
On the other hand, the performance of the IP-based method scales linearly with the strength of the interaction term. Depending on the specific Hamiltonian, the scaling with respect to its parameters can favor either method.

There are two key differences between the methods that depend on the parameters $t$ and $1/\epsilon$.
Asymptotically, the IP-based method scales as $\tilde{\mathcal{O}}\left(t\log_2(t/\epsilon)\right)$, whereas the second-order Suzuki–Trotter method scales as $\mathcal{O}\left(t^{3/2}/\epsilon^{1/2}\right)$.
Thus, for large $t$ and small $\epsilon$, the IP-based method is expected to outperform the $\textrm{PF}_2$ method. However, the overall picture remains nuanced once all parameters are considered. We examine the corresponding resource costs in detail in later sections.

In this section, we provide an overview of these two methods.

\subsection{Quantum simulation of the Schwinger model with Suzuki-Trotter method}\label{sec:QSimProductFormula}

The first fault-tolerant quantum algorithm for simulating the Schwinger model~\cite{shaw2020quantum} is based on the second-order Suzuki-Trotter formula~\cite{suzuki1990fractal}. Subsequently, the second-order Suzuki-Trotter formula was applied to simulate higher-dimensional lattice QED~\cite{kan2022lattice} and non-Abelian lattice gauge theories including QCD~\cite{kan2022simulating}. In what follows, we briefly review the second-order Suzuki-Trotter formula, and then we compare the algorithms in~\cite{shaw2020quantum} and~\cite{kan2022lattice}.

The symmetric, second-order Suzuki-Trotter formula~\cite{suzuki1990fractal} admits an ordered decomposition of a given Hamiltonian $H$: $H = \sum_{l=1}^K H_l$, where every summand $H_l$ is a Hermitian operator, and approximates the time-evolution operator according to
\begin{align}\label{eq:PF2_def}
    e^{-iH t} & \approx S(t) \equiv U(t/\mathsf{r})^\mathsf{r}  \\
    &\equiv \Bigg[\left ( \prod_{l=1}^{K-1} e^{-i H_l \frac{t}{2\mathsf{r}}}\right) e^{-i H_K \frac{t}{\mathsf{r}}}\left ( \prod_{l=K-1}^{1} e^{-i H_l \frac{t}{2\mathsf{r}}} \right)\Bigg]^\mathsf{r},
\end{align}
where $\prod_{l=1}^{N} A_l = A_1 A_2... A_N$, $\prod_{l=N}^{1} A_l = A_N A_{N-1}... A_1$, and $\mathsf{r}$ is the number of Trotter steps. The approximation error $\epsilon$, typically known as Trotter error, is given by~\cite{childs2021theory}
\begin{align}\label{eq:trot_err}
    \epsilon &= \left\| e^{-iHt} - S(t) \right\| \\ 
    &\leq \frac{t^3}{12\mathsf{r}^2}\sum_{i=1}^K \Bigg\| \Bigg[ \Bigg[H_i, \sum_{j=i+1}^K H_j\Bigg], \sum_{k=i+1}^K H_k\Bigg] \Bigg\| \nonumber \\ 
    & \quad + \frac{t^3}{24\mathsf{r}^2}\sum_{i=1}^K \Bigg\| \Bigg[ \Bigg[H_i,\sum_{j=i+1}^K H_j\Bigg], H_i\Bigg] \Bigg\| \equiv \frac{t^3 \rho}{\mathsf{r}^2},
\end{align}
where $||\cdot ||$ is the spectral norm. Then, to reach a target accuracy $\epsilon$, we choose the number of Trotter steps 
\begin{equation}
    \mathsf{r} = \left \lceil \sqrt{\frac{\rho t^3}{\epsilon}} \right\rceil.
\end{equation}
In practice, we implement the operation $W(t, \epsilon_{\rm{rot}})$, which is equivalent to $S(t)$ up to rotation synthesis error $\epsilon_{\rm{rot}}$. Furthermore, the gate cost of a simulation based on Suzuki-Trotter formula is roughly the gate cost of $U(t)$ multiplied by $\mathsf{r}$.

Note that the ordered decomposition adopted in~\cite{shaw2020quantum} and~\cite{kan2022lattice} are different. We choose to use the one in~\cite{kan2022lattice} because it incurs lower gate costs per Trotter step and a smaller Trotter error. Specifically, we implement the commuting $e^{-i H_E t}$ and $e^{-iH_M t}$ first, then we divide $H_I$ into four parts, i.e., $H_I = H_{1,e}+H_{1,o}+H_{2,e}+H_{2,o}$, before implementing their evolutions sequentially. The detailed resource cost of the resulting subroutines are explained Section~\ref{subsec:ResourceCostTrotter} and given in Table~\ref{tab:ResourceCostTrotter}.

\subsection{Quantum simulation of the Schwinger model with the interaction picture method}\label{sec:QSimIP}

A Hamiltonian simulation method based on the interaction picture has been established in Ref.~\cite{low2018hamiltonian}.
In this section, we give a brief overview of the algorithm with a more detailed account given in Appendix.~\ref{subsec:OverviewIP} and its detailed implementation in Appendix.~\ref{subsec:QuantumAlgorithmInteractionPicture}.

Given $H= H_0 + V$, the basic idea of the method relies on operating in the rotating frame defined by $e^{-iH_0 t}$, i.e., we make sure that at every stage, denoted by a time $s$, in the algorithm, the quantum state in the interaction picture is obtained by applying $e^{-iH_0 s}$ on the quantum state in the Schrödinger picture.
To be more precise, as straightforwardly shown in Lemma~\ref{lem:SchrodingerInteractionPictureTimeEvolutionRelation} or in a standard text~\cite{sakurai2020modern}:
\begin{align}
\ket{\psi(t)}& = e^{-iHt} \ket{\psi(t=0)}\\
 & = e^{-iH_0t}U_I(t) \ket{\psi(t=0)} ,
\end{align}
where
\begin{align}\label{eq:IPunitary}
U_I&(t) = \mathcal{T} \Big[ e^{-i \int^t_{0} V(s) ds} \Big]\\
 &= \sum_{n=0}^\infty \frac{(-i\hbar^{-1})^n}{n!} \int_{0}^{t}dt_1 ... dt_n \mathcal{T}\left[ \prod_{k=1}^{n} V(t_k) \right]
\end{align}

Beyond product formulas, quantum algorithms for Hamiltonian simulation use block-encodings of the Hamiltonian, which is a unitary matrix that contains $H/\alpha_H$ -- $H$ rescaled by a factor $\alpha_H \in \mathbb{R}$ -- in a chosen, accessible block; here, we follow what is typically done in the literature: the upper-left block, i.e., the subspace where the ancilla is in the 0 state. In general, the smaller the rescaling factor $\alpha_H$, the more efficient the algorithm.
Indeed, the algorithm based on the interaction picture, which we hereafter synonymously refer to as the IP-based algorithm, gives significant improvements when the block-encoding rescaling factors of $H_0$ and $V$ satisfy $\alpha_{H_0} \gg \alpha_V$ and $H_0$ is fast-forwardable.
Note that this is a common scenario, where $V$ is a slight perturbation of an easily diagonalizable Hamiltonian $H_0$.
For the Schwinger model, this is true when the electric field is strong.
In particular, we denote $H_0= H_E + H_M$ and $V= H_I$, and the complexity of the algorithm will scale only mildly (polylogarithmically) with the rescaling factor of $H_E$ and $H_M$. 

The IP quantum algorithm approximates the operator $U_I(t)$ given in Eq.~\eqref{eq:IPunitary}. 
In order to approximate this operator with a given precision $\epsilon$ in operator norm, we first expand the time-ordered integral via a Dyson series, and define two types of truncated discretized versions of it.
These are given in Definition~\ref{def:TruncDiscDysonSeries} and Definition~\ref{def:TruncDiscDysonSeriesWithCol}, in which \emph{collisions} are disregarded and included, respectively; the collisions arise when the time integral is approximated as a discretized Dyson series~\cite{low2018hamiltonian}.
$K$ denotes the truncation degree of the Dyson series, while $1/M$ denotes the discretization coarseness in the time integrals.
We analyze these approximations in Theorem~\ref{thm:ApproximatingInteractionPictureTimeEvolWithTruncatedDiscretizedDysonSeries} and in Theorem~\ref{thm:ApproximatingInteractionPictureTimeEvolWithTruncatedDiscretizedDysonSeriesWithCollision} and study the bounds for the truncation and discretization parameters $K$ and $M$, in terms of $\|V\|$, $\|H_0\|$, $t$ and $\epsilon$.
Corollary~\ref{cor:TimeEvolutionSchrPicWithTruncDiscDysonSeries} then combines the analysis for the final bounds for $K$ and $M$ that are used in the design of the IP-based algorithm. While our analysis for the algorithm without collisions leads to only minor improvements over previous analysis in~\cite{low2018hamiltonian}, our analysis for the algorithm with collisions is novel and indicates a similar performance compared to the one without collisions. Indeed, our new error analysis reveals that including collisions in the block-encoding leads to negligibly bigger approximation error compared. In addition, it is more expensive to exclude collisions than to include them, as excluding collisions requires additional subroutines to either flag the collisions out of the success subspace of the block-encoding, or to prepare a state that ensures no collisions will occur. Therefore, we find including collisions to be the more practical and economical choice.
\\

The IP-based quantum algorithm closely follows the LCU approach combined with oblivious amplitude amplification (OAA), as described in Refs.~\cite{berry2015simulating, kieferova2019simulating, low2018hamiltonian}. We divide the total time evolution $t$ into smaller pieces $t_0/\alpha_V$, and implement each piece sequentially.
The time step $t_0$ defines the rescaling factor of the block-encoding of the Dyson series~\cite{berry2015simulating}, which dictates the number of rounds of OAA one needs to perform per time step. In practice, $t_0$ is optimized such that each time slice requires only a single round of OAA. We construct a quantum circuit $W_{(K,M)}\left(\frac{t_0}{\alpha_V}, \frac{\epsilon'}{\mathsf{r}}\right)$ that approximates $U_{I}\left(\frac{t_0}{\alpha_V}\right)$ with an error that is at most $\frac{\epsilon'}{\mathsf{r}}$, where $\mathsf{r} = \left\lceil \frac{\alpha_V t}{t_0} \right\rceil$ is the number of time steps required to evolve for a total time $t$.

To enforce time ordering of the Dyson series, a natural choice is to implement coherently a sorting algorithm of $K$ time registers of size $\lceil \log_2 M \rceil$. This introduces a qubit overhead that is polynomial in both $K$ and $\lceil \log_2 M \rceil$. As an alternative, Ref.~\cite{low2018hamiltonian} proposes using only two time registers. This reduces the number of required qubits, but it increases the LCU rescaling factor, which in turn decreases the optimal choice of $t_0$. Specifically, with $K$ time registers and sorting, the optimal value $t_0 = \ln 2$, whereas the two-register implementation yields $t_0 = 0.5$. This corresponds to approximately a $39\%$ increase in the total number of time steps $\mathsf{r}$. We find that this increase leads to a higher overall T-count than the additional cost incurred by the sorting algorithm. Combined with the binary encoding of the truncation index, Ref.\cite{low2018hamiltonian} report a qubit scaling in $\lceil \log_2 K \rceil$ instead of $K$. In our case, since $K$ time registers are already present, the marginal qubit cost of using a unary encoding for the truncation index register is justified by the resulting reduction in T-count.

This quantum circuit is given in Lemma~\ref{lem:CircuitImplementationOfDysonSeries} which combines the block-encodings of $V$ and fast-forwarded implementations of $e^{-iH_0 s}$.
To be more precise, we execute the $W_{(K,M)}(t_0/\alpha_V, \epsilon'/\mathsf{r})$ subroutine $\mathsf{r} - 1 = \lfloor \alpha_V t / t_0 \rfloor$ times, and execute $W_{(K,M)}(t', \epsilon'/\mathsf{r})$ once for $t'= t - t_0 \lfloor \alpha_V t /t_0  \rfloor /\alpha_V$.
We finally apply $e^{-iH_0 t}$.
This implementation is given in detail in Corollary~\ref{thm:Main}.

We note that the piece-wise implementation of the LCU-based approach, leads to a the multiplicative $\ln(1/\epsilon)$ scaling rather than an additive scaling. The quantum computational resource cost of implementing the algorithm is given in detail in Table~\ref{tab:ResourceCostIP} and explained in Section~\ref{subsec:ResourceCostIP}. %

\section{Results}\label{sec:Results}

In this section, we summarize the resource cost analysis of implementing the $\textrm{PF}_2$ and IP-based algorithms. Both algorithms divide the full time evolution into smaller time steps that are repeated sequentially. These smaller time evolution operators are further decompose into subroutines, for which their costs are listed in Table~\ref{tab:ResourceCostTrotter} and Table~\ref{tab:ResourceCostIP}, respectively.
These cost expressions are then run for specific parameter regimes to give explicit T-gate and qubit counts for the performances of the algorithms.
We focus on the $\textrm{PF}_2$ approach in Sec.~\ref{subsec:ResourceCostTrotter}, and on interaction picture in Sec.~\ref{subsec:ResourceCostIP}. In Sec.~\ref{subsec:Comparison}, we give the asymptotic and numeric quantum resource estimates and  compare the results. Further details on the resource estimates can be found in Appendix~\ref{app:trotter} and ~\ref{subsec:CompilationAndResourceCount}.

\subsection{Resource costs for Suzuki-Trotter method}\label{subsec:ResourceCostTrotter}

As shown in Eq.~\eqref{eq:PF2_def}, the 2nd order Trotter formula approximates the time evolution operator as a product that involves the sub-evolutions $e^{-iH_E\tau}$, $e^{-iH_M\tau}$, and $e^{-iH_{1/2, e/o}\tau}$. Whenever possible, we merge contiguous sub-evolutions. Thus, the value of $\tau$ is either $t/2\mathsf{r}$ or $t/\mathsf{r}$. Below, we briefly explain our quantum circuit implementation of the sub-evolutions and provide their associated cost, which we summarize in Table~\ref{tab:ResourceCostTrotter}. To simplify notation, let's define $\eta= \log_2 2\Lambda$. We provide further details in Appendix~\ref{app:trotter}.

\begin{enumerate}
    \item $e^{-iH_E\tau}:$ We start from the implementation in~\cite{shaw2020quantum}: $e^{-iH_E\tau}$ is first expressed as a product of $e^{-iE^2 \tau}$ acting on every bosonic register. Then, each $e^{-iE^2 \tau}$ is decomposed into $\eta-2$ layers -- interleaved with layers of CNOT gates -- of at most $\eta$ $R_z$ rotations of the form $\otimes_i R_z(2^i \tau)$, as depicted in Fig.4 of Ref.~\cite{shaw2020quantum}. Instead of synthesizing every $R_z$ gate individually as in~\cite{shaw2020quantum}, we effect each $\otimes_i R_z(2^i \tau)$ using the phase catalysis circuit from~\cite{Wang2024optionpricingunder} shown in Fig.~(168), which is more gate-efficient, mainly because one only needs to synthesize a reusable catalyst state $\otimes_{i=0}^{2\eta-3} R_z(2^i \tau)\ket{+}$ once for the entire time-evolution circuit.
    \item $e^{-iH_M\tau}:$ Following~\cite{kan2022lattice}, we reduce $e^{-iH_M\tau}$ into $N$ same-angle $R_z$ rotations by applying 2 NOT gates each to half of the fermionic registers. Then, we implement the same-angle $R_z$ gates using catalyzed Hamming-weight phasing~\cite{kan2024resourceoptimized}, which improves upon the circuit in Fig.~(168) from Ref.~\cite{Wang2024optionpricingunder}. Note that~\cite{kan2022lattice} uses (uncatalyzed) Hamming-weight phasing to implement $e^{-iH_M\tau}$.
    \item $e^{-iH_{1/2, e/o}\tau}:$ Once again, we follow~\cite{kan2022lattice}, differing mainly in choosing catalyzed Hamming-weight phasing over (uncatalyzed) Hamming-weight phasing. More specifically, $e^{-iH_{1, e/o}\tau}$ can be expressed as a product of operators $e^{-i \tau x \sigma^+\sigma^+\sigma^- + h.c.}$, where $\sigma^\pm$ are the Pauli raising and lowering operators, and $e^{-iH_{2, e/o}\tau}$ is a product of $U e^{-i \tau x \sigma^+\sigma^+\sigma^- + h.c.}U^\dagger$, where $U$ is an adder circuit. We then compile every $e^{-i \tau x \sigma^+\sigma^+\sigma^- + h.c.}$ into a layer of two same-angle $R_z$ gates conjugated by temporary ANDs~\cite{gidney2018halving}, Hadamards and NOT gates~\cite{Wang2021resourceoptimized}. The resulting layer of $N$ same-angle $R_z$ gates in $e^{-iH_{1, e/o}\tau}$ is then implemented using catalyzed Hamming-weight phasing. $e^{-iH_{2, e/o}\tau}$ is implemented similarly, with the necessary, additional applications of adder circuits.
\end{enumerate}

\begin{table*}
    \centering
    \begin{tabular}{|c||c|c|c|c|}
    \hline
    Subroutine & $\#$ T-gates  & $\#$ rotations  & $\#$ ancilla &$\#$ total number of times \\
    \hline
    \hline
    $e^{-iH_M s/2}, e^{-iH_M s}$ & $4N-4 + 4\lfloor \log_2N \rfloor  $ & $1$ & $N + \lfloor \log_2N \rfloor + 1$ & $2, \mathsf{r} - 1$\\ \hline 
    $e^{-iH_E s/2},e^{-iH_E s}$ &  $2(N-1) (\eta^2 +\eta - 2)$ & $(N-1) \eta$ & $\eta$ & $2, \mathsf{r} - 1$\\ 
    \hline 
    $e^{-iH_{1, e/o} s/2}$ & $6N-4 + 4\lfloor \log_2N \rfloor $ & $1$ & $3N/2 + \lfloor \log_2N \rfloor$ & $2\mathsf{r}$\\ \hline 
    $e^{-iH_{2, e} s/2}, e^{-iH_{2, o} s}$ & $6N-4 + 4\lfloor \log_2N \rfloor + 8N\eta-8N$& $1$ & $\max(3N/2 + \lfloor \log_2N \rfloor, \eta)$ & $2\mathsf{r}, \mathsf{r}$\\  
    \hline 
    \end{tabular}
    \caption{The resource costs of each subroutine (in terms of the number of T-gates, the number of rotations, and the number of additional ancilla qubits) and the number of times they are used in the full quantum circuit. 
    The unit Trotter evolution time $s = t / \mathsf{r}$, where $r \in \mathbb{N}^+$ is predetermined such that $\|W_{s}(\epsilon/\mathsf{r}) - e^{-iHs}\| \leq \epsilon/\mathsf{r}$.
    Note that every operation is exact except the rotation syntheses, which are implemented precise enough to meet the total desired error~\cite{kliuchnikov2023shorter}.}
    \label{tab:ResourceCostTrotter}
\end{table*}

\subsection{Resource costs for the interaction picture method}\label{subsec:ResourceCostIP}

As explained in Section~\ref{sec:QSimIP}, the quantum circuit calls $W_{(K,M)}(t_0/\alpha, \epsilon'/\mathsf{r})$, $\mathsf{r}$ times, where $\mathsf{r}=\lceil t \alpha/t_0 \rceil$.
As we find, $\alpha= 2Nx$ which is the coefficient $1$-norm of the LCU expansion given in Eq.~\eqref{eq:LCUHint}.
In this section we study the gate cost of implementing $W_{(K,M)}(t_0/\alpha, \epsilon'/\mathsf{r})$, which is reported in Table~\ref{tab:ResourceCostIP}.
The quantum circuit employed to realize the operator $W_{(K,M)}(t_0/\alpha, \epsilon'/\mathsf{r})$ is given in Fig.~\ref{fig:IP}.
Its correctness is proven in Appendix~\ref{subsec:OverviewIP}, more precisely refer to Lemma~\ref{lem:CircuitImplementationOfDysonSeries} and Corollary~\ref{thm:Main}. 
More details of the circuit implementation can be found in Section~\ref{subsec:CompilationAndResourceCount}.
The quantum circuit given in Fig.~\ref{fig:IP} plays a key role:
it is used in OAA in order to implement $W_{(K,M)}(t_0/\alpha, \epsilon'/\mathsf{r})$.
We study the cost of the subroutines used in this quantum circuit, which are given Table~\ref{tab:ResourceCostIP}.

\begin{enumerate} 
    \item {\bf{$\PREP_{\sqrt{t^k_0/k!}}$:}} The first subroutine in PREP of the Dyson series is implemented by preparing a $k$-hot state:
    \begin{align}\label{eq:WOutputOfPREP}
     \text{PREP}_{\sqrt{t_0^k/k!}} \ket{0}^{\otimes K}  = \frac{1}{\beta}\sum_{k = 0}^{K-1} \sqrt{\frac{t_0^k}{k!}}\ket{1}^{\otimes k}\ket{0}^{\otimes K-k}.
    \end{align}

This can be obtained by the circuit in Fig.~\ref{fig:prep_k_hot}. 
The output register encodes the truncation index $k$, in $k$-hot unary representation.
It requires at most one rotation and $K-1$ controlled-rotations (which can be implemented with two rotations each).
It requires no additional ancilla and is called $6$ times, three times due to oblivious amplitude amplification (OAA), and twice within each OAA, as $\PREP$ and $\PREP^\dag$.

    \item{\bf{Additional $\PREP$:}} We now prepare the discretized time registers. 
    There are $K$ registers of size $\lceil \log_2M \rceil$ each.
    Assuming $M$ is an integer power of $2$, one can use controlled-Hadamards to prepare a uniform superposition of states on these time discretization registers. 
    In total, this costs $2 K \log_2 M $ T-gates.
    
    \item {\bf{$\SORT$:}} We then apply a bitonic sort~\cite{batcher1968sorting,muller1975bounds} which arrange the integer values in the $K$ time registers, each of size $\log_2 M$, in increasing order.
    This takes at most $\lfloor K/2 \rfloor (\lceil \log_2 K\rceil + 1)\left\lceil \log_2K\right\rceil$ comparators (with equality if $K$ is an integer power of $2$), and the same amount of C-SWAPs.
    Both the comparator and C-SWAPs require $\lceil \log_2 M \rceil$ Toffolis each. 
    This allocates at most $\lfloor K/2 \rfloor (\lceil \log_2 K\rceil + 1)\left\lceil \log_2K\right\rceil/2$ record qubits that will be deallocated at uncomputation. 
    An extra $\lceil \log_2 M \rceil$ temporary ancillae are used for the comparator. 
    The uncomputation can be done with a measurement-based procedure cheaper on average than the computation~\cite{luongo2024measurement}, whereas we assume calling this subroutine $6$ times: three times due to OAA, and twice within each OAA, as $\SORT$ and $\SORT^\dag$, since the uncomputation is at worst the same cost as the computation.
    
    \item {\bf{$\BE_{V/\alpha}$:}} The $k$th $\BE_{V/\alpha}$ that appears in the circuit, is singly controlled on the $k$th qubit of the state given in Eq.~\eqref{eq:WOutputOfPREP}.
    Note that this $k$ is related to the $k$th application of $V$ in the truncated-discretized Dyson series.
    We implement each $\BE_{V/\alpha}$ following the implementation of~\cite{rajput2022Hybridized},  which takes a linear combination of local terms that are of equal weight and implements $\PREP_{\BE}-\SEL_{\BE}-\PREP^\dag_{\BE}$.
    Importantly, the control of the block-encoding only controls certain parts of $\PREP_{\BE}$, and some operations of $\SEL_{\BE}$ (see Appendix~\ref{subsec:CompilationAndResourceCount} for more details).
    $\PREP_{\BE}$ prepares a $W$ state of $N-1$ qubits in a tensor product with $H^{\otimes 3}$. 
    $\SEL_{\BE}$ is implemented following Ref.~\cite{rajput2022Hybridized} (see Section~5.2, Fig.~4).
    By choosing the number of links ($N-1$ for open boundary condition, $N$ for periodic) as an integer power of two, the cost of $\BE_{V/\alpha}$ is $8N + 4(N - 1)(\eta - 1) - 1$ T-gates and $\eta - 1$ ancillae. 
    This subroutine is called $3K$ times in OAA. 
    
    \item {\bf{$e^{-iH_M (\cdot)}$:}} This operation is conditioned on the discretized time registers, $m_1$, $m_2$, \ldots, $m_K$.
    The implementation combines consecutive forward/backward time-evolutions and first computes $m'_k= m_k-m_{k-1}$ for $k \in [2,K]$, which are counted as additional gates below.
    Then controlled on $m'_k$, we implement $e^{-i H_M m'_k/(\alpha M)}$. We compare two implementations that we call PGA and Mult.
    In the PGA approach, one first computes the Hamming weight of the fermions, $n_f$. 
    Controlled on $m'_k$, we then perform a phase gradient addition with $n_f$.
    The cost of this is $N - 1 + \lceil\log_2M\rceil( \lfloor \log_2N\rfloor+1)$ Toffolis and $\lceil \log_2M \rceil$ rotations,
    $N + \lfloor \log_2(N)\rfloor + 1$ ancilla qubits. 
    In the Mult approach, we perform a multiplication of $m'_k$ and $n_f$ and than perform a phase gradient addition to add the result to the phase. 
    The cost of this approach is $N + 2\log_2M \lfloor\log_2 N\rfloor + 7\log_2M + 5\lfloor\log_2 N\rfloor + 4$ Toffolis, $1$ rotation and $N+ 2\log_2M + 2\lfloor\log_2 N\rfloor + 1$ ancillae.
    In both cases, the subroutines are called in total $3(K+1)$ many times per application of $W_{(K,M)}$.

    \item {\bf{$e^{-iH_E (\cdot)}$:}} The phase applied in this operation also depends on the discretized time registers, $m_1$, $m_2'$, \ldots, $m_{K-1}'$, $m_K$. Similar to the mass term, we also have two implementations, a PGA and a Mult approach. In the PGA case, each controlled time evolution is performed by a phase-gradient~\cite{sanders2020compilation} version of the implementation in Ref.~\cite{shaw2020quantum} (see their Section~3.2).
    The cost of this subroutine is
    $(N-1)\lceil\log_2M \rceil (\eta^2 +\eta - 2)/2$  Toffolis, and $(N-1)\lceil\log_2M \rceil\eta$ rotations. 
    An additional $\eta$ ancillae are used (excluding the catalyst state).
    In the Mult approach, the squared electric fields are added in a binary tree-fashion to reduce additional qubit costs.
    Then, the total electric term is multiplied by the time register. The result is then kicked-back as a phase via phase-gradient addition.
    This costs $N + 2\log_2M \lfloor\log_2 N\rfloor + 7\log_2M + 5\lfloor\log_2 N\rfloor + 4$ Toffolis, $1$ rotation and $N+ 2\log_2M + 2\lfloor\log_2 N\rfloor + 1$ ancillae.
    In both cases, the subroutines are called $K+1$ times in the circuit in Fig.~\ref{fig:IP}, which is repeated $3$ times via OAA, hence in total called $3(K+1)$ times.\\
    
    \item {\bf{Additional-$\SEL$:}} We count the assisting elementary gates that appear in Fig.~\ref{fig:IP} for the in-place subtraction/addition and the compression gadget used to reduce the ancillary cost of multiplying block-encodings. 
    For the latter, we used a unary implementation of the method discussed in~\cite{fang2023time}.
    There are $2(K-1)$ many subtractions/additions each of which costs $(\lceil \log_2 M\rceil - 1)$ Toffolis and $(\lceil \log_2 M\rceil - 1)$ temporary ancillae~\cite{gidney2018halving}.
    In parallel, the counter register is acted on via multi-controlled-NOTs, in total $K$ times, and each of them costs $N+1$ Toffolis.
    All of these are called $3$ times of OAA.\\
    
    \item {\bf{Reflection:}}
    Furthermore, there is an additional set of Toffolis that is used in OAA for reflections on three registers: the one that holds the value of $K$; the ones that hold the values of $m_1$, $m_2$, \ldots, $m_K$; and the counter register that is used for multiplying the block-encodings $\BE_{V/\alpha}$.
    In total, there are $K(2 +  \lceil \log_2 M \rceil )$ qubits required for the reflection.
    This costs $(2K + K \lceil \log_2 M \rceil -1)$  Toffolis and additional temporary ancillae.
    This subroutine is called twice for OAA.
\end{enumerate}

\begin{table*}
    \centering
    \resizebox{\textwidth}{!}{
    \begin{tabular}{|c|c|c|c|c|}
    \hline
    Subroutine & $\#$ T-gates  & $\#$ rotations  & $\#$ ancilla &$\#$ times \\
    \hline
    \hline
    $\PREP_{\sqrt{t^k_0/k!}}$ & - & $2K-1$& - & $6$ \\
    \hline
    Additional - $\PREP$ & $2 K \lceil \log_2M \rceil$ & $-$ & - & $6$ \\
 \hline
    $\SORT$ &  $4\lfloor K/2\rfloor(\lceil \log_2K\rceil + 1)\lceil \log_2 M \rceil$& - & $\lceil \log_2 M\rceil +\left \lfloor \frac{K}{2} \right \rfloor (\lceil\log_2K\rceil + 1)\frac{\left\lceil \log_2K\right\rceil}{2}$ & $6$\\
    \hline
    $\BE_{V/\alpha}$ & $8N + 4(N-1)(\eta-1) - 1$ & - & $\eta - 1$ & $3K$\\
    \hline
     $e^{-iH_M (\cdot)}$ (PGA)
     & $4N - 4 + 4\lceil\log_2M\rceil (\lfloor \log_2N\rfloor + 1)$ & $\lceil\log_2M\rceil$ &  $N + \lfloor \log_2N\rfloor + 1$  & \multirow{2}{*}{$3(K+1)$} \\
     $e^{-iH_M (\cdot)}$ (Mult)
     & $4[N + 2\log_2M \lfloor\log_2 N\rfloor + 7\log_2M + 5\lfloor\log_2 N\rfloor + 4]$ & 1 &  $N+ 2\log_2M + 2\lfloor\log_2 N\rfloor + 1$  & \\
    \hline
    $e^{-iH_E (\cdot)}$ (PGA) & $2(N-1)\lceil\log_2M \rceil (\eta^2 +\eta - 2)$ & $(N-1)\lceil\log_2M \rceil \eta$ & $\eta$ & \multirow{2}{*}{$3(K+1)$}\\
    $e^{-iH_E (\cdot)}$ (Mult)
     & $4N[4\eta^2 + 4\eta ]+ 4\log_2 M [4 \eta + 5 + 2 \lceil \log_2 N \rceil] + 20 \lceil \log_2 N \rceil - 8\eta^2 + 48\eta$ & 1 & $8 \eta  + 3\lceil \log_2N \rceil + 2\log_2M$ & \\
    \hline
 Additional - $\SEL$ & $ 8(\lceil \log_2M \rceil - 1)(K-1) + 4 K (N + 1) $ & $-$ & $\max(N+1, \lceil \log_2M\rceil - 1)$ & $3$ \\
 \hline
  Reflection & $8K+4K\lceil\log_2M\rceil-4$ & $-$ & $2K+K\lceil \log_2M\rceil - 1$ & $2$ \\
    \hline
    \end{tabular}
    }
    \caption{The resource costs of each subroutine (in terms of the number of T-gates, the number of rotations and the number of additional ancilla qubits) and the number of times they are used for constructing the quantum circuit $W_{(K,M)}(t_0/\alpha_V, \epsilon'/\mathsf{r})$ ($\approx e^{-iHt_0/\alpha}$). 
    To obtain $W(t,\epsilon')$ ($\approx e^{-iHt}$), one needs to repeat $W_{(K,M)}(t_0/\alpha, \epsilon'/\mathsf{r})$,  $\mathsf{r} = \left\lceil t \frac{\alpha}{t_0} \right\rceil$ many times. Notice that we assumed an implementation of each Toffoli as described in~\cite{gidney2018halving}, using 4 T-gates and one ancilla per Toffoli.
    Note that every operation is exact except the synthesis of rotation gates, which are implemented to our desired precision~\cite{kliuchnikov2023shorter}.}
    \label{tab:ResourceCostIP}
\end{table*}

\subsection{Asymptotic and numerical QREs and comparison}\label{subsec:Comparison}

In this section, we first give a comparison of the $\textrm{PF}_2$ method with the $\textrm{IP}$-based quantum algorithm in terms of the asymptotic scaling depending on the model and simulation parameters, i.e., $x$, $\mu$, $N$, $\Lambda$, $t$, and $\epsilon$.
Then, we give instantiations of the quantum resource estimates for parameter regimes for which one expects to observe Schwinger effect, see Table~\ref{tab:Parameters}.
Fig.~\ref{fig:trotter_vs_IP}, show the performance of the two methods, and can help with choosing the most efficient method to use.

$\textrm{PF}_2$ calls a unit Trotter evolution, $W_{s}(\epsilon/\mathsf{r})$, 
\begin{align}
    \mathsf{r}= \mathcal{O}(t^{3/2}\rho^{1/2}/\epsilon^{1/2})
\end{align}
many times, where $\rho$ is an upper bound on the norm of the third-order commutator.
When we employ the asymptotic expression for $\rho$, given as in Eq.~\eqref{eq:rho_asymptotic}, and observing that each $W_{t/\mathsf{r}}(\epsilon/\mathsf{r})$ costs $\tilde{\mathcal{O}}(N)$ (see Table~\ref{tab:ResourceCostTrotter}), we find the resulting cost given in Table~\ref{tab:ComparisonOfAsymptoticCosts} for $\textrm{PF}_2$.
On the other hand $\textrm{IP}$ calls a unit time evolution $W_{(K,M)}(t_0/\alpha,\epsilon'/\mathsf{r})$, $\mathsf{r} = \left\lceil \frac{t\alpha }{t_0} \right\rceil$ many times, where the rescaling factor 
\begin{align}
    \alpha &= \mathcal{O}(Nx), \\
    t_0 &= \mathcal{O}(1), \\
    \epsilon' &=\mathcal{O}(\epsilon).
\end{align}
The most expensive part of implementing $W_{t/\mathsf{r}}(\epsilon/\mathsf{r})$ is for implementing the time evolution of the electric term $e^{-iH_E(\cdot)}$ which is scales as $\mathcal{O}(N(\log_2M) \eta^2)$, as seen in Table~\ref{tab:ResourceCostIP}. 
$M$ can be found in Eq.~\eqref{thm:TruncDiscDysonApproximationErrorConditionM} which for $t = t_0/\alpha$ gives $\mathcal{O}(\log_2M) = \tilde{\mathcal{O}}(1)$.
Finally, recalling that $\eta= \log_2 (2\Lambda)$, we obtain the resulting cost given in Table~\ref{tab:ComparisonOfAsymptoticCosts} for $\textrm{IP}$.
This analysis shows that for larger time evolutions (large $t$), more precise simulation (large $1/\epsilon$), or high electric field cutoff $\Lambda$, $\textrm{IP}$ should be the method of choice, assuming that one is not limited by qubits.
On the other hand $\textrm{PF}_2$ performs better in terms of scaling in $N$, e.g. $\mathcal{O}(N^{3/2})$ vs. $\mathcal{O}(N^{2})$, with minimal ancilla overhead.
In terms of scaling with $x$, which can be a dominant factor in the continuum limit, $\textrm{PF}_2$ has a quadratic improvement.

\begin{table}
    \centering
    \begin{tabular}{|c||c|}
    \hline
     & Asymptotic gate complexity\\
    \hline
    $\textrm{PF}_2$ &  $ \tilde{\mathcal{O}}\left(\Lambda{\left(\frac{N^3 t^3 x}{\epsilon}\right)^{1/2}\left(1+ \frac{\mu^2}{\Lambda^2} + \frac{x}{\Lambda^2} + \frac{\mu}{\Lambda} + \frac{x \mu}{\Lambda^2} + \frac{x}{ \Lambda} \right)^{1/2}} \right) $ \\
    \hline
    IP & $ \mathcal{\tilde O} \left( N^2 x t (\log_2 \Lambda)^2 \log_2(1/\epsilon) \right)$ \\
    \hline
    \end{tabular}
    \caption{Asymptotic gate complexity for Schwinger model: 2nd order product formula ($\textrm{PF}_2$) vs. Interaction Picture (IP) of us in terms of $\Lambda, t, N, 1/\epsilon$, $\mu$ and $x$. Note that $\eta= \log_2 2\Lambda$.}
    \label{tab:ComparisonOfAsymptoticCosts}
\end{table}

Now, we compare the performance of different implementations of the IP-based algorithm, and the most-performant IP-based algorithm with the most performant $\textrm{PF}_2$-based algorithm.
Numerical results in Fig.~\ref{fig:with_vs_without_sort} and~\ref{fig:pga_vs_mult}  give the resource costs in terms of T count (y-axis) for varying time parameter $t$ (x-axis), and for $\epsilon= 0.1, 0.01, 0.001$ precision in the implementation of the time evolution operator.

First, in Fig.~\ref{fig:with_vs_without_sort} the choice of $t_0= 0.5$ (which does not use SORT) vs. $t_0= \ln 2$ (which uses SORT) are compared.
We find out that the choice of $t_0= \ln 2$ is more efficient, as expected, and the fact that it has the additional subroutine SORT does not drastically deteriorate the expected gain of $\sim 1.4 \times$.
This is, however, at the expense of hundred(s) of additional qubits.

\begin{figure}
    \centering
    \includegraphics[width=0.8\linewidth]{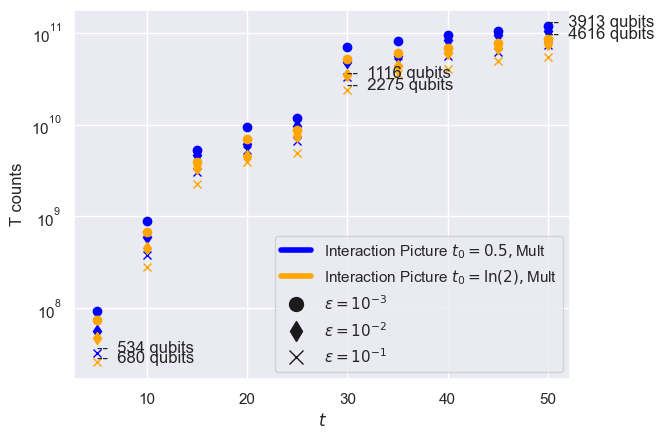}
    \caption{Resource cost comparison of two different implementations of IP-based algorithms.
    The blue data points correspond to the implementation without SORT, using $t_0= 0.5$, as described in Ref.~\cite{low2018hamiltonian}.
    The orange data points correspond to the implementation using the optimal $t_0= \ln(2)$ with SORT, which requires more qubits but fewer T-gates.
    The comparison is shown in terms of T counts and qubit counts for time parameters $t \in (t_\text{min}, 10t_\text{min})$, and accuracies $\epsilon \in \{0.001, 0.01, 0.1\}$. We fix $N_0 = 8$, $x = 0.1$, $\mu = 1$, $\Lambda_0 = \sqrt{10}\,\mu = \sqrt{10}$, and $t_\text{min} = 0.5/x$. The parameters $N$ and $\Lambda$ are chosen according to Tab.~\ref{tab:Parameters}.}
    \label{fig:with_vs_without_sort}
\end{figure}

Second, in Fig.~\ref{fig:pga_vs_mult}, we compare, for $t_0= \ln 2$, the performance of two different implementations that we called PGA and Mult.
Note that similar alternatives (PGA vs. Mult) are also available for the $\textrm{PF}_2$ based method.
However,  the best choice here is clearly the PGA approach, due to the fact that in the Trotter implementation, the time is a classically known value, and not a quantum one like in IP.
In the IP-based algorithm, this control-structure on PGA leads to a $\log_2 M$ multiplicative factor.
This turns out to be more expensive than the Mult approach which has additional arithmetic for computing the sum of the squares of electric fields.
In the $\textrm{PF}_2$ approach, the absence of control of the PGA makes it far more efficient.
In fact, our results, as seen in Fig.~\ref{fig:pga_vs_mult}, show that the Mult approach is an order of magnitude better in terms of T count at the expense of less than a hundred of additional qubits.

\begin{figure}
    \centering
    \includegraphics[width=0.8\linewidth]{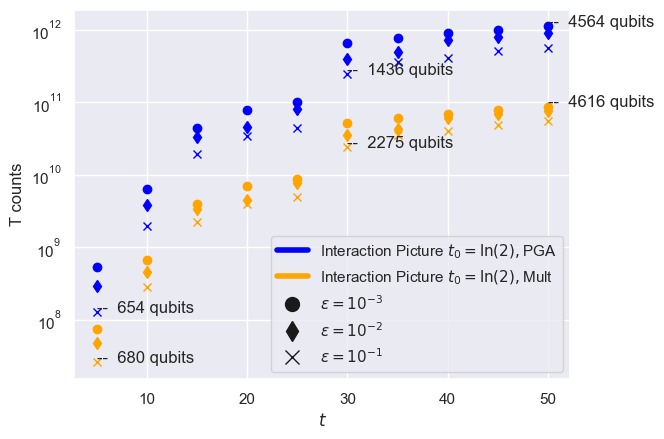}
    \caption{Resource cost comparison of two different compilations of the IP-based algorithm with $t_0 = \ln(2)$.
    The blue data points correspond to the compilations that use the PGA, and the orange data points correspond to the Mult approach.
    The comparison is shown in terms of T counts and qubit counts for time parameters $t \in (t_\text{min}, 10t_\text{min})$, and accuracies $\epsilon \in \{0.001, 0.01, 0.1\}$. We fix $N_0 = 8$, $x = 0.1$, $\mu = 1$, $\Lambda_0 = \sqrt{10}\,\mu = \sqrt{10}$, and $t_\text{min} = 0.5/x$. The parameters $N$ and $\Lambda$ are chosen according to Tab.~\ref{tab:Parameters}.}
    \label{fig:pga_vs_mult}
\end{figure}

Third, in Fig.~\ref{fig:trotter_vs_IP}, we compare the $\textrm{PF}_2$ vs. IP-based algorithm with the optimal implementation choices, e.g., for $t_0=\ln 2$ and the Mult approach.
We perform this comparison for different coupling strengths $x= \{0.1, 1, 10, 100\}$ and for $\epsilon = \{10^{-3},10^{-2},10^{-1}\}$, over $\{t_{\min}, 2t_{\min}, \ldots, 10 t_{\min}\}$ where $t_{\min}$ is an optimistic estimate for the shortest time for which we expect to observe the Schwinger effect.
Remark that the parameter regimes for different values of $x$, such as $t_{\min}$, are different, and tailored to optimistic estimates for observing the Schwinger effect.
For example, for large $x$ we expect a shorter time to observe the Schwinger effect given that is the strength of the particle-antiparticle creation term, hence a shorter $t_{\min}$ is chosen, e.g., as in Table~\ref{tab:Parameters}.
It is however not guaranteed that $t_{\min}$, or even $10 t_{\min}$, is sufficient to observe the phenomenon and one may need to carry out the simulation for longer times.
With this remark in mind, we recover the expected performances from the asymptotic formula given in Table~\ref{tab:ComparisonOfAsymptoticCosts}.
The IP-based algorithm is more advantageous for smaller $x$, longer times, and high accuracy simulations, compared to the $\textrm{PF}_2$ approach.

\begin{figure*}
    \centering
    \begin{subfigure}[t]{0.5\linewidth}
        \centering
        \includegraphics[width=.9\linewidth]{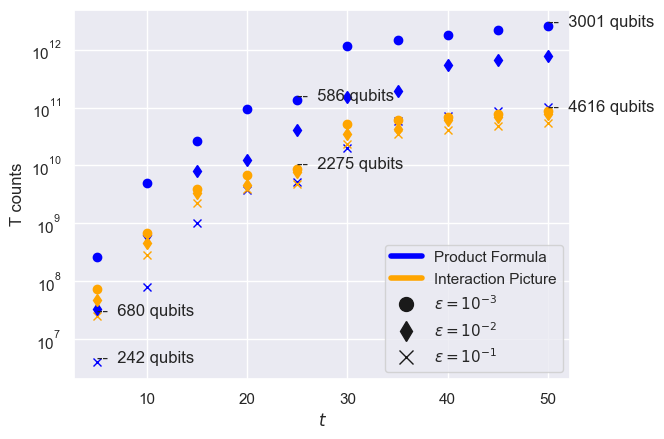}
        \caption{For $x=0.1$.}
    \end{subfigure}%
    ~ 
    \begin{subfigure}[t]{0.5\linewidth}
        \centering
        \includegraphics[width=.9\linewidth]{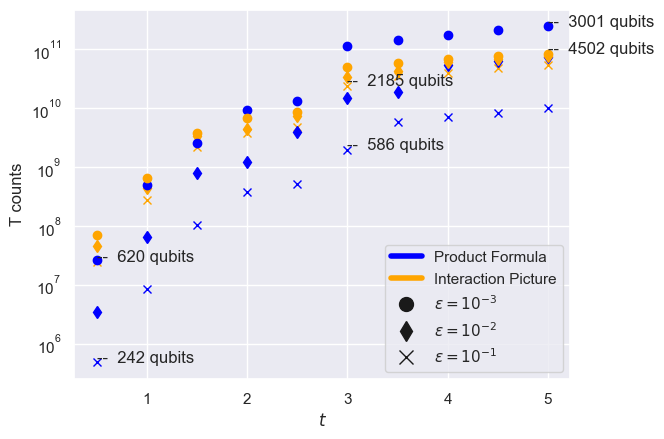}
        \caption{For $x=1$.}
    \end{subfigure}\\
    \centering
    \begin{subfigure}[t]{0.5\linewidth}
        \centering
        \includegraphics[width=.9\linewidth]{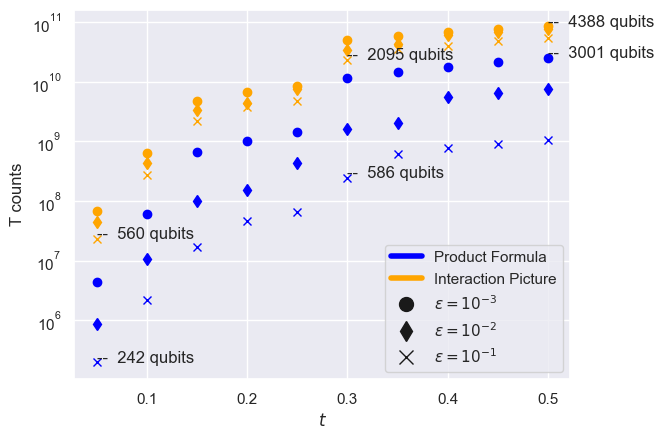}
        \caption{For $x=10$.}
    \end{subfigure}%
    ~ 
    \begin{subfigure}[t]{0.5\linewidth}
        \centering
        \includegraphics[width=.9\linewidth]{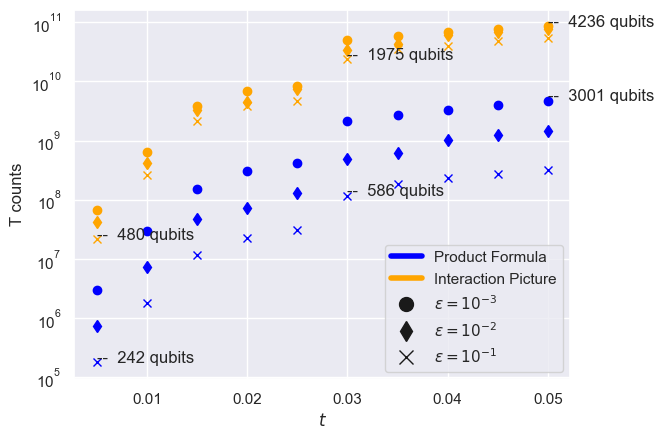}
        \caption{For $x=100$.}
    \end{subfigure}
    \caption{Resource cost comparison between the best implementations of IP-based algorithm and $\textrm{PF}_2$ approach. The comparison is given in terms of T-counts and qubit counts, for different time parameters $t \in (t_\text{min}, 10t_\text{min})$, and for accuracy $\epsilon= \{0.001, 0.01, 0.1\}$. We fix $N_0=8$, $\mu =1$, $\Lambda_0 = \sqrt{10}\mu = \sqrt{10}$, and $t_\text{min} = 0.5/x$. The parameters $N$ and $\Lambda$ are chosen according to Tab.~\ref{tab:Parameters}.}
    \label{fig:trotter_vs_IP}
\end{figure*}

It is clear that the $\textrm{PF}_2$ approach performs better for short simulation times, and when high accuracy is not required. 
While the simulations in this regime might not detect the Schwinger effect, they are the starting point for a set of numerical experiments.
In fact, for an early fault-tolerant quantum computer which is limited in not only the number of operations but also the number of qubits, the $\textrm{PF}_2$ approach is clearly favorable. Thus far, we have considered rigorous cutoffs that scale with the simulation time $t$, as we have explained in Sec.~\ref{sec:Setting}. However, we do not preclude the emergence of evidence, e.g., Ref.~\cite{ciavarella2025truncation}, which indicates that small cutoffs, possibly independent of $t$, may suffice for measuring certain observables. Motivated by this, we plot the resource costs for small, time-independent cutoffs in Fig.~\ref{fig:trotter_t_vs_qbts}.
These optimistic quantum resource estimates assume a simulation time $t = t_{\min} \propto 1/x$ for which the Schwinger effect may be observed though is not guaranteed.
This leads to quantum resource estimates that are higher in the weakly interacting regime than in the strongly interacting one, as shown in Fig.~\ref{fig:trotter_t_vs_qbts}.
Note that this will not be the case for a fixed time $t$, and will likely not be the case if one runs these simulations with varying time for probing the Schwinger effect or extracting another observable.

\begin{figure}
    \centering
    \includegraphics[width=0.8\linewidth]{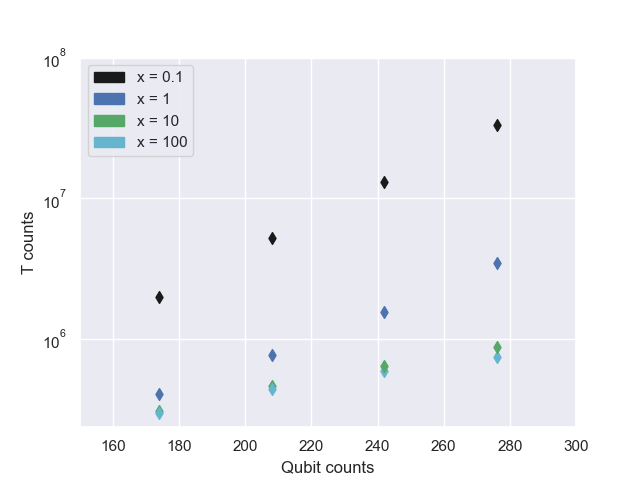}
    \caption{Resource cost of simulating the time evolution of the Schwinger model with a second-order product-formula method, for $x \in \{0.1, 1, 10, 100\}$ (indicated by the colors) and time-independent $\eta \in \{2, 3, 4, 5\}$ (from the smallest to the largest number of qubits).  We fix $\mu=1$, $N_0=8$ and $t = t_\text{min}= 0.5/x$. The number of fermions $N$ is chosen according to Tab.~\ref{tab:Parameters}.}
    \label{fig:trotter_t_vs_qbts}
\end{figure}

\section{Conclusions and Outlook}\label{sec:Conclusion}

In this work, we studied two algorithms and their various implementations for simulating the dynamics in Schwinger model.
Our findings demonstrate that while asymptotically the interaction picture based algorithm gives rise to the best performance, product formulas perform better for parameter regimes that lead to smaller electric field cutoffs and less stringent accuracy requirements.
We also studied the inter-dependency of the simulation parameters (such as $N$, $\Lambda$, $t_{\min}$, and optimal algorithm parameters) to the fundamental parameters of the model (such as the mass $\mu$ and the interaction strength $x$).

There are various directions to follow, both for increasing the performance of the algorithms and for extending the applications to more complicated models in the same or higher dimensions.
Recent works proved that quantum time evolution can be simulated more efficiently with product formula~\cite{csahinouglu2021hamiltonian, hejazi2024better, mizuta2025trotterization, tran2021faster} and quantum signal processing~\cite{zlokapa2024hamiltonian} when certain subspaces are considered, such as when the initial state is guaranteed to be in the low-energy subspace (with respect to the spectral norm of the Hamiltonian).
Product formulae approaches have been observed to perform better than the worst-case bounds~\cite{childs2018toward,chen2024average}, hence one may in fact hope for a better performance from those methods.
Furthermore, one can consider transforming the Hamiltonian to an equivalent one where the gauge degrees of freedom are eliminated.
This reduces the number of qubits, at the expense of more complicated, non-local Hamiltonian terms, yet it may still be more efficient to implement for early fault-tolerant quantum computers or quantum simulators~\cite{martinez2016real, nguyen2022digital, sakamoto2023endtoendcomplexitysimulatingschwinger, farrell2024quantum, cochran2024visualizing}. 
However, this strategy is less extensible beyond one spatial dimensions, where the number of constraints, i.e., Gauss' law, is strictly less than the number of gauge degrees of freedom and thus the gauge fields cannot be completely eliminated. Similar to our work, an interaction-picture algorithm based on Dyson series was applied to $U(1), SU(2)$ and $SU(3)$ lattice gauge theories in beyond $1+1$ spacetime dimensions in Ref.~\cite{rhodes2024exponential}. Compared to our algorithm, the LCU-based algorithm for U(1) lattice gauge theories from~\cite{rhodes2024exponential}, if applied to the Schwinger model, will lead a different resource estimates mainly because: (i) $H_M$ is not fast-forwarded, though it is ostensibly diagonal, and instead, it is block-encoded, thereby introducing an additional $\mu$-dependence in the gate complexity. (ii) The fermions are encoded using mappings that introduce an additive $\mathcal{O}(N)$ ancilla overhead compared to our choice of Jordan-Wigner mapping. (iii) In~\cite{rhodes2024exponential}, the authors optionally apply the HHKL algorithm~\cite{haah2021quantum} that provides an asymptotic advantage, which in practice, only manifests at large system sizes; the precise break-even point in system size for the Schwinger model, while interesting is beyond the scope of this work. Importantly, the improvements we made in the interaction-picture algorithm here, which are model-agnostic, e.g., better error bounds and more optimized circuit implementation, can be readily applied to improve the algorithms in Ref.~\cite{rhodes2024exponential}. In the future, it will be interesting to quantify such improvements in the context of non-Abelian theories, particularly QCD.

\subsection*{Author contributions and acknowledgments}

Authors and contributions are listed alphabetically. 
AK, JL and B\c{S} performed the technical analysis, conducted resource counts, and drafted the manuscript. 
JL and OO implemented the scripts that validate parts of the algorithm implementations. 
All authors contributed to discussions and code review. 
Authors acknowledge the support of Sukin Sim and William Pol with OO's internship and resource counts validation.
Additionally, the authors thank Kaiwen Gui, Mark Steudtner, and Sam Pallister for reviewing the manuscript, and are thankful for collaboration with our colleagues at PsiQuantum on other related topics.

\bibliographystyle{unsrt}
\bibliography{arxiv}

\newpage

\onecolumngrid

\section{Mapping the Hamiltonian to Qubits: Hamiltonian as linear combinations of Pauli-products}\label{sec:MappingHamiltonianToQubits}

In this section, we introduce the Schwinger model again, and perform the explicit Jordan-Wigner mapping (fermion to qubit mapping), as could also be found in the literature~\cite{shaw2020quantum} which we closely follow here.
The Hamiltonian consists of three terms, the electric field term $H_E$, the gauge-matter interaction term $H_{I}$,and the mass (or matter) term $H_M$.  
The Hamiltonian is a sum of three terms, electric, interaction and mass terms, as
\begin{align}
H= H_{E} + H_{I} +  H_{M},
\end{align}
where
\begin{align}
H_E&= \sum^N_{r=1} E^2_r,\\
H_I & = x \sum_{r} (U_r \psi^\dagger_r \psi_{r+1} - U^\dagger_r\psi_r\psi_{r+1}^\dagger),\\
H_M & = \mu \sum_r (-1)^r \psi^\dagger_r \psi_{r}. 
\end{align}
\noindent Note that $r$ denotes the lattice site, $E_r$ and $U_r$ act on the $2\Lambda$ dimensional qudit placed on the link $(r, r+1)$. 
The Kogut-Susskind construction~\cite{kogut1975hamiltonian} is such that electrons live on even numbered sites, whereas positrons live on odd numbered sites.
On even sites $r$, $\psi_r, \psi^{\dagger}_r$ are annihilation/creation operators.
On odd sites $r$, $\psi_r, \psi^{\dagger}_r$ are creation/annihilation operators.
The electric term, $H_E$, acts like a potential term that energetically penalizes higher electric field values.
The interaction term, $H_I$, can create/annihilate electron-positron pairs on neighboring sites by incrementing/decrementing the electric field value.
The mass term, $H_M$, penalizes states with higher numbers of fermions, proportional to the bare mass $\mu$ of each particle.
$x$ and $\mu$ are the unitless coupling and mass parameters, given in terms of the physical parameters of the system as in Eq.~\eqref{eq:ModelParametersFromPhysicalParamaters}.

After truncating the Hilbert space, the electric field value is considered to be an element of $\mathbb{Z}_{2 \Lambda}$, i.e., it periodically wraps the electric field at a cutoff $\Lambda$, meaning the Hamiltonian term $H_I$ has off-diagonal elements at $\ketbra{\Lambda}{\Lambda - 1} + h.c.$: 
\begin{align}
H_E&= \sum^N_{r=1} \sum^{\Lambda-1}_{\varepsilon = -\Lambda} \varepsilon^2 \ketbra{\varepsilon}_r,\\
\nonumber H_I & = x \sum^N_{r=1} \bigg\{\ketbra{-\Lambda}{\Lambda-1}_r \otimes \psi^\dagger_r \psi_{r+1} - \ketbra{\Lambda-2}{ \Lambda-1}_r \otimes \psi_{r} \psi^\dagger_{r+1} + \ketbra{-\Lambda+1}{-\Lambda}_r \otimes \psi^\dagger_r \psi_{r+1} \\
& \hspace{2cm}- \ketbra{\Lambda-1}{-\Lambda}_r \otimes \psi_{r} \psi^\dagger_{r+1}  + \sum^{\Lambda-2}_{\varepsilon = -\Lambda+1} (\ketbra{\varepsilon+1}{\varepsilon}_r \otimes \psi^\dagger_r \psi_{r+1} - \ketbra{\varepsilon-1}{\varepsilon}  \otimes \psi_{r} \psi^\dagger_{r+1})\bigg\},\\
H_M & = \mu \sum^N_{r=1} (-1)^r \psi^\dagger_r \psi_{r}. 
\end{align}

We then use the Jordan-Wigner transformation to represent the fermionic operators. For even sites, $|0\rangle$ denotes the vacuum state (i.e., the absence of an electron), while $|1\rangle$ denotes the occupied state (i.e., the presence of an electron), while for odd sites the convention is reversed.
The Jordan-Wigner transform reads

\begin{align}\label{eq:JWdef}
\psi^\dagger_r = \frac{X_r-iY_r}{2} \prod_{j=1}^{r-1} Z_j
\end{align}
which leads to
\begin{align}\label{eq:JW1}
    \psi^\dagger_r\psi_{r+1} & =  \frac{1}{4}(-iY_rX_{r+1}+Y_rY_{r+1}+X_rX_{r+1}+iX_rY_{r+1})
\end{align}
and
\begin{align}\label{eq:JW2}
    \psi_r\psi^\dagger_{r+1} & =  \frac{1}{4}(-iY_rX_{r+1}-Y_rY_{r+1}-X_rX_{r+1}+iX_rY_{r+1}).
\end{align}
Hence, combining Eqs.~\eqref{eq:JW1} and~\eqref{eq:JW2} we obtain

\begin{align}
\label{eq:LCUHint}   H_I & = \frac{x}{4} \sum_{r=1}^{N} \bigg\{ \left(U_r+U_r^\dagger\right)\left(X_rX_{r+1}+Y_rY_{r+1}\right) + i\left(U_r - U^\dagger_r\right)\left(X_rY_{r+1}-Y_rX_{r+1}\right)\bigg\}\\
    \nonumber & = \frac{x}{4} \sum^N_{r=1} \bigg\{\left(|-\Lambda\rangle\langle \Lambda-1|_r  + |\Lambda-2\rangle\langle \Lambda-1|_r  +|-\Lambda+1\rangle\langle  \Lambda-1|_r + |\Lambda-1\rangle\langle -\Lambda|_r \right)\otimes \left(X_rX_{r+1}+Y_rY_{r+1}\right)   \\
    & \hspace{0.5cm} +i \left(|-\Lambda\rangle\langle \Lambda-1|_r  - |\Lambda-2\rangle\langle \Lambda-1|_r  +|-\Lambda+1\rangle\langle  \Lambda-1|_r - |\Lambda-1\rangle\langle -\Lambda|_r \right)\otimes \left(X_rY_{r+1}-Y_rX_{r+1}\right) \\
\nonumber & \hspace{0.5cm}+ \sum^{\Lambda-2}_{\varepsilon = -\Lambda+1} \bigg[ (|\varepsilon+1\rangle \langle \varepsilon|_r + |\varepsilon-1\rangle\langle\varepsilon|)   \otimes (X_rX_{r+1}+Y_rY_{r+1}) + i(|\varepsilon+1\rangle \langle \varepsilon|_r + |\varepsilon-1\rangle\langle\varepsilon|) \otimes (X_rY_{r+1}-Y_rX_{r+1})  \bigg ]\bigg\}.
\end{align}
For the mass term, we get

\begin{align}
    \psi^\dagger_r\psi_{r} & = \frac{1}{2}(\mathbb{I} - Z_r).
\end{align}
N.B. this differs from~\cite{shaw2020quantum} by the identity term. 

\begin{align}\label{eq:JW_HM}
    H_M & = \frac{\mu}{2} \sum^N_{r=1} (-1)^r (\mathbb{I} - Z_r).
\end{align}
This completes the fermion/boson to qubit mappings that we use in order to simulate the system on a quantum computer.

\section{Overview of interaction picture simulation}\label{subsec:OverviewIP}

In this section, we first show the equivalence of time evolution operators in the Schr\"odinger and interaction pictures, given as in Lemma~\ref{lem:SchrodingerInteractionPictureTimeEvolutionRelation}. 
The nontrivial part of the time evolution operator, $U_I(t)$, is the time-ordered exponential operator as defined in Eq.~\eqref{eq:InteractionPictureUnitary}.
We then give two operators $D_{(K,M)}(t)$ and $\tilde{D}_{(K,M)}(t)$ each of which approximates $U_I(t)$, by truncating and discretizing the series expansion of the time-ordered integral, i.e., the well-known Dyson series.
These two approximations are given in Definition~\ref{def:TruncDiscDysonSeries} and Definition~\ref{def:TruncDiscDysonSeriesWithCol}, respectively.
The former follows the operator in Ref.~\cite{low2018hamiltonian}, which eliminates the overlapping time intervals in the discretization of the time-ordered intergrals, while the latter does not.
Then, Theorem~\ref{thm:ApproximatingInteractionPictureTimeEvolWithTruncatedDiscretizedDysonSeries} and Theorem~\ref{thm:ApproximatingInteractionPictureTimeEvolWithTruncatedDiscretizedDysonSeriesWithCollision} provide, respectively, the sufficient minimal values of $K$ and $M$ such that the truncated-discretized Dyson series $D_{(K,M)}(t)$ and $\tilde{D}_{(K,M)}(t)$ result in an $\epsilon$-additive approximation of the time-ordered evolution operator $U_I(t)$.
Combined with Lemma~\ref{lem:SchrodingerInteractionPictureTimeEvolutionRelation}, this implies Corollary~\ref{cor:TimeEvolutionSchrPicWithTruncDiscDysonSeries}, which gives an approximation of the time evolution operator $e^{-iHt}$ in terms of a product of $e^{-iH_0 t}$ and the truncated-discretized Dyson series.
A similar result holds for $\tilde{D}_{(K,M)}(t)$.
Using these general results, we set up a time-evolution circuit given by $W(t, \epsilon)$ that approximates $e^{-iHt}$ up to an error $\epsilon$.
Lemma~\ref{lem:CircuitImplementationOfDysonSeries} and its proof study the detailed quantum circuit for an implementation of $U_I(t_0/\alpha)$ and states the resource cost estimates.
The final result is given in Corollary~\ref{thm:Main}.

The quantum algorithm includes a couple of approximations employed when implementing the subroutines.
While, in principle, we could distribute the overall error budget equally to the error sources, this is suboptimal compared to an optimized, uneven distribution of errors, which we will adopt in our implementation. 
Below we give a table of errors allocated to its sources.
The results are stated in terms of these errors.

\begin{table}
    \centering
    \begin{tabular}{|c|c|c|}
    \hline
    Symbol & Source of the error & Determines\\
    \hline
    $\epsilon_1$ & Truncation error of Dyson series for time $t_0/\alpha$  & $K$\\
    \hline
    $\epsilon_2$ & Discretization error of Dyson series $t_0/\alpha$ & $M$ \\
    \hline
    $\epsilon_3$ & Synthesizing single qubit rotations & $b_{\rm{rot}}$ (Bit precision of rotation angles)  \\
    \hline
    $\epsilon_{\textrm{cutoff}}$ & Electric field cutoff & $\Lambda$  \\
    \hline
    \end{tabular}
    \caption{Error symbols, their source, and what they determine.}
    \label{tab:ListOfTheErrors}
\end{table}

The total error is then given by
\begin{align}
\epsilon= \epsilon_{\textrm{cutoff}} +  \tilde{\epsilon}_1 + \tilde{\epsilon}_2 +\tilde{\epsilon}_3 = \epsilon_{\textrm{cutoff}} + \epsilon',
\end{align}
where
\begin{align}
\tilde{\epsilon}_1 = \mathsf{r} \epsilon_1, \: \tilde{\epsilon}_2= \mathsf{r} \epsilon_2, \: \tilde{\epsilon}_3= n_{\text{rot}} \epsilon_3.
\end{align}
Recall that $\mathsf{r}= \lceil t \alpha/t_0 \rceil$ is the number of times $U_I(t_0/\alpha))$ is called, and $n_{\text{rot}}$ is the number of single-qubit rotations.
Fixing these errors determines the parameters that feed into the description of the quantum algorithm, as summarized in Table~\ref{tab:ListOfTheErrors}. Dyson was the first to explicitly formalize the transformations connecting the Schr\"odinger and interaction pictures~\cite{Dyson1949Matrix}. The now-standard approach is summarized in Lemma~\ref{lem:SchrodingerInteractionPictureTimeEvolutionRelation}.

\begin{lemma}[Time evolution operator in the Schr\"odinger and in the interaction picture]\label{lem:SchrodingerInteractionPictureTimeEvolutionRelation}
Let $H=H_0 + V$ be a time-independent Hamiltonian with Hermitian $H_0$ and $V$. Then the time evolution operator in the Schr\"odinger picture, i.e.,  the operator that evolves $\ket{\psi(0)}$ to $\ket{\psi(t)}$, is 
\begin{align}
U(t)= e^{-iHt}
\end{align}
whereas the time evolution in the interaction picture, i.e., the operator that evolves $\ket{\psi(0)}_I$ to $\ket{\psi(t)}_I$, is
\begin{align}\label{eq:InteractionPictureUnitary}
U_I(t)= \mathcal{T}\left[ e^{-i\int^t_0 V(s) ds} \right]
\end{align}
where $\ket{\psi(t)}_I:= e^{iH_0t}\ket{\psi}$ for any given $\ket{\psi} \in \mathcal{H}$ and $t \in \mathbb{R}$, and $V(s):= e^{iH_0 s} V e^{-iH_0 s}$ for any $s \in \mathbb{R}$. $\mathcal{T}$ denotes the time-ordering operator.
The time evolution operator in the Schr\"odinger and interaction pictures are related as follows:

\begin{align}\label{eq:SchrodingerInteractionPictureTimeEvolutionRelation}
U(t)= e^{-iH_0 t} U_{I}(t),
\end{align}
in the sense that
\begin{align}
|\psi(t)\rangle= U(t) |\psi(t=0)\rangle= e^{-iH_0 t} U_{I}(t)|\psi(t=0)\rangle_I,
\end{align}
\end{lemma}

\begin{proof}
We start with the time-dependent Schr\"odinger equation:
\begin{align}
\frac{d}{dt} \ket{\psi(t)}= -iH \ket{\psi(t)}.
\end{align}
Note that it is immediate that the time-independent Hamiltonian leads to a time-evolution operator $U(t)= e^{-iHt}$.
For the rest of the claim, substitute $\ket{\psi(t)}= e^{-iH_0 t} \ket{\psi(t)}_I$.
We obtain
\begin{align}
\frac{d}{dt} \ket{\psi(t)}_I= -i e^{iH_0 t} V e^{-iH_0 t} \ket{\psi(t)}_I.
\end{align}
This implies that the time-ordered exponential $U_I(t):= \mathcal{T}\left[ e^{-i\int^t_0 V(s) ds} \right]$ is the time-evolution operator in the interaction picture.
Namely,
\begin{align}
e^{-iH_0 t}\ket{\psi(0)}= \ket{\psi(t)}= e^{-iH_0 t}U_{I}(t)\ket{\psi(0)}_I.
\end{align}
\end{proof}

Next, we define two Dyson series that are truncated and discretized, both of which approximate the operator $U_I(t)$ given as in Lemma~\ref{lem:SchrodingerInteractionPictureTimeEvolutionRelation}.
These are the operators that we use in the quantum algorithm. Following the notation introduced in~\cite{low2018hamiltonian}, we defined the truncated-discretized Dyson series without collisions in Definition~\ref{def:TruncDiscDysonSeries} and with collisions in Definition~\ref{def:TruncDiscDysonSeriesWithCol}:

\begin{definition}[$D_{(K,M)}(t)$: $(K,M)$-truncated-discretized Dyson series]\label{def:TruncDiscDysonSeries} Let $H= H_0 + V$ be a Hamiltonian of a pair of Hermitian operators $H_0, V$. 
Let $t \in \mathbb{R}$, $K \in \mathbb{N}^+$ be the truncation parameter and $M \in \mathbb{N}^+$ be the time discretization parameter.
The $(K,M)$-truncated-discretized Dyson series is defined as

\begin{align}\label{eq:TruncDiscDysonSeries1}
D_{(K,M)}(t):= \sum^K_{k=0} (-i\Delta)^k B_{(k,M)}(t)
\end{align}
and
\begin{align}\label{eq:TruncDiscDysonSeries2}
B_{(k,M)}(t):=\sum^{M-1}_{m_k=k-1} \ldots \sum^{m_3-1}_{m_2=1} \sum^{m_2-1}_{m_1=0} V(m_k \Delta) \ldots V(m_2 \Delta) V(m_1 \Delta),
\end{align}
where $\Delta:= t/M$, and $V(s):= e^{iH_0s} V e^{-iH_0s}$ for $s \in \mathbb{R}$.
\end{definition}

Note that above definition excludes the \emph{collisions}, i.e., the cases where there is at least one $i$ such that $m_i= m_{i+1}$ in the expansion that defines the operators $B_{(k,M)}(t)$. 
Alternatively, one could include the collisions, and use an alternative discretization defined as follows.
We will see that in terms of approximating the original operator, including/excluding the collision terms does not matter much.
This is related to the fact that there are overwhelmingly more cases where there is no collision compared to the number of cases with collision.

\begin{definition}[$\tilde D_{(K,M)}(t)$: $(K,M)$-truncated-discretized Dyson series with collisions]\label{def:TruncDiscDysonSeriesWithCol} Let $H= H_0 + V$ be a Hamiltonian of a pair of Hermitian operators $H_0, V$.
Let $t \in \mathbb{R}$, $K \in \mathbb{N}^+$ be the truncation parameter and $M \in \mathbb{N}^+$ be the time discretization parameter.
The $(K,M)$-truncated-discretized Dyson series is defined as

\begin{align}\label{eq:TruncDiscDysonSeries3}
\tilde D_{(K,M)}(t):= \sum^K_{k=0} (-i\Delta)^k \tilde B_{(k,M)}(t)
\end{align}
and
\begin{align}\label{eq:TruncDiscDysonSeries4}
\tilde B_{(k,M)}(t):=\sum^{M-1}_{m_k=0} \ldots \sum^{m_3}_{m_2=0} \sum^{m_2}_{m_1=0} V(m_k \Delta) \ldots V(m_2 \Delta) V(m_1 \Delta),
\end{align}
where $\Delta:= t/M$, and $V(s):= e^{iH_0s} V e^{-iH_0s}$ for $s \in \mathbb{R}$.
\end{definition}

We then employ the following Theorem~\ref{thm:ApproximatingInteractionPictureTimeEvolWithTruncatedDiscretizedDysonSeries} and Theorem~\ref{thm:ApproximatingInteractionPictureTimeEvolWithTruncatedDiscretizedDysonSeriesWithCollision} to find the minimal truncation and discretization parameters $K, M$, respectively, in order to guarantee an error $\epsilon$ in approximating the time evolution operator $U_I$ in the interaction picture.

\begin{theorem}[Approximating time evolution in the interaction picture with truncated-discretized Dyson series]\label{thm:ApproximatingInteractionPictureTimeEvolWithTruncatedDiscretizedDysonSeries} Let $H= H_0 + V$ be a Hamiltonian of a pair of Hermitian operators $H_0, V$. 
Let $D_{(K,M)}(t)$, $\tilde{D}_{(K,\tilde{M})}(t)$ be the $(K,M)$-truncated-discretized Dyson series without/with collisions defined as in Definition~\ref{def:TruncDiscDysonSeries} and Definition~\ref{def:TruncDiscDysonSeriesWithCol} respectively, where $t \in \mathbb{R}$, $K \in \mathbb{N}^+$ and $M \in \mathbb{N}^+$, and let $\epsilon, \epsilon_1, \epsilon_2 \in \mathds{R}^+$ be such that $\epsilon= \epsilon_1 + \epsilon_2$.
Let $\mathcal{T}\left[ e^{-i\int^t_0 V(s) ds} \right]$ be the time evolution in the interaction picture defined in Lemma~\ref{lem:SchrodingerInteractionPictureTimeEvolutionRelation}. Then,

\begin{align}\label{thm:TruncDiscDysonApproximationError1}
\left\|\mathcal{T}\left[ e^{-i\int^t_0 V(s) ds} \right] - D_{(K,M)}(t) \right\| \leq \epsilon, \qquad \left\|\mathcal{T}\left[ e^{-i\int^t_0 V(s) ds} \right] - \tilde{D}_{(K,\tilde{M})}(t) \right\| \leq \epsilon
\end{align}
for all
\begin{align}\label{thm:TruncDiscDysonApproximationErrorConditionK}
K \geq -1 + \frac{\ln 1/\epsilon_1}{W\left(\frac{\ln 1/\epsilon_1}{te\max_s \|H(s)\|}\right)}, \; \textrm{or more explicitly}, \; K\geq \max \left\{2 \|V\| t, e\|V\|t + \ln (1/\epsilon_1) \right\},
\end{align}
where $W(\cdot)$ is the Lambert W function, and for all
\begin{align}\label{thm:TruncDiscDysonApproximationErrorConditionM}
M \geq \max\left\{2t\|H_0\|, \frac{(K-1)^2}{\ln 2}, \frac{2t^2 \|V\|e^{t\|V\|}( \|H_0\|  + 2\|V\|)}{\epsilon_2} \right\}, \; 
\tilde{M} \geq \max \left\{ 2t\|H_0\|,\frac{(K-1)^2}{\ln 2}, \frac{6t^2 \|H_0\| \|V\|e^{t\|V\|}}{\epsilon_2}\right\},
\end{align}
of the Dyson series without and with collisions, respectively.
\end{theorem}

\textit{Proof of Theorem~\ref{thm:ApproximatingInteractionPictureTimeEvolWithTruncatedDiscretizedDysonSeries}.}

We prove the theorem in two parts by first approximating the series expansion of $\mathcal{T}\left[ e^{-i\int^t_0 V(s) ds}\right]$ with a finite series (hence the truncation parameter $K$ enters), and then discretizing the integrals in the series expansion with a finite sum (hence the discretization parameter $M$ enters).
More precisely, we bound the terms appearing in the RHS of

\begin{align}\label{eq:ErrorDysonInProof}
\left\|\mathcal{T}\left[ e^{-i\int^t_0 V(s) ds} \right] - D_{(K,M)}(t) \right\| & \leq \left\| \mathcal{T}\left[ e^{-i\int^t_0 V(s) ds} \right] - D_K(t) \right\| + \left\| D_K(t) - D_{(K,M)}(t)\right\|
\end{align}
where $D_K(t)$ is the truncated Dyson series defined as

\begin{align}\label{eq:DefinitionOfTruncatedDysonSeries}
D_K(t):= \sum^K_{k=0} (-i)^k B_k(t), \textrm{where} \:\; B_k(t):= \int^t_{t_k} dt_k \ldots \int^{t_3}_{t_2} dt_2 \int^{t_2}_0 dt_1 V(t_k) \ldots V(t_2) V(t_1).
\end{align}
To complete the proof, we first prove Lemma~\ref{lem:truncerror} that bounds the first term in Eq.~\eqref{eq:ErrorDysonInProof}, i.e., the truncation error.
Then we demonstrate Lemma~\ref{lem:discretization_error} and Lemma~\ref{lem:discretization_error_with_col}, without and with collisions, respectively, that bounds the second term in Eq.~\eqref{eq:ErrorDysonInProof}, i.e., the discretization error.

\textbf{Remark:} The derivation of the truncation bound $K$ in Lemma~\ref{lem:truncerror} closely follows the approach described in~\cite{low2018hamiltonian}. However, we provide an improved bound for the discretization parameter $M$ in the collision-free scenario in lemma~\ref{lem:discretization_error}. To the best of our knowledge, Lemma~\ref{lem:discretization_error_with_col} presents the first rigorous bound for the discretized Dyson series parameter 
$\tilde{M}$ that explicitly accounts for collisions.

\begin{lemma}[Approximation error of truncating the Dyson series]\label{lem:truncerror} The truncation error of the Dyson series at order $K$ satisfies the bound
\begin{align}
\left\| \mathcal{T}\left[ e^{-i\int^t_0 V(s) ds} \right] - D_K(t) \right\| \leq \epsilon_1
\end{align}
whenever
\begin{align}
K \geq \max \left\{2 \|V\| t,-1 + \frac{\ln 1/\epsilon_1}{W\left(\frac{\ln 1/\epsilon_1}{te \|V\|}\right)}\right\}, \; \textrm{or more explicitly}, \; K\geq \max \left\{2 \|V\| t, e\|V\|t + \ln (1/\epsilon_1) \right\}.
\end{align}
\end{lemma}

\textit{Proof of Lemma~\ref{lem:truncerror}.} Following the definition, 

\begin{align}\label{eq:proof_truncation}
\left|\left|  \mathcal{T}\left[ e^{-i\int_0^t V(s)ds} \right] - \sum_{k=0}^K(-i)^kB_k(t) \right|\right| & \leq \sum_{k=K+1}^{\infty}  \left|\left|  B_k(t) \right|\right| \\ 
& =\sum_{k=K+1}^{\infty} \frac{1}{k!}\left\|\int_0^t...\int_0^t \mathcal{T}\left[V(t_1)...V(t_k) \right]dt_1...dt_k \right \| \\
& \leq \sum_{k=K+1}^{\infty}  \frac{1}{k!}\left  \|\int_0^t...\int_0^t \prod_{j=1}^k \|V(t_j) \| dt_1...dt_k \right \| \\
& \leq \sum_{k=K+1}^{\infty} \frac{(t\max_s ||V(s)||)^k}{k!} \\
& = \frac{(t\max_s ||V(s)||)^{K+1}}{(K+1)!}\sum_{k=K+1}^{\infty} \frac{(t\max_s ||V(s)||)^{k-K-1}}{(K+2)...(k-1)k} \\
& \leq \frac{(t\max_s ||V(s)||)^{K+1}}{(K+1)!}\sum_{k=K+1}^{\infty} \left(\frac{t\max_s ||V(s)||}{K+2}\right)^{k-K-1} \\
& \leq \frac{(t\max_s ||V(s)||)^{K+1}}{(K+1)!}\sum_{k=K+1}^{\infty} \left(\frac{K}{2(K+2)}\right)^{k-K-1} \\
& = \frac{(t\max_s ||V(s)||)^{K+1}}{(K+1)!}\sum_{k=0}^{\infty} \left(\frac{K}{2(K+2)}\right)^{k} \\
& \leq \frac{(t\max_s ||V(s)||)^{K+1}}{(K+1)!}\sum_{k=0}^{\infty} \left(\frac{1}{2}\right)^{k} \\
&=2\frac{(t\max_s ||V(s)||)^{K+1}}{(K+1)!} \\
&\leq \left(\frac{te\max_s ||V(s)||}{K+1}\right)^{K+1}\frac{2}{\sqrt{2\pi(K+1)}e^{1/(12(K+1)+1)}}\\
 &\leq \left(\frac{te\max_s ||V(s)||}{K+1}\right)^{K+1}
\end{align}

Note that $\|V(s)\|= \|e^{-iH_0 s}Ve^{iH_0 s} \|= \|V\|$.
The exact solution to 
\begin{align}\label{eq:InequalityForK}
\left(\frac{e\|V\| t}{K+1} \right)^{K+1} \leq \epsilon_1
\end{align}
is given by the Lambert-W function, i.e., 
\begin{align}
K \geq  \frac{\ln{1/\epsilon_1}}{W\left(\frac{\ln 1/\epsilon_1}{te \|V\|}\right)}-1.
\end{align}

We can derive a bound on $K$ without using the Lambert W function as follows. 
Eq.\eqref{eq:InequalityForK} is satisfied if (note that we assume $t>0$ or otherwise use the absolute value of it, i.e., $|t|$)

\begin{align}
K \ln \left(\frac{K}{e\|V\|t}\right) \geq \ln (1/\epsilon_1).
\end{align}

Assume $K \geq e\|V\|t$ (such as choose $K= e\|V\|t + a\ln(1/\epsilon_1)$ where $e\|V\|t > a\ln(1/\epsilon_1)$). Replace $K= e\|V\|t + a\ln(1/\epsilon_1)$ to obtain

\begin{align}\label{eq:Bound0}
\left( e\|V\| t + a \ln(1/\epsilon_1) \right) \ln \left( 1 + \frac{a \ln(1/\epsilon_1)}{e\|V\|t}\right) \geq \ln(1/\epsilon_1).
\end{align}

The expansion $\log(1+x)= x - x^2/2 + x^3/3 - x^4/4 + \ldots$ indicates that $x> \log(1+x)$. Combined with the condition that $x < 1$ (equivalent to $e\|V\|t > a\ln (1/\epsilon_1)$), we observe that Eq.~\ref{eq:Bound0} implies

\begin{align}
\left( e\|V\|t + a\ln(1/\epsilon_1) \right) \frac{a \ln(1/\epsilon_1)}{e\|V\|t} \geq \ln(1/\epsilon_1),
\end{align}
which simplifies to
\begin{align}
a\ln(1/\epsilon_1) + \frac{a^2 (\ln(1/\epsilon_1))^2}{e\|V\|t} \geq \ln(1/\epsilon_1).
\end{align}

This is satisfied for any $a \geq 1$, hence $K= e\|V\|t + \ln(1/\epsilon_1)$ is one solution of the inequality with the condition that $e\|V\|t > \ln(1/\epsilon_1)$.

\qed

We now bound the second term of the RHS of Eq.~\eqref{eq:ErrorDysonInProof}, for the case without collisions.

\begin{lemma}[Approximation error of discretizing the truncated Dyson series]\label{lem:discretization_error}
Let $H=H_0 + V$ be a given Hamiltonian.
Let $D_K(t)$ and $D_{K,M}(t)$ be truncated and truncated-discretized Dyson series, with parameters $K$ and $(K,M)$, respectively.
Let $\epsilon_2 >0$. 
Then,
\begin{align}
\left\| D_K(t) - D_{(K,M)}(t)\right\| \leq \epsilon_2
\end{align}
is satisfied for all choices of
\begin{align}
M \geq \max\left\{2t\|H_0\|,\frac{(K-1)^2}{\ln 2}, \frac{2t^2 \|V\|( \|H_0\| e^{t\|V\|} + 2\|V\|e^{t\|V\|})}{\epsilon_2} \right\}.
\end{align}
\end{lemma}

\textit{Proof of Lemma~\ref{lem:discretization_error}.}
We use the definitions for $D_K(t)$ (Eq.~\eqref{eq:DefinitionOfTruncatedDysonSeries}) and $D_{(M,K)}(t)$ (Definition~\ref{def:TruncDiscDysonSeries}), and bound the approximation error term by term, i.e., 

\begin{align}
\left\| D_K(t) - D_{(K,M)}(t)\right\| & \leq \sum^K_{k=1}  \|B_k(t) - \Delta^k B_{(k,M)}(t)\|.
\end{align}

Since we calculate errors on the RHS term by term, denote $\Xi_k:= B_k(t) - \Delta^k B_{(k,M)}(t)$, for the sake of future notational ease. 
Below we bound $\|\Xi_k\|$.
To get an idea of how it works, let us start from the easiest case, $\Xi_1$, and proceed for higher values of $k$ from there.\\
{\bf{Bounding $\|\Xi_1\|$:}} Let's first write down the expressions explicitly:

\begin{align}
\Xi_1= \int^t_{0} dt_1 V(t_1) - \Delta \sum^{M-1}_{m_1=0} V(m_1 \Delta)
\end{align}
where $\Delta:= t/M$, and $V(s):= e^{+iH_0 s} V e^{-iH_0 s}$ for $s \in \mathbb{R}$.
We then partition the integral into $M$ segments, each of which is of equal length and obtain:
\begin{align}
\Xi_1 &= \sum^{M-1}_{m_1=0} \left[\int^{(m_1 + 1)\Delta}_{m_1 \Delta} dt_1 V(t_1) -  V(m_1 \Delta ) \Delta \right]\\
&= \sum^{M-1}_{m_1=0} \left[\int^{(m_1 + 1)\Delta}_{m_1 \Delta} dt_1 (V(t_1) -  V(m_1 \Delta )) \right]\\
&= \sum^{M-1}_{m_1=0} e^{iH_0 m_1 \Delta} \left[\int^{\Delta}_{0} dt_1 (V(t_1) -  V) \right]e^{-iH_0 m_1 \Delta}
\end{align}
where the second line follows straightforwardly, and the last line uses the definition of $V(s)$.

By the triangle inequality, we get the bound

\begin{align}\label{eq:x1_noCol}
\| \Xi_1 \| & \leq \sum^{M-1}_{m_1=0} \left\| \int^\Delta_{0} dt_1 (V(t_1) - V) \right\|\\
&= M \left\| \int^\Delta_{0} dt_1 (V(t_1) - V) \right\|.
\end{align}

Using the BCH formula we show that:

\begin{align}
V(s) - V = \sum^{\infty}_{n=1} \frac{s^n}{n!} [iH_0, [iH_0, [\ldots, [iH_0,V]\ldots]]].
\end{align}
For general $H_0$ and $V$, i.e., no assumptions are made on $H_0$ or $V$, the integral can be bounded by
\begin{align}\label{eq:CommBound}
\left\|\int^\Delta_{s=0} ds (V(s) - V) \right\|\leq \sum^\infty_{n=1} \frac{{|\Delta|}^{n+1}2^n}{(n+1)!} \|H_0\|^{n}\|V\|.
\end{align}
This bound could be improved considerably if we assumed local $H_0$ and $V$, by exploiting locality as in the Lieb-Robinson bounds, which can be interesting for a future work.
Assuming $\Delta \|H_0\|\leq 1/2$ (hence $M \geq 2t \|H_0\|$), we arrive at the following upper bound
\begin{align}
\left\|\int^\Delta_{s=0} ds (V(s) - V) \right\| \leq \Delta^2 \|H_0\| \|V\|  \sum^\infty_{n=1} (\Delta \|H_0\|)^{n-1} \leq 2 \Delta^2 \|H_0\| \|V\|.
\end{align}

{\bf{Bounding $\|\Xi_2\|$:}} This will be more complicated than the previous one, but it will be illuminating for bounding the higher order terms $\|\Xi_k\|$ for arbitrary $k > 2$.
Let's proceed in a similar way as we did for $\Xi_1$, and write down the expression for $\Xi_2$ explicitly:

\begin{align}\label{eq:X2_noCol}
\Xi_2= \int^t_{t_2} dt_2  V(t_2) \int^{t_2}_{t_1=0} dt_1 V(t_1) - \Delta^2 \sum^{M-1}_{m_2=1} V(m_2 \Delta) \sum^{m_2-1}_{m_1=0} V(m_1 \Delta).
\end{align}
Notice that partitioning the integral and then matching them with the corresponding discretized versions, we can rewrite $\Xi_2$ as follows:

\begin{align}
\nonumber \Xi_2&= \sum^{M-1}_{m_2=1} \int^\Delta_{0} dt_2  V(m_2 \Delta + t_2)
\sum^{m_2-1}_{m_1=0} \int^\Delta_{0} dt_1  V(m_1 \Delta + t_1) - \Delta^2 \sum^{M-1}_{m_2=1} V(m_2 \Delta) \sum^{m_2-1}_{m_1=0} V(m_1 \Delta)\\
& + \frac{1}{2}\sum^{M-1}_{m=0} \int^\Delta_{0} dt_1 \int^\Delta_{0} dt_2 V(m\Delta + t_1)V(m\Delta + t_2).
\end{align}
Observe that we can match the first and the second term on the RHS one by one, whereas the third term is not matched. More precisely,

\begin{align}
\nonumber \Xi_2&= \sum^{M-1}_{m_2=1}
\sum^{m_2-1}_{m_1=0} \left[ \int^{\Delta}_{0} dt_2 \int^{\Delta}_{0} dt_1 (V(m_2 \Delta + t_2)V(m_1 \Delta + t_1) - V(m_2 \Delta)V(m_1 \Delta))  \right]\\
& + \frac{1}{2}\sum^{M-1}_{m=0} \int^\Delta_{0} dt_1 \int^\Delta_{0} dt_2 V(m\Delta + t_1)V(m\Delta + t_2).
\end{align}
Using the fact that $(A_1 A_2 - B_1 B_2)= ((A_1 - B_1)B_2 + A_1 (A_2 - B_2))$ we get

\begin{align}
\nonumber \Xi_2&= \sum^{M-1}_{m_2=1}
\sum^{m_2-1}_{m_1=0} \left[ \int^{\Delta}_{0} dt_2 (V(m_2 \Delta + t_2) - V(m_2 \Delta )) \int^{\Delta}_{0} dt_1 V(m_1 \Delta)\right.\\ 
&\left. + \int^\Delta_{0} dt_2 V(m_2 \Delta + t_2) \int^\Delta_{0} dt_1 (V(m_1 \Delta + t_1) - V(m_1 \Delta ))  \right]\\
\nonumber & + \frac{1}{2} \sum^{M-1}_{m=0} \int^\Delta_{0} dt_1 \int^\Delta_{0} dt_2 V(m\Delta + t_1)V(m\Delta + t_2).
\end{align}

Taking the norm of above expression, applying triangle inequality and submultiplicative and unitary invariance property of the norm, we get

\begin{align}
\|\Xi_2\| &\leq 2 {M \choose 2} \left\| \int^\Delta_{0} dt (V(t) - V) \right\| \Delta \|V\| + \frac{1}{2}\sum^{M-1}_{m=0} \int^\Delta_{0} dt_1 \int^\Delta_{0} dt_2 V(m\Delta + t_1)V(m\Delta + t_2)\\
&\leq 2 {M \choose 2} \left\| \int^\Delta_{0} dt (V(t) - V) \right\| \Delta \|V\| + \frac{1}{2} M (\Delta \|V\|)^2.
\end{align}

{\bf{Bounding $\|\Xi_k\|$:}} Following the same lines of steps for $\Xi_2$, we find the following bound for $\|\Xi_k\|$ for any arbitrary $k \in \mathbb{N}^+$:

\begin{align}
\|\Xi_k\| \leq k {M \choose k} \left\| \int^\Delta_{0} dt (V(t) - V) \right\| (\Delta \|V\|)^{k-1} +
\frac{2k (\Delta \|V\|)^k M^{k-1}}{(k-1)!}
\end{align}
where the second term comes from bounding the volume of the collision cases, i.e.,  the $k$-tuples where there is at least one collision such as $m_i = m_j$ for some $i \neq j$.
The same analysis has been performed in Ref.~\cite{low2018hamiltonian} (see Eq.(A16)).

The bound on the total error will then be

\begin{align}\label{eq:LemmaA2-1}
\|D_K(t) - D_{(K,M)}(t)\| &\leq \sum^{K}_{k=1} \|\Xi_k\|\\
& \leq \left\| \int^\Delta_{0} dt (V(t) - V) \right\| \sum^K_{k=1} k {M \choose k} (\Delta \|V\|)^{k-1} + \sum^K_{k=2} \frac{2k (\Delta \|V\|)^k M^{k-1}}{(k-1)!}.
\end{align}

The sum in the first term can be bounded by 

\begin{align}\label{eq:LemmaA2-2}
\sum^K_{k=1} k {M \choose k} (\Delta \|V\|)^{k} & =  \sum^{K}_{k=1} k \frac{M!}{k!(M-k)!}\left(\frac{t \|V\|}{M}\right)^{k} \\
& \leq  \sum^{K}_{k=1} k \frac{t^k \|V\|^k}{k!} \\
& = t\|V\| \sum^{K}_{k=1} \frac{(t \|V\|)^{k-1}}{(k-1)!}\\
& \leq\|V\| \sum^{\infty}_{k=0} \frac{(t \|V\|)^{k}}{k!}\\
& = t \|V\|e^{t \|V\|}
\end{align}

The second term can be bounded by :
\begin{align}\label{eq:LemmaA2-3} 
\sum^K_{k=2} \frac{2k (\Delta \|V\|)^k M^{k-1}}{(k-1)!} & =  2\Delta \|V\| \sum_{k=2}^K \frac{k (\Delta \|V\|)^{k-1} M^{k-1}}{(k-1)!} \\
&\leq 4\Delta \|V\| \sum_{k=2}^K \frac{(\Delta \|V\|)^{k-1} M^{k-1}}{(k-2)!} \\
& = 4\Delta^2 \|V\|^2 M \sum_{k=2}^K \frac{(\Delta \|V\| M)^{k-2}}{(k-2)!} \\
&\leq  4\Delta^2 \|V\|^2 M e^{\Delta \|V\| M} \\
& =  \frac{4t^2 \|V\|^2}{M} e^{t \|V\| }
\end{align}

We put Eq.~\eqref{eq:LemmaA2-2} and Eq.~\eqref{eq:LemmaA2-3} into the Eq.~\eqref{eq:LemmaA2-1} and find

\begin{align}
\|D_K(t) - D_{(K,M)}(t)\| \leq M e^{et\|V\| +1} \left\| \int^\Delta_{0} dt (V(t) - V) \right\| + \frac{4t^2\|V\|^2}{M}e^{t\left|\left|V\right|\right|}.
\end{align}

Furthermore, assuming $\Delta \|H_0\| \leq 1/2$ implies
\begin{align}
\left\| \int^\Delta_{0} dt (V(t) - V) \right\| \leq  2 \Delta^2 \|H_0\| \|V\|.
\end{align}

Using this, we obtain

\begin{align}
\frac{2t^2 \|H_0\| \|V\| e^{t\|V\|} + 4 t^2 \|V\|^2 e^{t \|V\|} }{M} \leq \epsilon_2,
\end{align}

and thus $\|D_K(t) - D_{(K,M)}(t)\| \leq \epsilon_2$ is satisfied for all

\begin{align}
M \geq \frac{2t^2 \|V\|( \|H_0\| e^{t\|V\|} + 2\|V\|e^{t\|V\|})}{\epsilon_2}.
\end{align}

Together with the condition $\Delta \|H_0\| \leq 1/2$, the proof is complete.

\qed

\begin{lemma}[Approximating time evolution in the interaction picture with the truncated-discretized Dyson series with collisions]\label{lem:discretization_error_with_col}
Let $H=H_0 + V$ be a given Hamiltonian.
Let $D_K(t)$ and $\tilde{D}_{K,\tilde{M}}(t)$ be the truncated Dyson series and the truncated-discretized Dyson series with collisions, with parameters $K$ and $(K,\tilde{M})$, respectively.
Let $\epsilon_2 >0$. 
Then,
\begin{align}
\left\| D_K(t) - \tilde D_{(K,\tilde{M})}(t)\right\| \leq \epsilon_2
\end{align}
for all
\begin{align}
\tilde{M} \geq \max \left\{ 2t\|H_0\|,\frac{(K-1)^2}{\ln 2}, \frac{6t^2 \|H_0\| \|V\|e^{t\|V\|}}{\epsilon_2}\right\}
\end{align}
\end{lemma}
\begin{proof}

For notational ease, we use $M$ instead of $\tilde{M}$ in this proof. 

Following the same method as for the proof of Lemma~\ref{lem:discretization_error}, we want to bound $\|\Xi_k\|$. The first term, $\|\Xi_1\|$ is identical, so let's consider $\Xi_2$ first:

\begin{align}
\tilde{\Xi}_2 &= \int^t_{t_2} dt_2  V(t_2) \int^{t_2}_{t_1=0} dt_1 V(t_1) - \frac{\Delta^2}{2} \sum^{M-1}_{m_2=0}\sum^{M-1}_{m_1=0} \mathcal{T}\left (V(m_2 \Delta)  V(m_1 \Delta)\right ) \\
     & = \sum^{M-1}_{m_2=1} \int^\Delta_{0} dt_2  V(m_2 \Delta + t_2)
\sum^{m_2-1}_{m_1=0} \int^\Delta_{0} dt_1  V(m_1 \Delta + t_1) - \Delta^2 \sum^{M-1}_{m_2=1} V(m_2 \Delta) \sum^{m_2-1}_{m_1=0} V(m_1 \Delta)\\
\nonumber & + \frac{1}{2}\sum^{M-1}_{m=0} \int^\Delta_{0} dt_1 \int^\Delta_{0} dt_2 V(m\Delta + t_2)V(m\Delta + t_1) -\frac{\Delta^2}{2}\sum^{M-1}_{m=0}V(m\Delta)^2
\end{align}

\noindent Since the first two terms are exactly what we had previously, let's focus on the two last terms for now:

\begin{align}
    A_2 & = \frac{1}{2}\sum^{M-1}_{m=0} \int^\Delta_{0} dt_1 \int^\Delta_{0} dt_2 \left(V(m\Delta + t_2)V(m\Delta + t_1) -V(m\Delta)^2 \right) \\
    &= \frac{1}{2}\sum^{M-1}_{m=0} \int^\Delta_{0} dt_1 \int^\Delta_{0} dt_2 \Bigg( \bigg(V(m\Delta+t_2)-V(m\Delta)\bigg)V(m\Delta) +V(m\Delta+t_2)\bigg(V(m\Delta+t_1)-V(m\Delta)\bigg) \Bigg).
\end{align}

\noindent Taking the norm of the above expression, applying the triangle inequality and submultiplicative and unitary invariance property of the norm, we get

\begin{align}
\|A_2\| &\leq {M \choose 1} \left\| \int^\Delta_{0} dt (V(t) - V) \right\| \Delta \|V\|.
\end{align}
In general, we have : 

\begin{align}
\frac{\|A_k\|}{ \left\| \int^\Delta_{0} dt (V(t) - V) \right\|( \Delta \|V\|)^{k-1}} &\leq \sum^{k-1}_{q=1}{k - 1 \choose q} {M \choose {k-q}} \\
& =  \sum^{k-1}_{q=1}\frac{(k-1)!}{(k-1)!} \frac{(k-1)!}{q!(k-1-q)!} {M \choose {k-q}} \\
& = \frac{1}{(k-1)!} \sum^{k-1}_{q=1} \frac{(k-1)!^2}{q!(k-1-q)!} \frac{M!}{(k-q)!(M-k+q)!}\\
& \leq \frac{1}{(k-1)!} \sum^{k-1}_{q=1} \frac{(k-1)^{2q}M^{k-q}}{q!} \\
& = \frac{M^k}{(k-1)!} \sum^{k-1}_{q=1} \frac{1}{q!} \left(\frac{(k-1)^{2}}{M}\right)^q\\
& \leq \frac{M^k}{(k-1)!} \sum^{\infty}_{q=0} \frac{1}{q!} \left(\frac{(k-1)^{2}}{M}\right)^q \\ 
& = \frac{M^k}{(k-1)!} e^{(k-1)^2/M}
\end{align}

If we assume that $(K-1)^2/\ln 2 < M$, we obtain:

\begin{align}
\frac{\|A_k\|}{ \left\| \int^\Delta_{0} dt (V(t) - V) \right\| ( \Delta \|V\|)^{k-1}} &\leq  
\frac{2 M^k}{(k-1)!}
\end{align}

When $\Delta \|H_0\| \leq 1/2$, we have

\begin{align}
\|\tilde{\Xi}_k\| &\leq  \left(k {M \choose k}+\frac{ 2M^k}{(k-1)!} \right)\left\| \int^\Delta_{0} dt (V(t) - V) \right\| (\Delta \|V\|)^{k-1}\\
&\leq 2\left(k {M \choose k}+\frac{2M^k}{(k-1)!}\right) (\Delta \|V\|)^{k} \Delta \|H_0\|.
\end{align}

Now considering the summation over $k$, the first part can be bounded by 

\begin{align}
2\Delta \|H_0\| \sum^K_{k=1} k {M \choose k} (\Delta \|V\|)^{k} & =  2\Delta \|H_0\| \sum^{K}_{k=1} k \frac{M!}{k!(M-k)!}\left(\frac{t \|V\|}{M}\right)^{k} \\
& \leq 2\Delta \|H_0\| \sum^{K}_{k=1} k \frac{t^k \|V\|^k}{k!} \\
& = 2\Delta t \|H_0\| \|V\| \sum^{K}_{k=1} \frac{(t \|V\|)^{k-1}}{(k-1)!}\\
& \leq \frac{2t^2 \|H_0\| \|V\|}{M} \sum^{\infty}_{k=0} \frac{(t \|V\|)^{k}}{k!}\\
& = \frac{2t^2 \|H_0\| \|V\|}{M}e^{t \|V\|}
\end{align}

the second term can be bounded by :
\begin{align}\label{eq:LemmaA3-3}
4\Delta \|H_0\| \sum_{k=1}^K\frac{M^k}{(k-1)!}(\Delta \|V\|)^{k} &= 4\Delta \|H_0\| \sum_{k=1}^K\frac{(t \|V\|)^k}{(k-1)!} \\
&= \frac{4t^2 \|H_0\|\|V\|}{M} \sum_{k=1}^K\frac{(t \|V\|)^{k-1}}{(k-1)!} \\
&\leq \frac{4t^2 \|H_0\|\|V\|}{M} e^{t \|V\|}.
\end{align}

We then obtain that $\|D_K(t) - \tilde D_{(K,M)}(t)\| \leq \epsilon/2$ is satisfied for all

\begin{align}
M \geq \frac{6t^2 \|H_0\| \|V\|e^{t\|V\|}}{\epsilon}.
\end{align}

Together with the condition $\Delta \|H_0\| \leq 1/2$, the proof is complete.
\end{proof}

\begin{theorem}[Approximating time evolution in the interaction picture with truncated-discretized Dyson series without collisions]\label{thm:ApproximatingInteractionPictureTimeEvolWithTruncatedDiscretizedDysonSeriesWithoutCollision} Let $H= H_0 + V$ be a Hamiltonian of a pair of Hermitian operators $A, B$. Let $ D_{(K,{M})}(t)$ be the $(K,{M})$-truncated-discretized Dyson series defined as in Definition~\ref{def:TruncDiscDysonSeriesWithCol}, where $t \in \mathbb{R}$, $K \in \mathbb{N}^+$ and ${M} \in \mathbb{N}^+$, and let $\epsilon, \epsilon_1, \epsilon_2 \in \mathds{R}^+$ be such that $\epsilon= \epsilon_1 + \epsilon_2$.
Let $\mathcal{T}\left[ e^{-i\int^t_0 V(s) ds} \right]$ be the time evolution in interaction picture defined as inside of Lemma~\ref{lem:SchrodingerInteractionPictureTimeEvolutionRelation}. Then,

\begin{align}\label{thm:TruncDiscDysonApproximationError}
\left\|\mathcal{T}\left[ e^{-i\int^t_0 V(s) ds} \right] -D_{(K,{M})}(t) \right\| \leq \epsilon
\end{align}
for all
\begin{align}\label{thm:TruncDiscDysonApproximationErrorConditionKWithCollisions}
K \geq \max \left\{2 \|V\| t,-1 + \frac{\ln 1/\epsilon_1}{W\left(\frac{\ln 1/\epsilon_1}{te \|V\|}\right)}\right\}, \; \textrm{or more explicitly}, \; K\geq \max \left\{2 \|V\| t, e\|V\|t + \ln (1/\epsilon_1) \right\},
\end{align}
and for all
\begin{align}\label{thm:TruncDiscDysonApproximationWithoutCollisionErrorConditionM}
M \geq \max\left\{2t\|H_0\|, \frac{(K-1)^2}{\ln 2}, \frac{2t^2 \|V\|e^{t\|V\|}( \|H_0\|  + 2\|V\|)}{\epsilon_2} \right\}.
\end{align}

\end{theorem}
\begin{proof}
Follows immediately from Lemma~\ref{lem:truncerror} and Lemma~\ref{lem:discretization_error}.
\end{proof}

\begin{theorem}[Approximating time evolution in interaction picture with truncated-discretized Dyson series with collisions]\label{thm:ApproximatingInteractionPictureTimeEvolWithTruncatedDiscretizedDysonSeriesWithCollision} Let $H= H_0 + V$ be a Hamiltonian of sum of Hermitian operators $H_0, V$. Let $\tilde D_{(K,\tilde{M})}(t)$ be the $(K,\tilde{M})$-truncated-discretized Dyson series defined as in Definition~\ref{def:TruncDiscDysonSeriesWithCol}, where $t \in \mathbb{R}$, $K \in \mathbb{N}^+$ and $\tilde{M} \in \mathbb{N}^+$, and let $\epsilon, \epsilon_1, \epsilon_2 \in \mathds{R}^+$ be such that $\epsilon= \epsilon_1 + \epsilon_2$.
Let $\mathcal{T}\left[ e^{-i\int^t_0 V(s) ds} \right]$ be the time evolution in interaction picture defined as inside of Lemma~\ref{lem:SchrodingerInteractionPictureTimeEvolutionRelation}. Then,

\begin{align}\label{thm:TruncDiscDysonApproximationError2}
\left\|\mathcal{T}\left[ e^{-i\int^t_0 V(s) ds} \right] - \tilde D_{(K,\tilde{M})}(t) \right\| \leq \epsilon
\end{align}
for all
\begin{align}\label{thm:TruncDiscDysonApproximationWithCollisionErrorConditionK}
K \geq \max \left\{2 \|V\| t,-1 + \frac{\ln 1/\epsilon_1}{W\left(\frac{\ln 1/\epsilon_1}{te \|V\|}\right)}\right\}, \; \textrm{or more explicitly}, \; K\geq \max \left\{2 \|V\| t, e\|V\|t + \ln (1/\epsilon_1) \right\},
\end{align}
and for all
\begin{align}\label{thm:TruncDiscDysonApproximationWithCollisionErrorConditionM}
\tilde{M} \geq \max \left\{ 2t\|H_0\|,\frac{(K-1)^2}{\ln 2}, \frac{6t^2 \|H_0\| \|V\|e^{t\|V\|}}{\epsilon_2}\right\}.
\end{align}

\end{theorem}
\begin{proof}
Follows immediately from Lemma~\ref{lem:truncerror} and Lemma~\ref{lem:discretization_error_with_col}.
\end{proof}

We collect these results as a corollary:

\begin{corollary}[Approximating time evolution in Schr\"odinger picture with truncated-discretized Dyson series]\label{cor:TimeEvolutionSchrPicWithTruncDiscDysonSeries}
Let $H= H_0 + V$ be a time-independent Hamiltonian with Hermitian $H_0$ and $V$.
Let $D_{(K,M)}$ and $\tilde{D}_{(K,\tilde{M})}$ be given as in Definition~\ref{def:TruncDiscDysonSeries} and  Definition~\ref{def:TruncDiscDysonSeriesWithCol}, respectively.
Then, in the case without collisions
\begin{align}\label{eq:TimeEvolutionSchrPicWithTruncDiscDysonSeries1}
\left\| e^{-iHt} - e^{-iH_0 t}D_{(K,M)}(t) \right\| \leq \epsilon
\end{align}
for all
\begin{align}
K \geq \max \left\{2 \|V\| t,-1 + \frac{\ln 1/\epsilon_1}{W\left(\frac{\ln 1/\epsilon_1}{te \|V\|}\right)}\right\}\; \textrm{or more explicitly}, \; K\geq \max \left\{2 \|V\| t, e\|V\|t + \ln (1/\epsilon_1) \right\},
\end{align}
and 
\begin{align}
M \geq \max\left\{2t\|H_0\|, \frac{(K-1)^2}{\ln 2}, \frac{2t^2 \|V\|e^{t\|V\|}( \|H_0\|  + 2\|V\|)}{\epsilon_2} \right\}.
\end{align}
In the case with collisions,
\begin{align}\label{eq:TimeEvolutionSchrPicWithTruncDiscDysonSeries2}
\left\| e^{-iHt} - e^{-iH_0 t}\tilde{D}_{(K,\tilde{M})}(t) \right\| \leq \epsilon
\end{align}
is satisfied for all
\begin{align}
K \geq \max \left\{2 \|V\| t,-1 + \frac{\ln 1/\epsilon_1}{W\left(\frac{\ln 1/\epsilon_1}{te \|V\|}\right)}\right\}\; \textrm{or more explicitly}, \; K\geq \max\{ 2\|V\|t, e\|V\|t + \ln (1/\epsilon_1) \},
\end{align}
and 
\begin{align}
\tilde{M} \geq \max \left\{ 2t\|H_0\|,\frac{(K-1)^2}{\ln 2}, \frac{6t^2 \|H_0\| \|V\|e^{t\|V\|}}{\epsilon_2}\right\}.
\end{align}
\end{corollary}

\proof{
Follows immediately from Lemma~\ref{lem:SchrodingerInteractionPictureTimeEvolutionRelation} and Theorem~\ref{thm:ApproximatingInteractionPictureTimeEvolWithTruncatedDiscretizedDysonSeries}.
} 

\subsection{Quantum algorithm for simulating time-evolution in the interaction picture}\label{subsec:QuantumAlgorithmInteractionPicture}

Our Hamiltonian consists of various types of terms, such as those that are diagonal in the basis of electric fields and fermion operators.
Furthermore, the electric field term ($H_E$) commutes with the mass term ($H_M$).
These two terms are collected in the Hamiltonian $H_0$, and the associated time evolution operators are fast-forwardable. In other words, for any time $s$, the Hamiltonian simulation
\begin{align}
    e^{-iH_0 s}=e^{-i(H_E+H_M)s}=e^{-iH_Es}e^{-iH_Ms}
\end{align}
is implementable with polylogarithmic number of gates in $\{N, 1/\epsilon, \eta, \ldots\}$. 
The quantum circuit for these fast-forwardable time evolutions are used while implementing $W_{(K,M)}$ and finally for $e^{-iHt}$. They consist of single-qubit rotations and CNOTs.
The error caused by these and other approximate single-qubit rotations, (excluding Clifford and T-gates), in total $n_{\text{rot}}$-many with an error $\tilde{\epsilon}_3= \epsilon_3/n_{\text{rot}}$, are considered altogether and give a total error at most $\epsilon_3$. 
Then the bit-precision for single qubit rotations are given by $b_{\rm{rot}}= \lceil \log_2 (1/\tilde{\epsilon}_3) \rceil= \lceil\log_2 (n_{\text{rot}}/\epsilon_3)\rceil$. 
Given that each rotation synthesis costs $O(b_{\rm{rot}})$ T-gates~\cite{kliuchnikov2023shorter}, and $b_{\rm{rot}}$ depends on $1/\epsilon_3$ only logarithmically, we typycally choose $\epsilon_3 \approx \epsilon_1/10 = \epsilon_2/10$.

We first give the following intermediate result that lays out the complexity of implementing a unitary that approximates the interaction picture time evolution $U_I(t_0/\alpha)$ for a given time segment $t_0/\alpha$.
In the proof, we construct a quantum circuit that approximates $U_I(t_0/\alpha)$, and give a detailed quantum circuit compilation and its cost in Section~\ref{subsec:CompilationAndResourceCount}. 

\begin{lemma}[Quantum circuit implementation of Dyson series $U_I(t_0/\alpha)$ for a time segment of $t_0/\alpha$]\label{lem:CircuitImplementationOfDysonSeries}
Let $D_{(K,M)}(t)$ and $\tilde{D}_{(K,M)}(t)$ be given as in Definition~\ref{def:TruncDiscDysonSeries} and Definition~\ref{def:TruncDiscDysonSeriesWithCol}, respectively.
Let $K,M$ be the truncation and discretization parameters for time $t_0/\alpha$ and error $\epsilon'$.
Let $\alpha$ be the rescaling factor that arises as a result of block-encoding $V$, and let $t_0 \in \mathbb{R}$ be such that $\beta= \sum^K_{k=0}t_0^k/k!= 2$. Finally, let $\mathcal{M}$ be the number of local terms in $H_0$.
Then, there exist quantum circuits that implement the unitary $W_{(K,M)}(t_0/\alpha, \epsilon)$ and $\tilde{W}_{(K,M)}(t_0/\alpha, \epsilon)$ such that
\begin{align}
\|W_{(K,M)}(t_0/\alpha,\epsilon) - U_I(t_0/\alpha)\| \leq \epsilon, \; \textrm{and} \; \|\tilde{W}_{(K,M)}(t_0/\alpha,\epsilon) - U_I(t_0/\alpha)\| \leq \epsilon.
\end{align}
Both circuits use $\mathcal{O}\left( \mathcal{M} \log(k\mathcal{M}/\epsilon) + K \right)$
single qubit rotations and $\tilde{\mathcal{O}}\left( K \log M \right)$ additional T-gates, where $\tilde{\mathcal{O}}$ hides additional $log$ factors in $K$.
Furthermore, they call the block encodings ($\textrm{BE}_{V/\alpha}$ and $\textrm{BE}^\dag_{V/\alpha}$) $2K$ many times.
\end{lemma}

\begin{proof}
We implement the $(K,M)$-truncated-discretized Dyson series with quantum circuits using the method of linear combination of unitaries (LCU). 
Using either $D_{(K,M)}(t_0/\alpha)$ or $\tilde{D}_{(K,M)}(t_0/\alpha)$, we have a linear combination of products of Hermitian operators, i.e., $V$, interleaved with unitary time evolutions, i.e., $e^{-iH_0 s}$ for various times $s$ that depend on the ancillae.
The rest of the proof uses only the particular approximation of the Dyson series that is given in Definition~\ref{def:TruncDiscDysonSeries}, except for one minor detail in the implementation of the bitonic sort, the proof and resource counts are identical.

\noindent The key circuit that implements an approximation to the Dyson series, conditioned on all the ancillae being in $|0\rangle$, is given in Fig.~\ref{fig:IP}.

\begin{sidewaysfigure}
    \centering
    \includegraphics[width=1\linewidth]{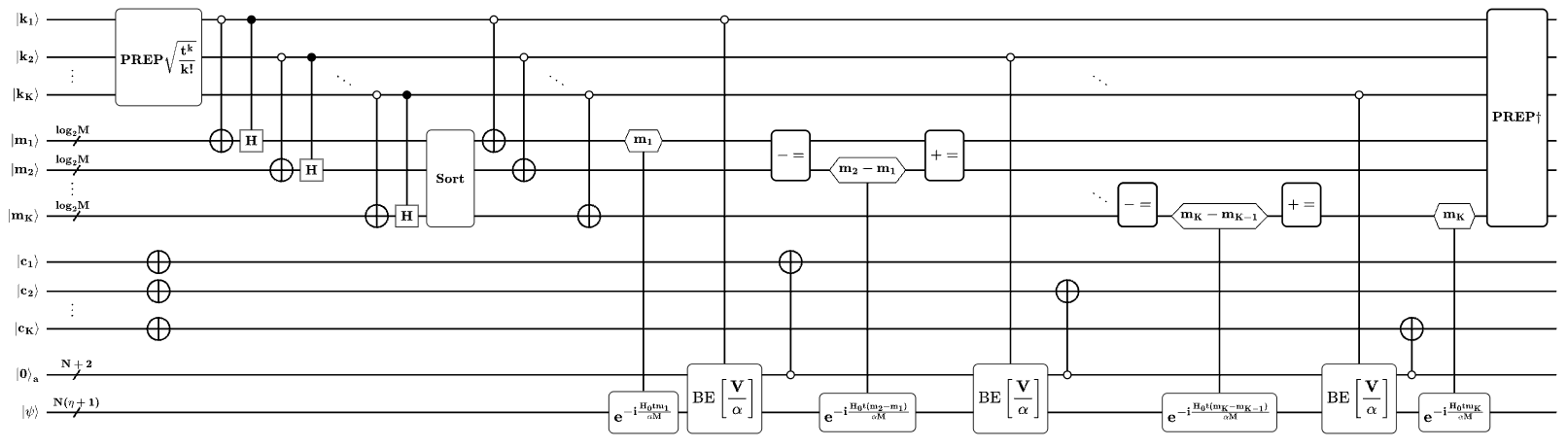}
    \caption{The circuit that implements the unitary $U$ that is given in Eq.~\eqref{eq:TheUnitaryU}}.
    \label{fig:IP}
\end{sidewaysfigure}

The cost of this circuit is: 
\begin{align}
\nonumber 2 & \left[\mathcal{C}  \left( \textrm{PREP}_{\sqrt{t_0^k/k!}}\right) + \right. \left. (K\log_2 M) \mathcal{C}(\textrm{C-Had}) + \mathcal{C}\left(\textrm{SORT}^{(K)}_{\log_2 M}\right)\right]\\ 
& \qquad + (2K-2)\mathcal{C}(\textrm{ADD}_{\log_2 M}) + K \mathcal{C}(\textrm{BE}_{V/\alpha}) + (K+1) \mathcal{C}(e^{-iH_0 s}) + K(N+1) \mathcal{C}(\text{Toffolis}).
\end{align} 
The first line is the cost of the $\PREP$ subroutine, whereas the second line is the cost of the $\SEL$ subroutine, and the construction is explained in Section~\ref{subsec:ResourceCostIP}.
Note that the prefactor $2$ in the first line comes from the uncompute (or $\PREP^\dag$), and the rest can be straightforwardly derived from the circuit in Fig.~\ref{fig:IP}. Moreover, $M$ can always be chosen as the next power of two, such that we can create a uniform state preparation with Hadamards instead of a general method such as in Ref.~\cite{babbush2018encoding}. 
Notice also that this is an upper bound, since some of the uncomputation can have a reduced cost.
Up to infidelity from rotation synthesis, this circuit results in the following unitary $U$ on any initial state $|0\rangle \otimes |\psi\rangle$:

\begin{align}\label{eq:TheUnitaryU}
U |0\rangle \otimes |\psi\rangle=  \frac{1}{\beta} |0\rangle \otimes D_{(K,M)}(t_0/\alpha)|\psi\rangle + (\ldots) |\Phi^{\perp}\rangle
\end{align}
where $(|0\rangle\langle 0| \otimes \mathds{1}) |\Phi^{\perp}\rangle= 0$, and $\beta$ (or implicitly, $t_0$) is chosen such that one needs only (and exactly) one round of amplitude amplification, i.e., $\beta\geq 2$. 
We prove this by following the circuit given in Fig.~\ref{fig:IP}. 
After applying $\textrm{PREP}_{\sqrt{t^k_0/k!}}$ and $\textrm{C-Had}^{\otimes K\log_2 M}$, we have

\begin{align}
\frac{1}{\sqrt{\beta}} \sum^K_{k=0}\sqrt{\frac{t_0^k}{k!}} \ket{1}^{\otimes k} \ket{0}^{\otimes K-k} \otimes \frac{1}{\sqrt{M^k}} \sum_{m_1, m_2, \ldots, m_k = 0}^{M-1} \ket{m_1} \ket{m_2}  \ldots \ket{m_k} \ket{0}^{K-k} \otimes \ket{\psi}.
\end{align}
After SORT, we obtain

\begin{align}\label{eq:AfterSortState}
\frac{1}{\sqrt{\beta}} \sum^K_{k=0}\sqrt{\frac{t_0^k}{M^k k!}} \ket{1}^{\otimes k} \ket{0}^{\otimes K-k} \otimes \Bigg[ \sum_{0 \leq m_{\sigma(1)} < \ldots < m_{\sigma(k)} \leq M-1} \ket{m_{\sigma(1)}}  \ldots \ket{m_{\sigma(k)}} \ket{0}^{K-k} \sum_{\sigma  \in \mathcal{S}_k | m_{\sigma(1)} < \ldots < m_{\sigma(k)}} \ket{\sigma} \\
\nonumber +\sum_{0 \leq m_1, \ldots, m_k \leq M-1 \atop  \exists i\neq j \text{ s.t }. m_j=m_i} \sum_{\sigma | m_{\sigma(1)} \leq ... \leq m_{\sigma(k)} } \ket{m_{\sigma(1)}}  \ldots \ket{m_{\sigma(k)}} \ket{0}^{K-k} \ket{\sigma} \Bigg] \otimes \ket{\psi} ,
\end{align}
where $\mathcal{S}_k$ is the permutation group of $k$ objects and $\sigma$ keeps track of the particular permutation that sorted the numbers $m_1, \ldots, m_k$ to $m_{\sigma_1}, \ldots, m_{\sigma_k}$ in increasing order. 
If one desires to implement $D_{(K,M)}$ rather then  $\tilde D_{(K,M)}$, we will then need to flag the branch without collision (the second term in eq.~(\ref{eq:AfterSortState}):
\begin{align}
\frac{1}{\sqrt{\beta}} \sum^K_{k=0}\sqrt{\frac{t_0^k}{M^k k!}} \ket{1}^{\otimes k} \ket{0}^{\otimes K-k} \otimes \Bigg[ \sum_{0 \leq m_{\sigma(1)} < \ldots < m_{\sigma(k)} \leq M-1} \ket{m_{\sigma(1)}}  \ldots \ket{m_{\sigma(k)}} \ket{0}^{K-k} \sum_{\sigma  \in \mathcal{S}_k | m_{\sigma(1)} < \ldots < m_{\sigma(k)}} \ket{\sigma} \ket{0} \\
\nonumber +\sum_{0 \leq m_1, \ldots, m_k \leq M-1 \atop  \exists i\neq j \text{ s.t }. m_j=m_i} \sum_{\sigma | m_{\sigma(1)} \leq ... \leq m_{\sigma(k)} } \ket{m_{\sigma(1)}}  \ldots \ket{m_{\sigma(k)}} \ket{0}^{K-k} \ket{\sigma} \ket{1} \Bigg] \otimes \ket{\psi} ,
\end{align}
Note that one could amplify over the flag qubit being in state $\ket{0}$. However, the number of states in the no collision branch is $\sum_{k=0}^K {M \choose k}k!$ which, for large $M$, makes the collision branch negligible. In other words, since the collisions contribute only to a slight increase of the induce one-norm, the amplification will be included in the overall OAA of the Dyson series.   

Indeed, this is expected, because Theorem~\ref{thm:ApproximatingInteractionPictureTimeEvolWithTruncatedDiscretizedDysonSeriesWithCollision} implies that $D_{(K,M)}(t)$ is $\epsilon$-close to unitary. 
After applying the sequence of $e^{-iH_0 m_it_0/\alpha M}$ and block encodings of $-i V/\alpha$, we obtain the state

\begin{align}
\nonumber & \frac{1}{\sqrt{\beta}}  \sum^K_{k=0}\sqrt{\frac{t_0^k}{M^k k!}}   \ket{1}^{\otimes k} \ket{0}^{\otimes K-k} \nonumber  \Bigg[  \sum_{0=m_1 < m_2 < \ldots < m_k \leq M-1} \ket{m_1} \ket{m_2}  \ldots \ket{m_k} \ket{0}^{K-k}\sum_{\sigma \in \mathcal{S}_k} \ket{\sigma} \ket{0} \prod_{j=1}^{k} \frac{V(m_j \Delta)}{\alpha} \ket{\psi} |0\rangle_{BE} \\
& +\sum_{0 \leq m_1, m_2, \ldots, m_k \leq M-1 \atop  \exists i\neq j \text{ s.t }. m_j=m_i} \sum_{\sigma | m_{\sigma(1)} \leq ... \leq m_{\sigma(k)} }\ket{m_{\sigma(1)}} \ket{m_{\sigma(2)}}  \ldots \ket{m_{\sigma(k)}} \ket{0}^{K-k} \ket{0}^{K-k} \ket{\sigma} \ket{1}\prod_{j=1}^{k} \frac{V(m_{\sigma(j)} \Delta)}{\alpha}\ket{\psi} |0\rangle_{BE}\Bigg] ,
\end{align}
After uncomputing the SORT, controlled Hadamards, and applying $\textrm{PREP}^\dag$, we obtain, in the $\ket{0}$ and no-collision branch of the state:
\begin{align}
& \nonumber =
\frac{1}{{\beta}} \ket{0}^{K(\log_2M+1)} \otimes \ket{0} \sum^K_{k=0}{\frac{t_0^k}{\alpha^k M^k}}\sum_{0=m_1 < m_2 < \ldots < m_k \leq M-1} \prod_{j=1}^{k} V(m_j \Delta) \ket{\psi} \ket{0}_{BE}  \\
& +\frac{1}{{\beta}} \ket{0}^{K(\log_2M+1)} \otimes \ket{1}\sum^K_{k=0}{\frac{t_0^k}{\alpha^k M^k k!}} \sum_{m_1, m_2, \ldots, m_k \leq M-1 \atop  \exists i\neq j \text{ s.t }. m_j=m_i} \mathcal{T}\left(\prod_{j=1}^{k} V(m_j \Delta)\right) \ket{\psi}\ket{0}_{BE} +  \ldots \\ 
& =  \frac{1}{{\beta}} \ket{0}^{K(\log_2M+1)+1} \ket{0}_{BE} \otimes D_{(K,M)}(t_0/\alpha)\ket{\psi} + \ldots
\end{align}

With similar arguments, and discarding the flag register for collision cases (by not creating it in the first place), we get:

\begin{align}
\nonumber = & 
\frac{1}{{\beta}} \ket{0}^{K(\log_2M+1)} \sum^K_{k=0}{\frac{t_0^k}{\alpha^k M^k}}\sum_{0=m_1 < m_2 < \ldots < m_k \leq M-1} \prod_{j=1}^{k} V(m_j \Delta) \ket{\psi}\ket{0}_{BE}  \\
& + \frac{1}{{\beta}} \ket{0}^{K(\log_2M+1)} \sum^K_{k=0}{\frac{t_0^k}{\alpha^k M^k k!}} \sum_{m_1, m_2, \ldots, m_k \leq M-1 \atop  \exists i\neq j \text{ s.t }. m_j=m_i} \mathcal{T}\left(\prod_{j=1}^{k} V(m_j \Delta)\right) \ket{\psi}\ket{0}_{BE} + \ldots\\
= & \frac{1}{{\beta}} \ket{0}^{K(\log_2M+1)}\ket{0}_{BE} \otimes \tilde D_{(K,M)}(t_0/\alpha)\ket{\psi} + \ldots
\end{align}

\noindent We further choose $K$ and $M$ such that $D_{(K,M)}(t_0/\alpha)$ approximates the Dyson series $U_{I}(t_0/\alpha)$ up to error $\epsilon/4 = (\epsilon_1 + \epsilon_2)  \lceil \frac{t \alpha}{t_0}\rceil \frac{1}{4}$, then we know that $D_{(K,M)}(t_0/\alpha)$ is a unitary up to error $\epsilon/4$. 
After a single round of robust oblivious amplitude amplification~\cite{berry2015simulating}, we ensure that we implement a unitary that is $3\epsilon/4$ close to $D_{(K,M)}(t/\alpha)$, which is also $\epsilon$ close to $U_{I}(t_0/\alpha)$ by the choice of $K$ and $M$. This requires using two of the unitary circuits $U$ and one $U^\dag$. The resources estimate reported assume that the cost of $U$ and $U^\dagger$ are identical.
The remaining unbudgeted error is spared for other components in the rest of the circuit, e.g., error from cutoff and single-qubit rotations, see Table~\ref{tab:ListOfTheErrors}.
Hence the main claim follows.

The rest of the proof is devoted to the computational resource counts parts of the circuit except the cost of the block encoding $\textrm{BE}_{V/\alpha}$, since the circuit implementation of this highly depends on the structure of $V$ and hence can vary case by case. For more details on these resource costs, see. Appendix~\ref{subsec:CompilationAndResourceCount}.
Expressing these costs in big-O notation completes the proof. 
\end{proof}

We finally put all the pieces together, to find the total cost of implementing the quantum circuit $W(t, \epsilon)$ that approximates the time-evolution operator up to an error at most $\epsilon$.

\begin{corollary}[Quantum circuit implementation of the time evolution based on the interaction picture]\label{thm:Main}
Let $H= H_0 + V$ be a Hamiltonian such that $H_0$ can be expressed as linear combination of $\mathcal{M}$ commuting $k$-local unitaries (or more generally exponentially fast-forwardable), and let $\alpha>1$ be the minimal rescaling factor of $V$ such that an efficient unitary block encoding $\textrm{BE}_{V/\alpha}$ of $V/\alpha$ is possible.
Let $t \in \mathbb{R}$.
Then, there exists a quantum circuit implementing the unitary $W(t, \epsilon')$ that approximates the time evolution $e^{-iHt}$ such that  
\begin{align}
\|W(t,\epsilon') - e^{-iHt}\| \leq \epsilon'.
\end{align}
The circuit uses
\begin{align}
\mathcal{O} \left(\alpha t \mathcal{M} \log(k\mathcal{M} \alpha t/\epsilon') + \log(\alpha t /\epsilon') \right)
\end{align}
single qubit rotations and
\begin{align}
\mathcal{O}\left(\alpha t \log M + \alpha t \log M \log(\alpha t/\epsilon') \right)
\end{align}
additional Clifford+T gates.
Furthermore, the circuit calls the block encoding
$\textrm{BE}_{V/\alpha}$
\begin{align}
\mathcal{O}(\alpha t \log(\alpha t/\epsilon'))
\end{align}
many times.
\end{corollary}

\begin{proof}
The circuit $W(t,\epsilon')$ we implement repeats the product of unitaries $W_0$ and $W_{(K,M)}(t_0/\alpha)$, $\mathsf{r}= \lceil t\alpha/t_0 \rceil$ many times.
More precisely,

\begin{align}
W(t, \epsilon')= e^{-iH_0 (t - (\mathsf{r}-1) t_0/t\alpha)} W_{(K,M)}(t - (\mathsf{r}-1) t_0/t\alpha , \epsilon'/\mathsf{r}) \left[e^{-iH_0 t_0/\alpha} W_{(K,M)}(t_0/\alpha, \epsilon'/\mathsf{r}) \right]^{\mathsf{r}-1}.
\end{align}
Note that this circuit approximates the unitary operator $e^{-iHt}$ up to an error at most $\epsilon'$, which accounts for the truncation, discretization and rotation synthesis errors.
The truncation and discretization parameters $K$ and $M$ are determined in Theorem~\ref{thm:ApproximatingInteractionPictureTimeEvolWithTruncatedDiscretizedDysonSeries} and Theorem~\ref{thm:ApproximatingInteractionPictureTimeEvolWithTruncatedDiscretizedDysonSeriesWithCollision}, both for approximating the Dyson series without or with collisions.
The complexity of $W(t,\epsilon')$ is then upper bounded by
\begin{align}\label{eq:CostEquation}
\mathsf{r} \left( \mathcal{C} (e^{-iH_0 t_0/\alpha}) + \mathcal{C} \left(W_{(K,M)}\left(\frac{t_0}{\alpha}, \frac{\epsilon'}{\mathsf{r}}\right)\right) \right).
\end{align}

Then, the result follows by combining the costs $\mathcal{C} (e^{-iH_0 t_0/\alpha})$ and $\mathcal{C} (W_{(K,M)}(t_0/\alpha, \epsilon'/\mathsf{r}))$, given in Section~\ref{subsec:CompilationAndResourceCount}.
Note that the choice of $t_0$ affects the total number times one calls the block encoding of $V$ in a few ways. 
Making the choice for $2= \beta= \sum^{K-1}_{k=0} t^k_0 /k!$, we find $t_0 = \ln 2 \approx 0.7$.
Then, each circuit $W_{(K,M)}(t_0/\alpha, \epsilon'/\mathsf{r})$ calls a number of $2K$ times the operator $BE_{V/\alpha}$  and  $K$ times its dagger $BE_{V/\alpha}^\dagger$.
In total the number of calls to the block encodings is bounded by
\begin{align}
\left\lceil \frac{3K \alpha t}{\ln2} \right\rceil ,
\end{align}
where $K=K\left(\frac{\ln2}{\alpha}, \frac{\ln2\epsilon'}{3t \alpha}\right)$, whose value is
\begin{align}
K= \mathcal{O}(\log(\alpha t/\epsilon')),
\end{align}
asymptotically.
\end{proof}

\subsection{Compilation and resource count for the Schwinger model}\label{subsec:CompilationAndResourceCount}

In this section, we give quantum circuit compilation details of the subroutines that appears in the implementation depicted in Fig.~\ref{fig:IP}.
In particular, these subroutines are elements that realize $W_{(K,M)}$, and the resource cost analysis is used for calculating the second term $\mathcal{C}(W_{(K,M)})$ in Eq.\eqref{eq:CostEquation}.
Note that the first term in Eq.\eqref{eq:CostEquation} is subdominant and just a special case of implementing $e^{-i(H_E + H_M)(\cdot)}$ for time $t_0/\alpha$ whose cost can be found below.

\paragraph{\bf{$\PREP$.}} Our implementation uses a first register that encodes a weighted history state of the form:
\begin{align}\label{eq:A1}
\PREP_{\sqrt{t_0/k!}} \ket{00\ldots0}=    \frac{1}{\sqrt{\beta}}\left(\left|00...0\right>+\sqrt{t_0}\left|10...0\right>+{\frac{t_0}{\sqrt{2}}}\left|110...0\right>+...+\sqrt{\frac{t^K_0}{K!}}\left|11...1\right>\right),
\end{align}
where $\beta$ is the normalization factor:
\begin{align}
    \beta = \sum_{k=0}^K \frac{t^k_0}{k!}.
\end{align}
As mentioned previously, in order to have exactly one step of AA, we choose $t_0$ such that $\beta$ is arbitrarily close to $2$.
This step uses one rotation and $K-1$ controlled rotations which can each be implemented with Clifford gates and 2 rotations, i.e., $2K-1$ rotations of arbitrary angles. An example for $K=6$ can be found in Fig.~\ref{fig:prep_k_hot} .

\begin{figure}
    \centering
    \includegraphics[width=0.2\linewidth]{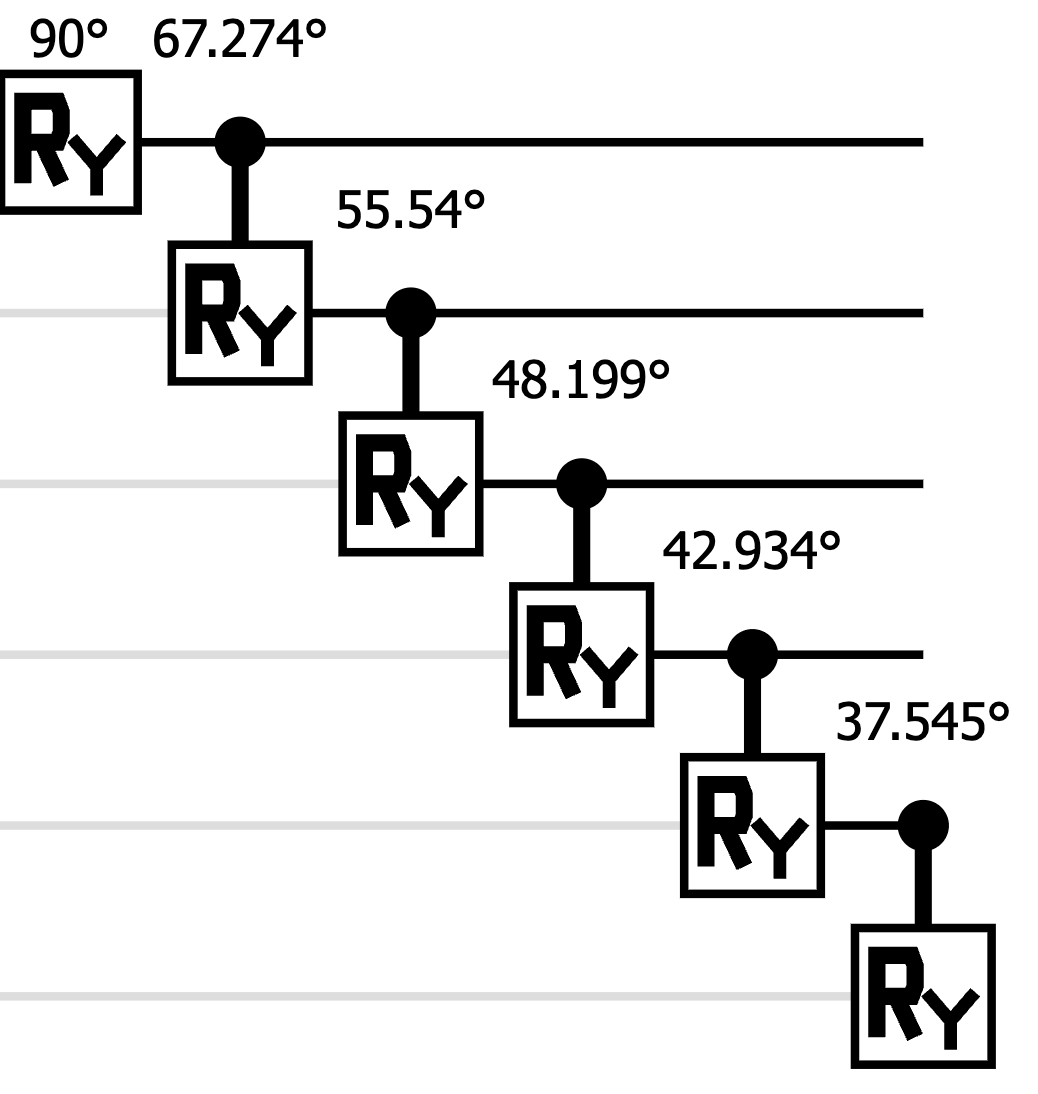}
    \caption{Implementation of PREP$_{\sqrt{t_0^k/k!}}$ for $K = 6$ and $t_0 = \ln2$}
    \label{fig:prep_k_hot}
\end{figure}

Then, we need $K$ registers of size $\lceil \log_2M \rceil$ that encodes the time steps.
We choose $M$ to be an integer power of $2$, for the sake of ease in compilation/resource cost.
Controlled on the $k$th qubit in the first register given in Eq.~\eqref{eq:A1}, we create an equal weight superposition $\frac{1}{\sqrt{M}} \sum^{M-1}_{m_k=0} \ket{m_k}$ state in the $k$th register of this set of time parameters.
This can be done with $K\log_2M$ controlled-Hadamard gates.
Each can be implemented with $2$ T-gates, hence it has a total cost of $2K\log_2M$ T-gates. 
Controlled on the $k$th qubit being in $\ket{0}$ state, the $k$th time register is set to $\ket{M-1}$, performed via C-NOTs this operation only requires Cliffords.
These last $K-k$ time registers will be set back to zero after the sort, importantly before calling $\BE_V$. 
The result at this stage is:

\begin{align}
   \propto \sum_{k=0}^{K} c_k \ket{1}^{\otimes k}\ket{0}^{\otimes K - k} \sum^{M-1}_{m_1, \ldots, m_k=0} \ket{m_1} \ldots \ket{m_k} \ket{M-1}  \ldots \ket{M-1}.
\end{align}

Then, the time registers are sorted in increasing order, i.e., they are put in order $m_1 \leq m_2 \leq \ldots \leq m_k$.
This is achieved by a bitonic sort, as in Ref.~\cite{batcher1968sorting}.
The number of comparators in the bitonic sort is upper bounded by
\begin{align}
\sum^{\lceil \log_2 K \rceil}_{k=1} k \lfloor K/2 \rfloor = \frac{\lceil \log_2 K \rceil + 1}{2} \left\lfloor \frac{K}{2} \right\rfloor \left\lceil \log_2K\right\rceil .
\end{align}
Each comparator uses $\log_2 M$ Toffolis and an additional ancilla that stores the result of the comparator.
The uncomputation of each Toffoli can be performed via phase-fixup measurement-based uncomputation similar to what is suggested in Ref.~\cite{gidney2018halving}.
Additionally, $\log_2M$ temporary ancillae are shared across all comparators.
Controlled on the comparator's result, the two $\log_2 M$ qubit-size registers are swapped, with an additional $\log_2M$ Toffolis.
Hence, this stage costs
\begin{align}
\log_2 M (\lceil \log_2 K \rceil + 1) \left\lfloor \frac{K}{2} \right\rfloor \left\lceil \log_2K\right\rceil
\end{align}
Toffolis and
\begin{align}
\frac{\lceil \log_2 K \rceil + 1}{2} \left\lfloor \frac{K}{2} \right\rfloor\left\lceil \log_2K\right\rceil  + \log_2 M
\end{align}
ancillae.

If we choose to not include the collisions, additional ancillae will be added to flag out the collisions. 
We do not need to apply AA since the probability of success of having zero collision is almost $1$ for large $M$. 
To flag the no-collision branch, it requires $K-1$ controlled-comparators, i.e., $(K-1)\log_2M$ Toffolis, and $K-1$ additional ancillas to record the results of these additional comparators. 
Finally, we AND the results of the individual comparators into a single flag qubit, which requires $(K-2)$ Toffolis and $K-2$ reusable/temporary additional ancilla.

This is the cost of one application of $\PREP$.
In the quantum circuit $U$, given in Fig.~\ref{fig:IP} and acting as in Eq.~\eqref{eq:TheUnitaryU}, $\PREP$ and $\PREP^\dagger$ are each called once. To construct $W_{(K,M)}(t_0/\alpha, \epsilon'/\mathsf{r})$, $U$ itself is called twice and $U^\dagger$ is called once.
Naively, the cost a subroutine $S$ has the same cost as $S^\dagger$ (= uncomputation of $S$), and hence Table~\ref{tab:ResourceCostIP} refers to the number of calls to the subroutine $S$ as the total number of calls to $S$ and $S^\dagger$.
However, one can implement the uncomputations of certain subroutines more efficiently.
For instance, in bitonic sort, the comparators can be uncomputed via phase fixup measurement-based uncomputation, which has been studied in detail in Ref.~\cite{luongo2024measurement}, which costs on average half of the Toffolis of the original subroutine itself.

\paragraph{\bf{SELECT.}}
The select operators consist of the time evolution operator for the mass and electric terms and a block encoding of the interaction terms.
The block encoding is controlled on the history state register and the time evolution only needs to be controlled on the time steps registers. 
Note that since the operators are of the form $e^{-it_0m_j/(\alpha M)}$BE$_{V/\alpha} e^{it_0 m_j/(\alpha M)}$, we will compute the difference between $m_j$ and $m_{j+1}$ and reduce by close to half the number of rotations required for the mass and electric terms. This costs $\log_2M - 1$ Toffolis and $\log_2M - 1$ reusable ancillas per addition or subtraction. In total, there are $2(K-1)$ additions and subtractions. 
Moreover, in order to multiply the block-encodings with minimal additional ancillae, a compression gadget is used.
Specifically, we use the unary ($k$-hot to be specific) version of the compression gadget given in Ref.~\cite{fang2023time}.
See Fig.~\ref{fig:IP}. 
The compression gadget requires $K$ many $(N+2)$-controlled Toffolis, where the controls are on the block-encoding ancillae (consisting of $N+2$ qubits). 
This results in a total of $K(N+1)$ Toffolis and $N+1$ reusable ancillae in addition to the $K$ ancillae used as the ``counter" of the compression gadget.
Hence, in addition to the subroutines we study below, $K(N+1) + 2(K-1)(\lceil \log_2M \rceil-1)$ additional Toffolis and $\max(\lceil\log_2M\rceil - 1, N+1)$ temporary ancillae are used.

The subroutine $\BE_{V/\alpha}$ is implemented as a standard $\PREP_{\BE}^\dagger-\SEL_{\BE}-\PREP_{\BE}$ method.
For the purpose of efficient implementation, we put the control on the $\PREP_W$ (and not on $\SEL_\BE$), due to the use of the unary encoding in the $\PREP_W$ registers.
This register creates a $W$-state that indexes the link $r$ on which $V$ acts.
This takes $N-1$ qubits.
For the number of links ($N-1$) that is an integer power of $2$, the quantum circuit for the controlled-$\PREP_W$ is given in Fig.~\ref{fig:c-prep_w}. 
This circuit costs $N-2$ c-Had gates, which can be implemented with $2$ T-gates each.
There is an additional $\PREP'$ which creates an equal superposition over $3$ qubits for block-encoding the interaction term, consisting of $8$ terms, for a fixed link $r$.
The complete subroutine $\PREP_{\BE}= \PREP_W \otimes \PREP'$ costs $2 (N-2)$ many T-gates.
$\SEL_{\BE}$ consists of gates that select between the linear combination of unitaries given in Eq.~\eqref{eq:LCUHint}.
We implement the circuit given in Ref.~\cite{rajput2022Hybridized} (see Section~5.2, Fig.~4), with a slight modification: instead of controlling $\PREP'$, we control the phase gates, i.e., $S$ and $CZ$ gates acting on the last qubit. 
The subroutines $U$ are incrementers of size $\eta$ and the subroutines $Q$ are NOT gates, and we have an additional control-$S$ and a $CCZ$ gate.
Then, $\SEL_{\BE}$ costs $(N-1)(\eta-1) + 1$ Toffolis, $4N + 3$ T-gates, and $\eta - 1$ temporary ancillae.
Hence, in total, $\BE_{V/\alpha}$ costs $8N + 4(N-1)(\eta-1) - 1$ T-gates, derived as twice the cost of $\PREP_\BE$ and once the cost of $\SEL$, and additional $\eta - 1$ ancillae.\\

\begin{figure}
    \centering
    \includegraphics[width=0.4\linewidth]{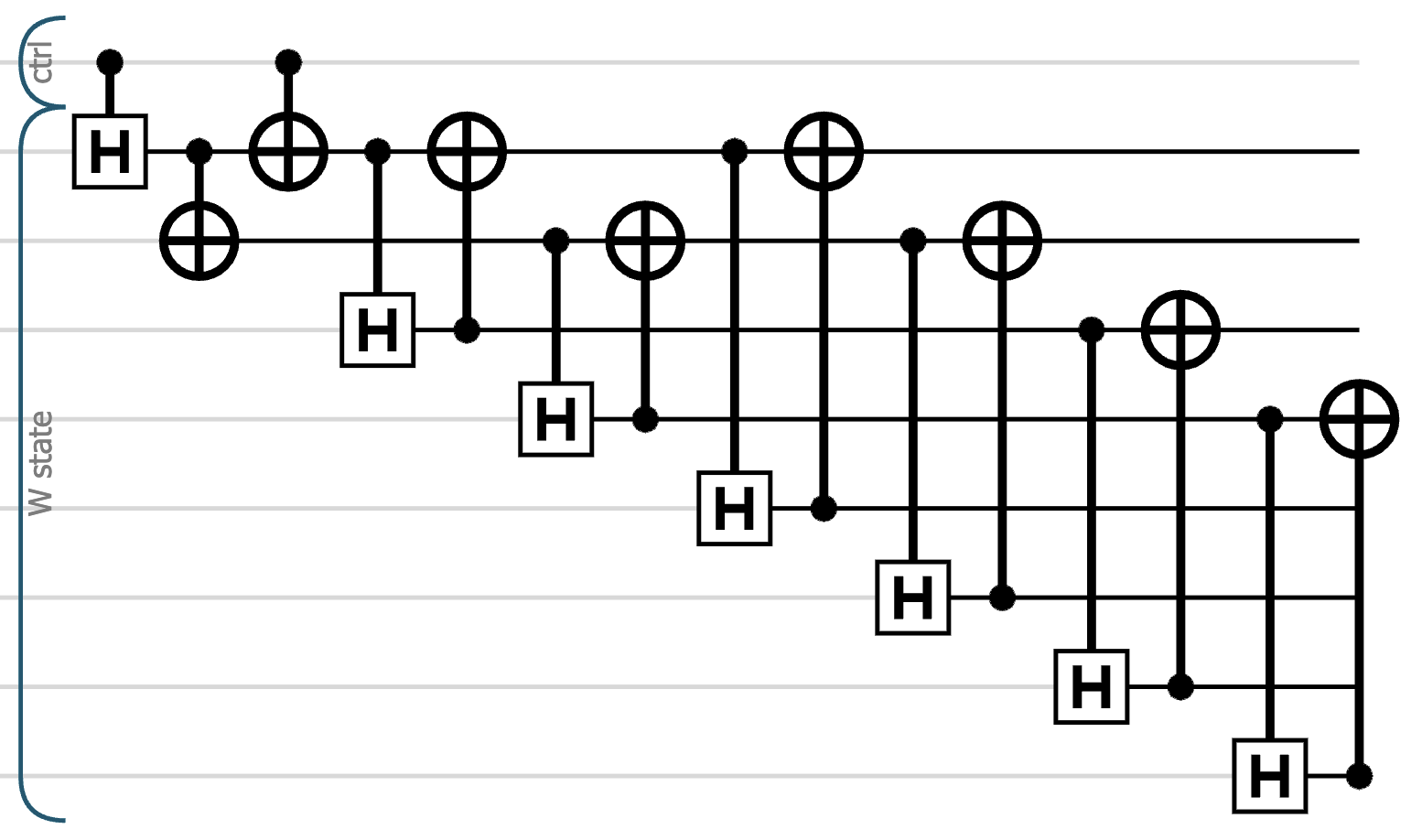}
    \caption{Controlled-$\PREP_W$ when the number of links, i.e., $N-1$, is a power of two}
    \label{fig:c-prep_w}
\end{figure}

We implement the fast-forwardable $e^{-i(H_E + H_M)s}$ for a given $s$, by implementing $e^{-i(H_M)s}$ and $e^{-i(H_E)s}$ in parallel, given that they act on different part of the system.
Note that the duration of the time evolution $s$ is stored in a quantum register, hence both $e^{-i(H_M)s}$ and $e^{-i(H_E)s}$ are implemented controlled on the value of that register.
The mass term as described in Eq.~\eqref{eq:JW_HM} can be re-write as follows:
\begin{align}
    e^{- i H_M s} = e^{-i s \mu / 2 \sum_{r=1}^{N}(-1)^r\mathds{1}}e^{i s \mu / 2 \sum_{r=1}^{N}(-1)^rZ}.
\end{align}
Since the block-encoding is conjugated by the forward and backward time evolution operator, the terms proportional to the identity cancel, leaving us to implement only $e^{i s \mu / 2 \sum_{r=1}^{N}(-1)^rZ}$. 
To do so, we first compute the number of fermions, $n_f$.
This can simply be done by first applying $\otimes_{r: \textrm{odd}}X_r$ on positron registers, computing the Hamming weight of the fermion registers~\cite{gidney2018halving}, and then applying $\otimes_{r: \textrm{odd}}X_r$ on positron registers. The cost of computing $n_f$ is $N - \omega(N)$ Toffolis, where $\omega(x)$ is the Hamming weight of $x$. 

We then make use of phase gradient addition~\cite{sanders2020compilation, nam2019low} in order to realize $e^{-i m_k \mu n_f t_0/ (2 \alpha M)}$ coherently depending on $m_k$ and $n_f$.
We present two different strategies to do so. The first one, which we will call the PGA approach, effectively multiplies $m_k$ and $n_f$ in the phase by performing $\log_2 M$ phase gradient additions with the catalyst state :
\begin{align}
    \ket{\frac{\mu t_0}{2 \alpha}} = \frac{1}{\sqrt{M}} \sum_{k=0}^{M-1} e^{- i \pi k \mu t_0/(\alpha M)} \ket{k}.
\end{align}

\noindent The controls are on the discretized-time registers and the targets are on the resulting Hamming weight of the fermions.
The phase gradient addition is called $\log_2 M$ times, at a Toffolis cost of $\lfloor \log_2 (N) \rfloor + 1$ each time, and a single qubit rotations each time as well. 
In total, considering the uncomputation of the temporary registers,  the PGA approach has the following resource requirements: the number of Toffolis is at most $N - 1 + \lceil \log_2 M \rceil(\lfloor \log_2 (N) \rfloor + 1)$, the number of rotations is $\lceil \log_2 M \rceil$, and the number of temporary ancillae is $N + \lfloor \log_2 (N) \rfloor + 1$.

The second approach, which we call the Mult approach, performs the multiplication of $m_k$ and $n_f$ in the computational basis, and then performs one big phase gradient addition with the catalyst state:
\begin{align}
    \ket{\frac{\mu t_0}{2 \alpha}} = \frac{1}{\sqrt{M + N}} \sum_{k=0}^{N + M -1} e^{- i \pi k \mu t_0/(\alpha M)} \ket{k}.
\end{align}

\noindent The cost of multiplying a register of size $\log_2M$ with the register of size $\lfloor\log_2 N\rfloor + 1$ containing $n_f$ is $\log_2M (\lfloor\log_2 N\rfloor + 1) + 2 (\log_2M + \lfloor\log_2 N\rfloor + 1)$ Toffolis, and $\log_2M + \lfloor\log_2 N\rfloor + 1 + \max(\log_2M,\lfloor\log_2 N\rfloor + 1)$ ancillas~\cite{litinski2024quantum}.
The cost of the phase gradient addition is $\log_2M + \lfloor\log_2 N\rfloor + 1$ Toffolis, $1$ rotations and $\log_2M + \lfloor\log_2 N\rfloor + 1$ ancillae. 
In total, considering the uncomputation of the temporary registers, the Mult approach has the resource cost as follows: $N + 2\log_2M \lfloor\log_2 N\rfloor + 7\log_2M + 5\lfloor\log_2 N\rfloor + 4$ Toffolis, $1$ rotation and $N+ 2\log_2M + 2\lfloor\log_2 N\rfloor + 1$ ancillae.

Similarly, we can implement the evolution with the electric term in two ways.
The first approach, which we call PGA again, follows the implementation in Ref.~\cite{shaw2020quantum} (see Fig.~4 in Section~3.1), modified with the use of Hamming weight phasing to implement the parallel rotations.
This makes use of another catalyst state:
\begin{align}
    \ket{\frac{t_0}{\alpha}} = \frac{1}{\sqrt{\gamma}} \sum_{k=0}^{\gamma - 1} e^{- 2 i \pi k t_0/(\alpha \gamma)} \ket{k},
\end{align}
where $\gamma = 2^{\log_2M + \eta - 2} = 2 M \Lambda/ 4$.
There are in total $N-1$ links each of $\eta (= \log_2 (2 \Lambda))$ qubits. 
This means, that we repeat the phase gradient addition, in a similar fashion to the mass term, for the $N-1$ links and the $\eta$ many parallel rotations. 
These $\eta$ rotation ``towers'' are of decreasing size, from $j=\eta$ up to $j= 1$.
Excluding the final tower ($j=1$) consisting of only a single rotation, we perform a phase gradient addition of size $j$. 
Hence, the cost of the electric term is

\begin{align}
    (N-1)\lceil \log_2M \rceil \sum_{j=2}^{\eta} j = (N-1) \lceil \log_2M \rceil (\eta^2 + \eta - 2)/2
\end{align}
Toffolis, $(N-1)\lceil \log_2M \rceil \eta$ single-qubit rotations, and $\eta$ temporary ancillae.

The second approach, which we call Mult again, performs a square in the computational basis of the electric field, and then multiplies the result by the time register. 
Finally, a phase-gradient addition is performed.
The cost of squaring is $\eta(\eta - 1)$ Toffolis~\cite{su2021fault} and $\eta$ temporary ancillae and $2\eta$ qubits to store the result. 
We add the $(N-1)$-many (for each link) squared electric field values, to find $\sum_{r} E^2_r$.
This takes 
\begin{align}
\sum_{i=0} ^{\log_2N} (2\eta+i)(N-1)/(2^{i+1}) &= \frac{((-1 + N) (-2\eta \ln(4) + N \ln(4) + 2\eta N \ln(16) - \ln(4 N)))}{(N \ln(4))} \\
& \leq 4 \eta N
\end{align}
Toffolis.
The final register size is $2\eta + \lceil \log_2 N \rceil$.
Then, we multiply this with the time register of size $\log_2 M $.
This costs $\log_2 M (2\eta + \lceil \log_2 N \rceil) + 2 (\log_2M+2\eta + \lceil \log_2 N \rceil)$ Toffolis and results in a register of size $2\eta + \lceil \log_2 N \rceil + \log_2 M$ and $\max(\log_2M, 2\eta + \lceil \log_2 N \rceil)$ temporary ancillae. Finally, the phase gradient addition requires $2\eta + \lceil \log_2 N \rceil + \log_2 M$ Toffolis, $1$ rotation and $2\eta + \lceil \log_2 N \rceil + \log_2 M$ ancillae.
Hence, in total, considering the uncomputation of the temporary registers, the Mult approach has the resource cost as follows:
$N[4\eta^2 + 4\eta ]+ \log_2 M [4 \eta + 5 + 2 \lceil \log_2 N \rceil] + 
5 \lceil \log_2 N \rceil - 2\eta^2 + 12\eta
$ Toffolis, $1$ rotation, and $8 \eta  + 3\lceil \log_2N \rceil + 2\log_2M$ ancilla.

In the implementation of {\bf{SELECT}} of the Dyson series, the block-encoding of the interaction term is called $K$ times, and the fast-forwardable term $e^{-i(H_E + H_M)s}$ is called $K + 1$ times. {\bf{SELECT}} and {\bf{SELECT$^\dagger$}} are called in total 3 times in oblivious amplitude amplification. 
Hence, $\BE_{V/\alpha}$ is called $3K$ times, and  $e^{-iH_E (\cdot)}$ and $e^{-iH_M (\cdot)}$ are each called $3(K+1)$ times.

\paragraph{\bf{REFLECTION.}} In order to perform OAA, we must construct a reflection around the starting state and the success state.
\begin{align}
    W_{(K,M)}\left( \frac{t_0}{\alpha}, \frac{\epsilon'}{\mathsf{r}}\right) &= U R_0 U^\dagger R_0U 
\end{align}
where $U$ is given in Eq.~\eqref{eq:TheUnitaryU} and the cost is described by the subroutine above. 
The reflection $R_0 = \mathds{1} - 2P_0$ add a phase when all the ancillae (for the truncation, discretization and the counter registers) are in the $0$ state. 
This is perform by applying a multi-controlled-$Z$. 
Hence, the cost of reflection is $K(2+\lceil \log_2M \rceil) - 1$ Toffolis and  $K(2+\log_2M) - 1$ temporary ancillae, and it is called twice.

This concludes the resource estimation for the subroutines to implement the time-evolution of the Schwinger model in the interaction picture. 

\section{Trotter implementation and resource estimation}
\label{app:trotter}

In this section, we report our implementation for the Trotter-based simulation of the Schwinger model, and the associated resource estimation. We combine elements from~\cite{shaw2020quantum} and~\cite{kan2022simulating} along with our own optimization.

First, we consider the interaction term, which, upon Jordan-Wigner transformation, can be written as~\cite{kan2022simulating}
\begin{gather}
    H_{I} = x \sum_{r} (\sigma^+_{r} U_r \sigma^-_{r+1} + h.c.),
\end{gather}
where $\sigma^+ = \ketbra{1}{0}$ and $\sigma^- = \ketbra{0}{1}$ are the Pauli raising and lowering operators, respectively. 
Note that this decomposition does not lead to a linear combination of Pauli operators, which is more amenable to a block-encoding approach as given in Eq.~\eqref{eq:LCUHint}.
However, this decomposition is favorable for a Trotter implementation, as it halves the number of terms that arise from $H_I$~\cite{kan2022simulating}, compared to the scenario where we decompose into $H_I$ Pauli operators that are evolved individually. Next, we decompose $U$ as follows~\cite{kan2022simulating}:
\begin{equation}
    U_r = \sigma^+_{0} + U_r \sigma^+_{0} U^\dag_r,
\end{equation}
where $\sigma^+_0$ is the raising operator acting on the least significant bit of the bosonic register. Intuitively, this decomposition is correct because the first term increments even numbers by flipping the least significant bit from 0 to 1; the second term decreases odd numbers by one with $U^\dag$, increases the now even number by one with $\sigma^+_0$, and then adds one again by $U$. Inserting this into the interaction term, we get
\begin{equation}
    H_{I} = x \sum_{r} [\sigma^+_{r} \sigma^+_{0} \sigma^-_{r+1} +\sigma^+_{r} (U\sigma^+_{0} U^\dag) \sigma^-_{r+1} + h.c.].
\end{equation}
This decomposition allows us to decompose $H_I$ into two terms in the Trotter scheme, which we describe below. First, we recall that the simulation time is $t$ and $\eta = \lceil \log_2(2\Lambda)\rceil$ is the size of a bosonic register with a cutoff $\Lambda$. We write the second-order Trotter formula as
\begin{equation}
    S(t) = \left( \prod_{j=1}^l e^{-i H_j t/2\mathsf{r}} \prod_{j=l}^1 e^{-i H_j t/2\mathsf{r}} \right)^{\mathsf{r}},
\end{equation}
where $H = \sum_{j=1}^l H_j$ is an ordered decomposition of the Hamiltonian into $l$ terms.

We decompose the Hamiltonian into $l=6$ terms, where
\begin{gather}
    H_1 = H_E, \: H_2 = H_M, \\
    H_3= H_{1,e} = x \sum_{\text{even } r} [\sigma^+_{r} \sigma^+_{0} \sigma^-_{r+1}  + h.c.], \\
    H_4= H_{2,e} = x \sum_{\text{even } r} [\sigma^+_{r} (U\sigma^+_{,0} U^\dag) \sigma^-_{r+1} + h.c.], \\
    H_5= H_{1,o} = x \sum_{\text{odd } r} [\sigma^+_{r} \sigma^+_{0} \sigma^-_{r+1}  + h.c.], \\
    H_6= H_{2,o} = x \sum_{\text{odd } r} [\sigma^+_{r} (U\sigma^+_{0} U^\dag) \sigma^-_{r+1} + h.c.].
\end{gather}
For $R$ Trotter steps, $H_1$ and $H_2$ are evolved $\mathsf{r}+1$ times, $H_{3-5}$ are evolved $2\mathsf{r}$ times, and $H_6$ is evolved $\mathsf{r}$ times because contiguous evolutions of the same Hamiltonian can be combined into a single evolution. 

The circuit implementation of the Trotter formula $W(t, \epsilon_{\rm{rot}})$ implements an $\epsilon_{\rm{rot}}$ approximation of $S(t)$, i.e., 
\begin{align}
\|W(t, \epsilon_{\rm{rot}}) - S(t)\| \leq \epsilon_{\rm{rot}},
\end{align}
where each rotation is synthesized with accuracy $\epsilon_{\rm{rot}}/n_{\rm{rot}}$.
The Trotter error is given by
\begin{equation}
    \|e^{-iHt} - S(t)\| \leq  \epsilon_t = \frac{t^3 \rho}{\mathsf{r}^2},
\end{equation}
where $\rho$ is the Trotter commutator error bound~\cite{childs2021theory}.
The above adds up to a total error of at most $\epsilon= \epsilon_{rot} + \epsilon_t$.

We simply apply the bounds for U(1) lattice gauge theories at arbitrary spatial dimensions and with periodic boundary conditions from~\cite{kan2022simulating} to the Schwinger model with open boundary conditions (by removing irrelevant terms) to arrive at
\begin{align}
    \rho^{\textrm{nat}} &= \frac{1}{12} \left( {4Nm^2\sqrt{x} g} + \frac{Ng^4(4\Lambda^2-1)}{4\sqrt{x}g} + {10(N-1)(\sqrt{x}g)^3}  \right) \nonumber \\
    & + \frac{1}{24} \left(\frac{mg^2N(2\Lambda - 1)}{2} + {8mNxg^2} + {2Ng^3\sqrt{x}(2\Lambda + 1)} + {9(N-1)(\sqrt{x}g)^3} \right).
\end{align}

By removing the units, we equivalently have
\begin{align}
    \rho(x, \mu) &= \frac{1}{12} \left( {8Nx\mu^2} + {2N{x}(4\Lambda^2-1)} + {80(N-1)x^3}  \right) \nonumber \\
    & + \frac{1}{24} \left({2x\mu N(2\Lambda - 1)} + {32Nx^2\mu} + {16Nx^2(2\Lambda + 1)} + {72(N-1)x^3} \right),
\end{align}
when written in terms of only $x$ and $\mu$.
Asymptotically, this is given as
\begin{align}\label{eq:rho_asymptotic}
    \rho(x, \mu) &= Nx \mathcal{O}\left(\mu^2 + \Lambda^2 + x + \mu\Lambda + x \mu + x \Lambda \right) ,
\end{align}
where $N, m, x$, and $g$ are the number of lattice sites, base fermionic mass, lattice spacing, and bare gauge coupling. Then, the number of Trotter steps $\mathsf{r}$ is given by $\mathsf{r} = \lceil t^{3/2}\rho^{1/2} / \epsilon_t^{1/2} \rceil$. Now we proceed to describe our circuit implementation of the evolutions of $H_{1-6}$. 

\paragraph{\bf{ELECTRIC TERM.}}
We begin with $H_1 = H_E$, the electric term. Following~\cite{shaw2020quantum}, by noting that
\begin{equation}
    E^2 = \left(E+\frac{I}{2}\right)^2 - \left(E+\frac{I}{2}\right) + \frac{I}{4},
\end{equation}
$e^{-i t H_E}$ can be implemented, up to a global phase, as a product of $e^{-i(E+I/2)^2t}$, where
\begin{equation}
    (E+I/2)^2 = \frac{1}{12}(4^\eta-1)+\sum_{j=0}^{\eta-2}\sum_{k>j}^{\eta-1}2^{j+k-1}Z_j Z_k.
\end{equation}
As a result, each $e^{-i(E+I/2)^2t}e^{i(E+I/2)t}$, which is equivalent to $e^{-iE^2t}$ up to a global phase, can be implemented using the circuit shown in Fig.4 of Ref.~\cite{shaw2020quantum}. Instead of synthesizing the $R_z$ gates one by one, as done in~\cite{shaw2020quantum}, we implement them one layer at a time using a phase catalysis circuit in Fig.~(168) from Ref~\cite{Wang2024optionpricingunder}. Note the last layer is a single rotation gate and is implemented without using the phase catalysis circuit. Briefly, assuming access to a reusable catalyst state of the form $\otimes_{i=0}^{n-1}[R_z(2^i \theta)\ket{+}]$, this circuit applies $\otimes_{i=0}^{n-1}[R_z(2^i \theta)]$ using $n$ Toffoli gates and one $R_z$ gate. We synthesize the $R_z$ gates using the mixed fallback method in~\cite{kliuchnikov2023shorter}, the costs of which are amortized over the number of times the catalyst is used which scales linearly in the number of Trotter steps here.
The cost of this implementation is $(\mathsf{r}+1)(N-1)(\eta-1)(2\eta + 4)$ T-gates, $(2\eta-3)+(\mathsf{r}+1)(N-1)\eta$ $R_z$ gates, and $3\eta-3$ ancilla qubits, out of which $2\eta-3$ qubits are occupied by the catalyst state; we have used the fact that each Toffoli can be synthesized using four T-gates~\cite{PhysRevA.87.022328}, and that $\eta$ qubits are used (and reused whenever possible) for the applications of the phase catalysis circuit.

\paragraph{\bf{MASS TERM.}}
Note $e^{-iH_M t}$ can be implemented as $R_z$ gates, one per site, with angles of the same magnitude but of opposing signs depending on the parity of the site number. We conjugate every negative-angle $R_z$ gate with a pair of NOT gates, transforming $e^{-iH_M t}$ into a layer of $N$ same-angle $R_z$ gates, which we implement using catalyzed Hamming-weight phasing~\cite{gidney2018halving,kan2024resourceoptimized}: (i) compute the Hamming weight of the to-be-rotated, $N$-qubit register; (ii) use the phase catalysis circuit in Fig.~(168) from Ref~\cite{Wang2024optionpricingunder} to implement a phase kick-back operation; (iii) uncompute (i). Catalyzed Hamming weight phasing costs $N - w(N) + \lfloor \log_2(N) \rfloor + 1$ Toffoli gates and clean ancilla qubits, where $w(N)$ is the Hamming weight of $N$, and 1 $R_z$ gate per application; it further requires $\lfloor \log_2(N) \rfloor + 1$ ancilla qubits to store a catalyst state of the form $\otimes_{i=0}^{\lfloor \log_2(N)\rfloor+1}[R_z(2^i \theta) \ket{+}]$.

\paragraph{\bf{INTERACTION TERMS.}}
Finally we consider the interaction terms $H_{3-6}$. Each term consists of a layer of $N/2$~\footnote{We assume $N$ is even for brevity, but our analysis can be simply adapted to the case of odd $N$ by appropriate use of ceiling and floor functions.} unitaries either of the form $e^{i\theta (\sigma^+ \sigma^+ \sigma^- + h.c.)}$ or that but conjugated by a pair of adders. $e^{i\theta (\sigma^+ \sigma^+ \sigma^- + h.c.)}$ can be constructed as~\cite{Wang2021resourceoptimized,kan2022simulating}
\begin{equation}
    \includegraphics[height=2.5cm]{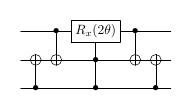}
\end{equation}
and the doubly-controlled $R_x(2\theta)$ can be compiled into
\begin{equation}
    \includegraphics[height=2.5cm]{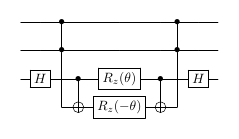}
\end{equation}
where the third bit from the top is the target. The circuit is largely taken from~\cite{Wang2021resourceoptimized}, except we uncompute the Toffoli using a Toffoli-free, measurement-feedforward circuit. Once we negate the negative angles via NOT gates, we get a layer of $N$ same-angle $R_z$ gates, which we implement using catalyzed Hamming weight phasing~\cite{kan2024resourceoptimized}, as in the mass term. Catalyzed Hamming weight phasing costs $N - w(N) + \lfloor \log_2(N) \rfloor + 1$ Toffoli gates and 1 $R_z$ gate per application. 

In summary, evolving $H_3$ or $H_{5}$ costs $4(3N/2 - w(N) + \lfloor \log_2(N) \rfloor + 1)$ T-gates and 1 $R_z$ gate, and evolving $H_4$ or $H_6$ costs an extra $4(2\eta - 2)$ T-gates for the adders. Furthermore, $H_3$ or $H_5$ require each $3N/2-w(N)+\lfloor \log_2(N) \rfloor + 1$ clean ancilla qubits and $\lfloor \log_2(N) \rfloor + 1$ to store the catalyst state. For $H_4$ or $H_6$, the clean ancilla qubit count is $\max \{3N/2-w(N)+\lfloor \log_2(N) \rfloor + 1, \eta-1\}$ to account for the extra adders. Note that since $H_6$ is the final term and is evolved for twice as long, its catalyst states have twice the angle compared to those of $H_{3-5}$, i.e., $\otimes_{i=0}^{\lfloor \log_2(N) \rfloor} [R_z(2^{i+1}\theta) \ket{+}]$ vs $\otimes_{i=0}^{\lfloor \log_2(N) \rfloor} [R_z(2^i \theta) \ket{+}]$. Instead of preparing two $(\lfloor \log_2(N) \rfloor + 1)$-qubit states, we can just prepare a $(\lfloor \log_2(N) \rfloor + 2)$-qubit state, i.e., $\otimes_{i=0}^{\lfloor \log_2(N) \rfloor + 1} [R_z(2^i \theta) \ket{+}]$, and use the first and last $\lfloor \log_2(N) \rfloor +1$ qubits for $H_{3-5}$ and $H_6$, respectively.

\end{document}